\documentclass[3p]{elsarticle} 

\myfooter[R]{Version: March 21, 2024.}

\usepackage{amscd}      
\usepackage{amsbsy}     
\usepackage{amsmath}    
\usepackage{amsfonts}   
\usepackage{amsthm}     
\usepackage{enumerate}  
\usepackage{fancybox}   
\usepackage{geometry}   
\usepackage{setspace}   

\usepackage[linesnumbered,ruled,longend]{algorithm2e} 

\usepackage{tikz}       
\usetikzlibrary{arrows,calc,decorations,shapes}

\usepackage{cancel}     

\usepackage{caption}    
\usepackage{fancyhdr}   
\usepackage{float}      
\usepackage{graphics}   
\usepackage{graphicx}   
\usepackage{subcaption} 

\usepackage{bbm}        
\usepackage[OT1]{fontenc}

\usepackage{bookmark}   
\usepackage{hyperref}   
\usepackage{cleveref}   

\usepackage{hhline}     
\usepackage{longtable}  
\usepackage{multirow}   
\usepackage{pdflscape}  
\usepackage{tabularx}   

\renewcommand\footnoterule{\kern2pt\hrule width \textwidth height.25pt\kern4pt}         
\numberwithin{equation}{section}	                                                    

\DeclareMathOperator{\ATM}{ATM}

\DeclareMathOperator{\RR}{RR}

\newcommand{\xva}{\text{xVA}}

\DeclareMathOperator{\CVA}{CVA}

\DeclareMathOperator{\DVA}{DVA}
\DeclareMathOperator{\BCVA}{BCVA}

\DeclareMathOperator{\PD}{PD}
\DeclareMathOperator{\PFE}{PFE}
\DeclareMathOperator{\PFL}{PFL}

\DeclareMathOperator{\EPE}{EPE}
\DeclareMathOperator{\ENE}{ENE}

\DeclareMathOperator{\bs}{BlackScholes}

\newcommand{\NMix}{N}

\newcommand{\NExp}{N^{\text{exp}}}
\newcommand{\NTen}{N^{\text{ten}}}
\newcommand{\NStrike}{N^{\text{strike}}}
\newcommand{\NCot}{N^{\text{cot}}}

\newcommand{\TCot}{T^{\text{cot}}}

\newcommand{\TExpCot}{T^{\text{exp,cot}}}
\newcommand{\TTenCot}{T^{\text{ten,cot}}}
\newcommand{\strikeATM}{\strike_{\ATM}}

\newcommand{\TExpCalib}{\mathbb{T}^{\text{exp}}}
\newcommand{\TTenCalib}{\mathbb{T}^{\text{ten}}}
\newcommand{\strikeCalib}{\mathbb{\strike}}

\newcommand{\half}{\frac{1}{2}}
\renewcommand{\d}{{\rm d}}
\newcommand{\e}{{\rm e}}                
\newcommand{\E}{\mathbb{ E}}            
\newcommand{\F}{\mathcal{F}}            
\newcommand{\I}{\mathbbm{1}}            

\newcommand{\N}{\mathcal{N}}            
\newcommand{\Nmeas}{\mathbb{ N}}        
\newcommand{\M}{\mathbb{ M}}            
\newcommand{\Q}{\mathbb{ Q}}            

\newcommand{\R}{\mathbb{ R}}            
\newcommand{\U}{\mathcal{U}}            

\def\ds{{\d}s}

\def\dt{{\d}t}
\def\du{{\d}u}
\def\dv{{\d}v}
\def\dW{{\d}W}
\def\dx{{\d}x}

\newcommand{\pderiv}[2]{\frac{\partial#1}{\partial #2}}
\newcommand{\ppderiv}[2]{\frac{\partial^2 #1}{\partial {#2}^2}}

\newcommand{\deriv}[2]{\frac{\d\,#1}{\d\,#2}}

\newcommand{\maxFun}[2]{\max \left( #1, #2 \right)}
\newcommand{\maxOperator}[1]{\left( #1 \right)^+}

\newcommand{\minOperator}[1]{\left( #1 \right)^-}
\newcommand{\indicator}[1]{\I_{\left\{#1\right\}}}

\newcommand{\expPower}[1]{\e^{#1}}
\newcommand{\expBrace}[1]{\exp{\left\{#1\right\}}}

\newcommand{\rdef}{=:}
\newcommand{\ldef}{:=}

\newcommand{\var}{\mathbb{V}\text{ar}}      

\newcommand{\condExpSmall}[2]{\condExpSmallGeneric{#1}{#2}{}}
\newcommand{\condVarSmall}[2]{\condVarSmallGeneric{#1}{#2}{}}

\newcommand{\condExpSmallGeneric}[3]{\E_{#2}^{#3}\left[ #1 \right]}
\newcommand{\condVarSmallGeneric}[3]{\var_{#2}^{#3}\left(#1\right)}

\newcommand{\etal}{\textit{et al. }}
\newcommand{\zeroRomanUpperCase}[1]{\ifcase #1 \relax 0 \else {\MakeUppercase{\romannumeral #1}}\fi}

\newtheorem{definition}{Definition}[section]    
\newtheorem{thrm}{Theorem}[section]             
\newtheorem{prop}{Proposition}[section]         
\newtheorem{crllry}{Corollary}[section]         
\newtheorem*{rem}{Remark}                       

\newcommand{\resultFigureSize}{0.49\linewidth}

\newcommand{\annuity}{A}
\newcommand{\bank}{B}

\newcommand{\objFun}{E}

\newcommand{\payoff}{H}

\newcommand{\strike}{K}

\newcommand{\notional}{N}
\newcommand{\shortRate}{r}

\newcommand{\zcb}{P}
\newcommand{\recovRate}{\RR}

\newcommand{\stock}{S}
\newcommand{\swapRate}{S}

\newcommand{\tradeVal}{V}

\newcommand{\brownian}{W}

\newcommand{\dct}{\alpha}
\newcommand{\shift}{\gamma}

\newcommand{\swapType}{\delta}

\newcommand{\hazardRate}{\xi}

\newcommand{\vol}{\sigma}

\newcommand{\normCDF}{\Phi}

\newcommand{\optType}{\omega}

\newcommand{\impliedVolFun}[1]{\vol^{\text{imp},#1}} 
\newcommand{\impliedVolMkt}{\impliedVolFun{\text{mkt}}}
\newcommand{\impliedVolMdl}{\impliedVolFun{\text{mdl}}} 
\newcommand{\tradeValMkt}{\tradeVal^{\text{mkt}}}
\newcommand{\tradeValMdl}{\tradeVal^{\text{mdl}}} 

\title{On the Hull-White model with volatility smile for Valuation Adjustments}

\begin{document}

\author[1,2]{Thomas van der Zwaard\corref{cor1}}
\ead{T.vanderZwaard@uu.nl}
\author[1,2]{Lech A.~Grzelak}
\ead{L.A.Grzelak@uu.nl}
\author[1]{Cornelis W.~Oosterlee}
\ead{C.W.Oosterlee@uu.nl}
\cortext[cor1]{Corresponding author at Mathematical Institute, Utrecht University, Utrecht, the Netherlands.}
\address[1]{Mathematical Institute, Utrecht University, Utrecht, the Netherlands}
\address[2]{Rabobank, Utrecht, the Netherlands}

\begin{abstract}
\noindent Affine Diffusion dynamics are frequently used for Valuation Adjustments ($\xva$) calculations due to their analytic tractability.
However, these models cannot capture the market-implied skew and smile, which are relevant when computing $\xva$ metrics.
Hence, additional degrees of freedom are required to capture these market features.
In this paper, we address this through an SDE with state-dependent coefficients.
The SDE is consistent with the convex combination of a finite number of different AD dynamics.
We combine Hull-White one-factor models where one model parameter is varied.
We use the Randomized AD (RAnD) technique to parameterize the combination of dynamics.
We refer to our SDE with state-dependent coefficients and the RAnD parametrization of the original models as the rHW model.
The rHW model allows for efficient semi-analytic calibration to European swaptions through the analytic tractability of the Hull-White dynamics.
We use a regression-based Monte-Carlo simulation to calculate exposures.
In this setting, we demonstrate the significant effect of skew and smile on exposures and $\xva$s of linear and early-exercise interest rate derivatives.
\end{abstract}

\maketitle

{\let\thefootnote\relax\footnotetext{The views expressed in this paper are the personal views of the authors and do not necessarily reflect the views or policies of their current or past employers.
The authors have no competing interests.}}

\section{Introduction}  \label{sec:introduction}

It is common practice to generate interest rate (IR) scenarios for Valuation Adjustment ($\xva$) calculations using one-factor short-rate models that belong to the Affine Diffusion (AD) class, motivated by the analytic tractability of these dynamics.
The most commonly used example is the Hull-White one-factor (HW) model with a time-dependent volatility parameter~\cite{Green201511}.
The model generates implied volatility skew, but the model skew cannot be controlled to fit the market-observed skew.
Furthermore, the model does not generate volatility smile.
In this paper, we address the critical need for a model to incorporate skew and smile in $\xva$ calculations for IR derivatives.
Furthermore, we demonstrate the significant impact of smile and skew on exposures and $\xva$s of IR derivatives, offering insights for practitioners that aim to align their $\xva$ models with market conditions.

Skew and smile have a significant effect on $\xva$ calculations, both in terms of mispricing $\xva$s, as well as for risk-management purposes.
The most obvious case is when the derivatives in a portfolio are smile-dependent.
From a consistency perspective, if the valuation model of a derivative includes smile, then, ideally, so does the xVA valuation model.
See~\cite[Section 19.1.3]{Green201511} for further discussion.
Having consistent valuation models for the trading book and $\xva$s is relevant from a hedging perspective to have more consistent $\xva$ hedges with the trading book hedges.
Including skew and smile in the $\xva$ model implies a consistent smile sensitivity for both valuations. 
In addition, the skew and smile are also relevant for linear derivatives, especially for legacy trades that may have a different rate than the current at-the-money (ATM) rate and, therefore, not primarily driven by ATM volatilities.

The relevance of skew and smile for $\xva$ computations has been demonstrated in literature for equity and FX derivatives~\cite{GraafFengKandhaiOosterlee201406,FengOosterlee201709,SimaitisGraafHariKandhai201605} and for Margin Valuation Adjustments of IR derivatives~\cite{HoencampKortKandhai202212}.
The most pronounced effects were observed for Potential Future Exposure ($\PFE$) as this is a tail metric of the future exposure distribution, and skew and smile affect the tails of the state variable's distribution.
In addition, the effect increased when considering, for example, early-exercise derivatives or other exotic derivatives.

It is challenging to find a suitable IR model for $\xva$ calculations that captures skew and smile while at the same time allowing for efficient and accurate calibration and pricing, which are the two main requirements when choosing underlying dynamics for $\xva$ purposes.
Some authors proposed to use Cheyette-type dynamics.
The one-factor Cheyette model fits in the Heath-Jarrow-Morton framework.
Adding a Displaced Diffusion to the Cheyette model allows skew to be generated.
Stochastic volatility can be added to this (e.g., through a CIR process for variance) to introduce curvature (smile).
For example, Andreasen used a four-factor Cheyette model with local and stochastic volatility~\cite{Andreasen201401}.
Hoencamp \etal used a two-factor Cheyette short-rate model, including a stochastic volatility component~\cite{HoencampKortKandhai202212}.
Due to the Cheyette separable form of instantaneous forward rate volatilities, this model is still Markovian with analytic bond prices.
Even though the Cheyette model can be calibrated to the ATM strips of European swaptions and their volatility slopes, however, only the smile curvature of one strip can be included~\cite[Section 16.3.2]{Green201511}.
Hence, not all the market information on smiles is incorporated in the model, but only a general slope of the smile.
Andreasen states that as a consequence, the curvature of all the smiles have to be roughly equivalent to have a sensible model~\cite{Andreasen201401}.
However, this is not necessarily the case in market data.
Furthermore, the calibration of the model to European swaptions requires (multiple) approximations of the swap rate before Fourier transform techniques can be used for the option pricing~\cite{Andreasen201401,HoencampKortKandhai202212}.
Due to the inability of the model to capture multiple market smiles, and due to the required swap rate approximations for European swaption pricing, the Cheyette dynamics do not meet all the requirements for $\xva$ calculations.

Therefore, we propose an alternative modelling approach, leveraging the analytic tractability of AD models~\cite{DuffiePanSingleton200011} while capturing the skew and smile observed from the model implied volatilities.
Additional degrees of freedom are required to achieve this.
Using the Fokker-Planck equation, we derive an SDE with state-dependent coefficients consistent with the convex combination of a finite number of AD dynamics.
We use the Randomized AD (RAnD) class of models~\cite{Grzelak202208,Grzelak202211} to reduce the dimensionality of extra model parameters to avoid overfitting.
For European option valuation, where the option value only depends on the terminal distribution, the option values are the weighted sum of the underlying option values.
This result is directly applicable during model calibration.
Generic derivative pricing and exposure simulation are done in a Monte-Carlo framework with regression methods.
After introducing a general framework, we introduce the randomized Hull-White (rHW) model based on the underlying HW dynamics, which belong to the class of AD models.
In this setting, we demonstrate the significant impact of skew and smile on (potential future) exposures and $\xva$ metrics of linear and early-exercise IR derivatives.
The simple and elegant rHW dynamics retain the desirable properties of AD models while introducing the flexibility to model smile and skew, satisfying the requirements of efficient and accurate calibration and pricing.

\section{Randomized Affine Diffusion} \label{sec:randGeneral}

The class of Affine Diffusion (AD) dynamics contains the SDEs where linearity conditions on the drift, diffusion and IR components are satisfied, i.e., the affinity conditions.
Under these conditions, the characteristic function (ChF) has a semi-closed exponential form~\cite{DuffiePanSingleton200011}, where the coefficients satisfy a set of complex-valued Ricatti ordinary differential equations. 
Using Fourier inversion techniques, the ChF can be used for efficient (occasionally analytical) derivative pricing, allowing for efficient calibration.
However, due to the linearity constraints, these models are unsuitable for exotic option pricing.
See~\cite[Section 7.3]{OosterleeGrzelak201911} for an extensive discussion on this class of models.
The Zero Coupon Bond (ZCB) is the ChF evaluated at $0$.
Therefore, these models allow for an analytic expression for the ZCB in an exponential form.
The analytic tractability motivates the frequent use of these models for $\xva$ purposes.

The starting point of this work is the convex combination of $\NMix \in \mathbb{N}$ different AD short-rate models $\shortRate_n(t)$, $n \in \{1,\ldots,\NMix\}$.
For each process $\shortRate_n(t)$, one model parameter is chosen and takes different values for each $n$.
We derive the SDE for short rate $\shortRate(t)$ with state-dependent drift or diffusion, consistent with this convex combination of short-rate models.
The dynamics of these processes under a generic measure $\M$ are driven by the same source of randomness $\brownian^{\M}(t)$:
\begin{align}
  \d\shortRate(t)
      &= \mu_{\shortRate}^{\M}(t, \shortRate(t)) \dt + \eta_{\shortRate}(t, \shortRate(t)) \dW^{\M}(t), \label{eq:randSDE} \\
  \d\shortRate_n(t)
    &= \mu_{\shortRate_n}^{\M}(t, \shortRate_n(t)) \dt + \eta_{\shortRate_n}(t, \shortRate_n(t)) \dW^{\M}(t), \label{eq:randUnderlyingSDE}
\end{align}
with initial condition $\shortRate(0) = \shortRate_n(0) = f^M(0,0)$ where the latter is the market instantaneous forward rate, which is defined in terms of the ZCBs implied by the market yield curve, i.e., $f^M(0,t) = - \pderiv{\log \zcb^{\text{M}}(0,t)}{t}$.
Here, we did not yet assume any particular form of the $\shortRate_n(t)$ dynamics, such that all results that follow are generic.

The definition of these $\NMix+1$ different short rates raises the question of how the measures, bank accounts and ZCBs are defined.
The short rate $\shortRate(t)$ defines bank account $\bank_{\shortRate}(t) = \expPower{\int_0^t \shortRate(s) \ds}$ which is the numeraire corresponding to the measure $\Q_{\shortRate}$.
Furthermore, the ZCB is the price of a unit payoff at terminal time $T$ under $\Q_{\shortRate}$, i.e.,
\begin{align}
  \zcb_{\shortRate}(t,T)
    &= \condExpSmallGeneric{\expPower{-\int_t^T \shortRate(s) \ds}}{t}{\Q_{\shortRate}}. \label{eq:randZCB}
\end{align}
Analogously, the $N$ individual processes $\shortRate_n(t)$ for $n \in \{1,\ldots,\NMix\}$ respectively define their own bank accounts $\bank_{\shortRate_n}(t)$ which are the numeraires corresponding to the measures $\Q_{\shortRate_n}$.
The ZCB is defined similarly as in Equation~\eqref{eq:randZCB}, where $\shortRate$ is replaced by $\shortRate_n$.
Additionally, due to affinity, the ZCB depends solely on the short rate at time $t$, i.e., $\zcb_{\shortRate_n}(t,T) = \zcb_{\shortRate_n}(t,T; \shortRate_n(t))$.

Combining $N$ processes increases the number of model parameters significantly compared to using a single dynamics to model an underlying.
Naively using these $N$ model parameters will lead to overfitting and difficulties during the model calibration.
However, we use the Randomized Affine Diffusion (RAnD) class of models~\cite{Grzelak202208,Grzelak202211}, where quadrature points of a distribution driven a few parameters are used to generate the $N$ parameters for the AD processes.
Hence, there is significant extra flexibility while keeping a limited amount of degrees of freedom in the model, which is particularly relevant during the calibration.
The original RAnD approach~\cite{Grzelak202208} focused on equities, with randomized parameters for the Black-Scholes and Bates dynamics, to achieve consistent option pricing on the S\&P500 and VIX.
The approach was extended to Heath-Jarrow-Morton short-rate dynamics, focusing on the HW model~\cite{Grzelak202211}.

Going forward, the process $\shortRate(t)$ is referred to as the RAnD model.
In this section, we derive the dynamics for $\shortRate(t)$ using the Fokker-Planck equation to ensure that the marginal distribution of $\shortRate(t)$ is consistent over time with the convex combination of marginal distributions of underlying short rates $\shortRate_n(t)$.
We prove that the pricing of plain vanilla derivatives under the RAnD model can be done using the weighted sum of prices under the finite number of AD dynamics.
To reduce the number of additional parameters in the RAnD model, we parameterize the model using the quadrature points of an exogenously specified distribution as in~\cite{Grzelak202208,Grzelak202211}.
Finally, we put our methodology in perspective by comparing our approach with mixture models and the uncertainty volatility model.

\subsection{Density equation} \label{sec:generalDensityEquation}

A convex combination of symmetric distributions yields heavy-tailed distributions.
Hence, combining distributions allows for a more flexible range of distributions.
We use this to reflect market smile and skew in our model.
These convex combinations are often defined in terms of the marginal distribution of the underlying dynamics.
Hence, the starting assumption is that under measure $\M$, at all points in time, the marginal distribution of $\shortRate(t)$ is consistent with the convex combination of marginal distributions of underlying short rates $\shortRate_n(t)$, see Definition~\ref{def:randPDF}.
All the results in this section will be given for a generic measure $\M$, since the results are required at later stages for various specific measures such as the risk-neutral measure $\M = \Q_{\shortRate}$ with numeraire $\bank_{\shortRate}(t)$, as well as the $T$-forward measure $\M = \Q_{\shortRate}^T$ with numeraire $\zcb_{\shortRate}(t, T)$.
\begin{definition}[RAnD PDF] \label{def:randPDF}
  Let $\shortRate(t)$ and $\shortRate_n(t)$ under measure $\M$ be defined as in Equations~\eqref{eq:randSDE} and~\eqref{eq:randUnderlyingSDE}, respectively.
  Let $f_{\shortRate(t)}^{\M}$ and $f_{\shortRate_n(t)}^{\M}$ denote respectively the PDFs of processes $\shortRate$ and $\shortRate_n$ at time $t$, under measure $\M$.
  The RAnD marginal distribution $f_{\shortRate(t)}^{\M}$ is defined as
  \begin{align}
    f_{\shortRate(t)}^{\M} (y)
      &\ldef \sum_{n=1}^{\NMix} \omega_n f_{\shortRate_n(t)}^{\M}(y), \label{eq:randPDF}
  \end{align}
  where non-negative weights $\omega_n \geq 0$ satisfy $\sum_{n=1}^{\NMix} \omega_n = 1$.
\end{definition}
This density equation defines the model for $\shortRate(t)$.
In~\ref{app:densityResults}, we prove that the density equation~\eqref{eq:randPDF} is measure invariant.
Therefore, the choice of measure $\M$ in Equation~\eqref{eq:randPDF} is not restrictive, i.e., the convex combination holds consistently through time and for all measures.
Two straightforward consequences of Definition~\ref{def:randPDF} for the RAnD CDF and RAnD moments are presented in Corollary~\ref{crllry:randCDF} and~\ref{crllry:randMoments}, respectively.
\begin{crllry}[RAnD CDF] \label{crllry:randCDF}
  Let the RAnD density be defined as in Definition~\ref{def:randPDF}, and let $F_{\shortRate_n(t)}^{\M}(y)$ denote the CDF of $\shortRate_n(t)$, then the RAnD CDF $F_{\shortRate(t)}^{\M}$ is given by:
  \begin{align}
    F_{\shortRate(t)}^{\M} (y)
      &= \sum_{n=1}^{\NMix} \omega_n F_{\shortRate_n(t)}^{\M}(y). \label{eq:randCDF}
  \end{align}
\end{crllry}
\begin{proof}
  The results follows from Equation~\eqref{eq:randPDF} by changing the integration and summation.
\end{proof}

\begin{crllry}[RAnD moments] \label{crllry:randMoments}
  Let the RAnD density be defined as in Definition~\ref{def:randPDF}, then the RAnD $m$'th central moment for $s < t$ is given by 
  \begin{align}
    \condExpSmallGeneric{\shortRate^m(t)}{s}{\M}
      &= \sum_{n=1}^{\NMix} \omega_n \condExpSmallGeneric{\shortRate_n^m(t)}{s}{\M}. \label{eq:randMoments}
  \end{align}
  Hence, when all moments of $\shortRate_n(t)$ are known analytically, the same holds for the $\shortRate(t)$ moments.
\end{crllry}
\begin{proof}
  The results follows from Equation~\eqref{eq:randPDF} by changing the integration and summation.
\end{proof}

\subsection{Density evolution} \label{sec:generalDensityEvolution}
For problems where the initial distribution is known, we can use the Fokker-Planck (FP) equation to obtain a PDE that describes the future evolution of the PDF in time.
This PDE is also known as the Kolmogorov forward PDE.
We re-iterate this well-known result in Proposition~\ref{prop:fokkerPlanck}.

\begin{prop}[Fokker-Planck equation] \label{prop:fokkerPlanck}
In general, for a process $X(t)$ that is governed by the following SDE
\begin{align}
  \d X(t)
    &= \mu(t, X(t)) \dt + \eta(t, X(t)) \dW(t), \nonumber
\end{align}
the FP equation for the density $f_{X(t)}(y)$ of $X(t)$ is
\begin{align}
  \pderiv{}{t} f_{X(t)}(y)
    &= - \pderiv{}{y} \left[ \mu(t, y) f_{X(t)}(y) \right] + \half \ppderiv{}{y} \left[ \eta^2(t, y) f_{X(t)}(y) \right], \label{eq:fokkerPlanck0}
\end{align}
where the initial condition is given by the Dirac delta function $f_{X(t_0)}(y) = \delta(y = X(t_0))$.
\end{prop}

\subsection{RAnD dynamics} \label{sec:randGeneralDynamics}

The goal is to find $\mu_{\shortRate}(t, \shortRate(t)) $ and $\eta_{\shortRate}(t, \shortRate(t))$ in Equation~\eqref{eq:randSDE}, such that Equation~\eqref{eq:randPDF} holds when the dynamics $\shortRate_n(t)$ are given as in Equation~\eqref{eq:randUnderlyingSDE}.
The result in Proposition~\ref{prop:randDynamics} specifies the required drift and diffusion terms such that the marginal law of process $\shortRate(t)$ follows the same evolution as the convex combination of PDFs.

\begin{prop}[RAnD dynamics] \label{prop:randDynamics}
  Let the marginal distribution of process $\shortRate(t)$ be defined as in Definition~\ref{def:randPDF}.
  Let the dynamics of $\shortRate_n(t)$ be given as in Equation~\eqref{eq:randUnderlyingSDE}, assuming that the diffusion is non-explosive, i.e., $\left(\eta_{\shortRate_n}(t, y)\right)^2 \leq C_n \left( 1 + y^2\right)$ $\forall n$ $\forall t$.
  Then, the dynamics of $\shortRate(t)$ are given as in Equation~\eqref{eq:randSDE} with the following drift and diffusion, again assuming the non-explosion condition $\left(\eta_{\shortRate}(t, y)\right)^2 \leq C \left( 1 + y^2\right)$ $\forall t$:
  \begin{align}
    \mu_{\shortRate}^{\M}(t, y)
      &= \sum_{n=1}^{\NMix} \mu_{\shortRate_n}^{\M}(t, y) \Lambda_n^{\M}(t,y), \label{eq:randSDEDrift0} \\
    \eta_{\shortRate}(t, y)
      &= \sqrt{\sum_{n=1}^{\NMix} \left(\eta_{\shortRate_n}(t, y)\right)^2 \Lambda_n^{\M}(t,y)}, \label{eq:randSDEDiffusion0} \\
    \Lambda_n^{\M}(t,y)
     &\ldef \frac{\omega_n  f_{\shortRate_n(t)}^{\M}(y)}{\sum_{i=1}^{\NMix} \omega_i  f_{\shortRate_i(t)}^{\M}(y)}. \label{eq:randSDEHelper}
  \end{align}
  The above holds for all measures $\M$.
  As $f_{\shortRate_i(0)}^{\M}(y)$ will typically not be defined due to zero variance at initial time, we set $f_{\shortRate_i(0)}^{\M}(y) = 1$ $\forall i$, such that $\Lambda_n^{\M}(0,y) = \omega_n$.
\end{prop}
\begin{proof}
  The proof is similar as the proof of Proposition 2.1 in~\cite{Grzelak202211}, which employs the Fokker-Planck equation to ensure that the marginal distribution of $\shortRate(t)$ is consistent over time with the weighted sum of marginal distributions of the processes $\shortRate_n(t)$.
    First, write down the FP equation from Proposition~\ref{prop:fokkerPlanck} for both $\shortRate(t)$ and $\shortRate_n(t)$:
    \begin{align}
      \pderiv{}{t} f_{\shortRate(t)}^{\M}(y)
        &= - \pderiv{}{y} \left[ \mu_{\shortRate}^{\M}(t, y) f_{\shortRate(t)}^{\M}(y) \right] + \half \ppderiv{}{y} \left[ \left(\eta_{\shortRate}(t, y)\right)^2 f_{\shortRate(t)}^{\M}(y) \right], \label{eq:fokkerPlanck1} \\
      \pderiv{}{t} f_{\shortRate_n(t)}^{\M}(y)
        &= - \pderiv{}{y} \left[ \mu_{\shortRate_n}^{\M}(t, y) f_{\shortRate_n(t)}^{\M}(y) \right] + \half \ppderiv{}{y} \left[ \left(\eta_{\shortRate_n}(t, y)\right)^2 f_{\shortRate_n(t)}^{\M}(y) \right]. \label{eq:fokkerPlanck2}
    \end{align}
    Substitution of Equation~\eqref{eq:randPDF} into Equation~\eqref{eq:fokkerPlanck1} yields:
    \begin{align}
      \pderiv{}{t} \left(\sum_{n=1}^{\NMix} \omega_n  f_{\shortRate_n(t)}^{\M}(y)\right)
        &= - \pderiv{}{y} \left[ \mu_{\shortRate}^{\M}(t, y) \sum_{n=1}^{\NMix} \omega_n  f_{\shortRate_n(t)}^{\M}(y) \right] + \half \ppderiv{}{y} \left[ \left(\eta_{\shortRate}(t, y)\right)^2 \sum_{n=1}^{\NMix} \omega_n  f_{\shortRate_n(t)}^{\M}(y) \right]. \label{eq:fokkerPlanck3}
    \end{align}
    Due to linearity of the derivative operator we can write:
    \begin{align}
      \pderiv{}{t} \left(\sum_{n=1}^{\NMix} \omega_n  f_{\shortRate_n(t)}^{\M}(y)\right)
        &= \sum_{n=1}^{\NMix} \omega_n  \pderiv{}{t}f_{\shortRate_n(t)}^{\M}(y), \nonumber 
    \end{align}
    where in the latter we can plug in the result from Equation~\eqref{eq:fokkerPlanck2}, which results in (using linearity of the derivative operator)
    \begin{align}
      \pderiv{}{t} \left(\sum_{n=1}^{\NMix} \omega_n  f_{\shortRate_n(t)}^{\M}(y)\right)
        &= \sum_{n=1}^{\NMix} \omega_n  \left( - \pderiv{}{y} \left[ \mu_{\shortRate_n}^{\M}(t, y) f_{\shortRate_n(t)}^{\M}(y) \right] + \half \ppderiv{}{y} \left[ \left(\eta_{\shortRate_n}(t, y)\right)^2 f_{\shortRate_n(t)}^{\M}(y) \right] \right) \nonumber \\
        &=  - \pderiv{}{y} \left[ \sum_{n=1}^{\NMix} \omega_n \mu_{\shortRate_n}^{\M}(t, y) f_{\shortRate_n(t)}^{\M}(y) \right] + \half \ppderiv{}{y} \left[ \sum_{n=1}^{\NMix} \omega_n \left(\eta_{\shortRate_n}(t, y)\right)^2 f_{\shortRate_n(t)}^{\M}(y) \right]. \label{eq:fokkerPlanck4}
    \end{align}
    Next, equate the results from~\eqref{eq:fokkerPlanck3} and~\eqref{eq:fokkerPlanck4}, matching the terms and integrating yields
    \begin{align}
      \mu_{\shortRate}^{\M}(t, y) \sum_{n=1}^{\NMix} \omega_n  f_{\shortRate_n(t)}^{\M}(y)
        &= \sum_{n=1}^{\NMix} \omega_n \mu_{\shortRate_n}^{\M}(t, y) f_{\shortRate_n(t)}^{\M}(y), \nonumber \\
      \left(\eta_{\shortRate}(t, y)\right)^2 \sum_{n=1}^{\NMix} \omega_n  f_{\shortRate_n(t)}^{\M}(y) 
        &= \sum_{n=1}^{\NMix} \omega_n \left(\eta_{\shortRate_n}(t, y)\right)^2 f_{\shortRate_n(t)}^{\M}(y), \nonumber
    \end{align}
    where from~\cite[Section 2.3]{Grzelak202211} it is clear that the relevant integration constants are all zero in order to satisfy uniform convergence requirements.
\end{proof}

Proposition~\ref{prop:randDynamics} implies that the drift/diffusion of $\shortRate(t)$ is the weighted average of the drifts/diffusions of underlying processes $\shortRate_n(t)$, where the weights are functions of the marginal distributions of $\shortRate_n(t)$.
Depending on the particular form of the $\shortRate_n(t)$ dynamics, the drift and/or diffusion of $\shortRate(t)$ can be state-dependent.
The latter case is known as a local volatility model.

By construction, the density equation~\eqref{eq:randPDF} holds for all points in time.
Using the Fokker-Planck equation, we have derived the SDE such that the RAnD density has the same evolution as the sum of weighted densities for all points in time.
Hence, the RAnD SDE captures the evolution of the convex combination of AD dynamics.

\subsection{Pricing equation} \label{sec:randGeneralPricing}

When considering convex combinations of densities for underlyings that are observable quantities, e.g., equities, `simple' derivative prices are the weighted sum of the derivative prices under the underlying models~\cite{BrigoMercurio200009}, where the density equation~\eqref{eq:randPDF} is imposed under the $T$-forward measure.
The result stated in Theorem~\ref{thrm:fastPricing} extends the aforementioned result to the case of IR derivatives, where the SDE does not correspond to an observable quantity.

\begin{thrm}[Pricing formula under the RAnD model] \label{thrm:fastPricing}
  Let $\tradeVal_{\shortRate_n}(t; T)$ denote the time $t$ value of a derivative under the AD dynamics $\shortRate_n(t)$ with payoff $\payoff(T;\shortRate_n(T))$ at time $T$.
  Assume that the payoff at time $T$ can be written as an explicit function of the state variable at time $T$, i.e., $\shortRate_n(T)$.
  Let the RAnD model $\shortRate(t)$ be defined as per Definition~\ref{def:randPDF}.
  Then, due to affinity of the $\shortRate_n(t)$ dynamics, the derivative price under $\shortRate(t)$, i.e., $\tradeVal_{\shortRate}(t;T)$, can be expressed as the convex combination of $\tradeVal_{\shortRate_n}(t;T)$:
  \begin{align}   
    \tradeVal_{\shortRate}(t;T)
      &= \sum_{n=1}^{\NMix} \omega_n \tradeVal_{\shortRate_n}(t;T). \label{eq:convexCombinationPricing}
  \end{align}
\end{thrm}
\begin{proof}
  Start from martingale pricing under each of the individual underlying affine models $\shortRate_n(t)$, and move from the risk-neutral measure $\Q_{\shortRate_n}$ to the $T$-forward measure $\Q_{\shortRate_n}^T$, i.e.,
  \begin{align}
    \tradeVal_{\shortRate_n}(t;T)
      &= \condExpSmallGeneric{ \frac{\bank_{\shortRate_n}(t)}{\bank_{\shortRate_n}(T)} \tradeVal_{\shortRate_n}(T;T)}{t}{\Q_{\shortRate_n}}
      = \zcb_{\shortRate_n}(t,T)  \condExpSmallGeneric{ \payoff(T;\shortRate_n(T))}{t}{\Q_{\shortRate_n}^T}. \label{eq:riskNeutralPricing}
  \end{align}
  Here, $\tradeVal_{\shortRate_n}(t;T)$ denotes the time $t$ value of the derivative with payoff $\payoff(T;\shortRate_n(T))$ at time $T$.
  The convex combination of Equation~\eqref{eq:riskNeutralPricing} yields
  \begin{align}
    \sum_{n=1}^{\NMix} \omega_n \tradeVal_{\shortRate_n}(t;T)
      &= \sum_{n=1}^{\NMix} \omega_n \zcb_{\shortRate_n}(t,T)  \condExpSmallGeneric{ \payoff(T;\shortRate_n(T))}{t}{\Q_{\shortRate_n}^T} \nonumber \\
      &= \zcb_{\shortRate}(t,T) \sum_{n=1}^{\NMix} \omega_n \condExpSmallGeneric{  \frac{\zcb_{\shortRate_n}(t, T)}{\zcb_{\shortRate}(t, T)} \payoff(T;\shortRate_n(T))}{t}{\Q_{\shortRate_n}^T} \nonumber \\
      &= \zcb_{\shortRate}(t,T) \sum_{n=1}^{\NMix} \omega_n \condExpSmallGeneric{  \payoff(T;\shortRate_n(T))}{t}{\Q_{\shortRate}^T}, \nonumber
  \end{align}
  moving from $\Q_{\shortRate_n}^T$ to the $\Q_{\shortRate}^T$ measure using the following Radon-Nikodym derivative $\forall n$:
  \begin{align}
    \lambda_{\Q_{\shortRate}^T}^{\Q_{\shortRate_n}^T}(T)
      &= \left.\deriv{\Q_{\shortRate_n}^T}{\Q_{\shortRate}^T}\right|_{\F(T)}
      = \frac{\zcb_{\shortRate}(t, T)}{\zcb_{\shortRate}(T, T)} \frac{\zcb_{\shortRate_n}(T, T)}{\zcb_{\shortRate_n}(t, T)}
      = \frac{\zcb_{\shortRate}(t, T)}{\zcb_{\shortRate_n}(t, T)}. \nonumber
  \end{align}
  Apply Equation~\eqref{eq:randPDF} with $\M = \Q_{\shortRate}^T$ and change the order of summation and integration:
  \begin{align}
    \sum_{n=1}^{\NMix} \omega_n \condExpSmallGeneric{  \payoff(T;\shortRate_n(T))}{t}{\Q_{\shortRate}^T}
      &= \sum_{n=1}^{\NMix} \omega_n \int_{\R} \payoff(T; x) f_{\shortRate_n(T)}^{\Q_{\shortRate}^{T}}(x) \dx 
      = \int_{\R} \payoff(T; x) \sum_{n=1}^{\NMix} \omega_n f_{\shortRate_n(T)}^{\Q_{\shortRate}^{T}}(x) \dx \nonumber \\
      &= \int_{\R} \payoff(T; x) f_{\shortRate(T)}^{\Q_{\shortRate}^{T}}(x) \dx 
      = \condExpSmallGeneric{  \payoff(T;\shortRate(T))}{t}{\Q_{\shortRate}^T}. \nonumber
  \end{align}
  Now, the payoff only depends on $\shortRate(T)$ rather than the whole trajectory of $\shortRate$, i.e., $\tradeVal_{\shortRate}(T;T) = \payoff(T;\shortRate(T))$.
  Hence, the affinity of $\shortRate_n(t)$ $\forall n$ can be used to rewrite the RAnD payoff in a convenient form.
  Standard martingale pricing yields:
  \begin{align}
    \zcb_{\shortRate}(t,T) \condExpSmallGeneric{  \payoff(T;\shortRate(T))}{t}{\Q_{\shortRate}^T}
      &= \zcb_{\shortRate}(t,T) \condExpSmallGeneric{  \tradeVal_{\shortRate}(T;T)}{t}{\Q_{\shortRate}^T}
      = \tradeVal_{\shortRate}(t;T). \nonumber
  \end{align}
  Combining all the results above yields the required result.
\end{proof}

In Theorem~\ref{thrm:fastPricing}, $\tradeVal_{\shortRate_n}(t;T)$ are arbitrage-free prices.
However, while calibrating the RAnD model to market prices, the model prices $\tradeVal_{\shortRate_n}(0;T)$ will not be consistent with the market prices $\tradeVal^{\text{mkt}}(0;T)$.
Only arbitrage-free price $\tradeVal_{\shortRate}(0;T)$ will be consistent with the market prices.

Theorem~\ref{thrm:fastPricing} also illustrates why the normalization of weights is imposed, i.e., $\sum_{n=1}^{\NMix} \omega_n = 1$.
For the RAnD model to fit the initial market yield curve $\zcb^{\text{M}}(0;T)$, this property is required, i.e., applying Equation~\eqref{eq:convexCombinationPricing} to the derivative $\zcb_{\shortRate}(0;T)$ yields:
\begin{align}   
  \zcb_{\shortRate}(0;T)
    &= \sum_{n=1}^{\NMix} \omega_n \zcb_{\shortRate_n}(0;T)
    = \zcb^{\text{M}}(0;T) \sum_{n=1}^{\NMix} \omega_n 
    = \zcb^{\text{M}}(0;T) . \label{eq:convexCombinationPricingZCB}
\end{align}

Due to the affinity, the ZCB under the AD dynamics depends solely on the state variable at the start date, i.e., $\zcb_{\shortRate_n}(t,T) = \zcb_{\shortRate_n}(t,T; \shortRate_n(t))$.
Hence, Theorem~\ref{thrm:fastPricing} holds for all linear IR derivatives and their European options.
As such, the RAnD model can be calibrated (semi) analytically to European option prices, depending on the pricing under the AD dynamics.
When calibrating to a term structure, due to the measure invariance of the density equation, each calibration instrument with a different expiry $T_i$ can be valued using Equation~\eqref{eq:convexCombinationPricing}.

Due to the analytic tractability of the underlying AD dynamics $\shortRate_n$, the valuation of $\tradeVal_{\shortRate_n}$ is typically fast in the case of linear or European derivatives.
Hence, Theorem~\ref{thrm:fastPricing} gives a powerful result where this efficient pricing under the underlying dynamics extends to the pricing under the RAnD dynamics.
However, this valuation formula cannot be naively used to price any derivative. 
The assumptions in Theorem~\ref{thrm:fastPricing} restrict to derivatives where the payoff can be expressed explicitly as a function of the state variable at payoff time.
Hence, for path-dependent derivatives, alternative valuation methods are required.

The elegance of this modelling setup is demonstrated in Theorem~\ref{thrm:fastPricing}.
If we would have started from the martingale pricing equation $\tradeVal_{\shortRate}(t;T) = \condExpSmallGeneric{ \frac{\bank_{\shortRate}(t)}{\bank_{\shortRate}(T)}  \tradeVal_{\shortRate}(T;T)}{t}{\Q_{\shortRate}}$, there would be no certainty that the payoff only depends on $\shortRate(T)$ rather than the whole trajectory of $\shortRate$, not even for European-type payoffs.
So, by construction, due to the affinity of the underlying dynamics, we obtain the explicit form with dependence only on $\shortRate(T)$.
In a way, we have correctly `guessed' the suitable form of the solution when the payoff is not path-dependent.

\subsection{RAnD parametrization} \label{sec:RAnDparametrization}

The RAnD model can be constructed starting from AD dynamics.
First, pick one of the model parameters from the drift $\mu_{\shortRate_n}^{\M}(t, \shortRate_n(t))$ or the diffusion $\eta_{\shortRate_n}(t, \shortRate_n(t))$ in Equation~\eqref{eq:randUnderlyingSDE}.
Then, use a different parameter value $\theta_n$ for each process, and continue to use the notation $\shortRate_n(t)\ldef \shortRate_n(t; \theta_n)$ for these processes.
The marginal distributions of these processes are then scaled by weights $\omega_n$ as per Equation~\eqref{eq:randPDF}, resulting in a set of $\NMix \in \mathbb{N}$ parameter pairs $\left\{\omega_n,\theta_n\right\}_{n=1}^{\NMix}$.
The setup can be interpreted as one of the AD model parameters following a discrete distribution, i.e., the parameter is randomized.
See~\cite{BrigoMercurio200009} for an example of the combination of lognormal distributions.

To avoid over-parametrization, we choose parameter pairs $\left\{\omega_n,\theta_n\right\}_{n=1}^{\NMix}$ to be given by the quadrature rule based on the first $\NMix$ moments of an exogenously specified continuous distribution~\cite{Grzelak202208}.
The only requirement of the exogenous stochasticity is that the moments of the continuous distribution exist and are finite.
Hence, the parameter pairs can still be viewed as a discrete distribution, which is specifically generated by a continuous distribution.
For example, the continuous distribution could be a normal or uniform distribution, i.e., $\U([\hat{a}, \hat{b}])$ or $\N (\hat{a}, \hat{b}^2 )$, such that one additional degree of freedom is introduced in the RAnD model w.r.t. the original AD dynamics.
The significant additional flexibility by merely adding one parameter is a key advantage of the model.

There is no approximation error like the quadrature error that needs to be considered since the quadrature is only used to reduce the dimensionality of the model parameters.
In other words, we do not view the discrete distribution $\left\{\omega_n,\theta_n\right\}_{n=1}^{\NMix}$ as the approximation of a continuous distribution of model parameter $\vartheta$, where convergence in $\NMix$ would need to be considered as in~\cite{Grzelak202208}.
The latter would lead to an approximate pricing formula compared to the equality in Theorem~\ref{thrm:fastPricing}, e.g., for $\vartheta \sim \N (\hat{a}, \hat{b}^2 )$ with CDF $F_{\vartheta}(x)$, where the approximation is driven by the quadrature error:
\begin{align}
  \tradeVal_{\shortRate(t;\vartheta)}(t; T)
    &= \int_{[a,b]} \tradeVal_{\shortRate(t;\theta)}(t; T) \d F_{\vartheta}(\theta)
    \approx \sum_{n=1}^{N} \omega_n \tradeVal_{\shortRate(t;\theta_n)}(t; T). \label{eq:fastPricingContinuousApprox}
\end{align}

At each quadrature point $\left\{\omega_n,\theta_n\right\}$ one has an AD model.
The parameter randomization relaxes the affinity constraints, as it is an external layer over the class of AD models.
The RAnD model can be calibrated using the pricing equation from Theorem~\ref{thrm:fastPricing}.
Recall that the parameters $\hat{a}$ and $\hat{b}$ drive the distribution $\N (\hat{a}, \hat{b}^2 )$ from which the quadrature points $\left\{\omega_n,\theta_n\right\}_{n=1}^{\NMix}$ are generated.
During the calibration, $\hat{a}$ and $\hat{b}$ are altered to generate new quadrature points, after which the pricing equation~\eqref{eq:convexCombinationPricing} is used to match the market implied price of the calibration instrument(s).
To allow for efficient calibration, the distribution of the randomized variable needs to have moments which can be computed in closed form.
The two main examples of uniform and normal distributions satisfy this requirement, and even allow to tabulate the quadrature points for a standard case and rescale them according to the parameters $\hat{a}$ and $\hat{b}$.
For further information on how these quadrature points are computed, see~\cite{Grzelak202208}.

\subsection{Mixture models and the uncertain volatility model}

The convex combination of distributions is also known as a mixture model.
In the early 2000s, several works appeared using (lognormal) mixtures to capture volatility smiles present in equity markets~\cite{BrigoMercurio200009,BrigoMercurio2002,BrigoMercurioSartorelli200303}.
While mixture models offer ample flexibility, there was criticism that some were stretching this flexibility beyond the reasonable.
For example, see the feedback by Piterbarg~\cite{Piterbarg200307}.
Mixture models are safe for European option valuation as their option value does not depend on the volatility path but only on the average volatility between now and expiry, i.e., the terminal distribution.
Hence, when valuing a non-path-dependent derivative, its value is the weighted sum of the derivative values in the various scenarios.
However, mixture models cannot be used to price contracts that depend on the dynamics of the underlying market variables.
The issue is that multiple dynamics can be consistent with the same mixture, leading to different prices.
For mixture models to work for all derivatives, one must fully specify the evolution of all model state variables through time and not use the valuation formula where prices are weighted over the different scenarios.
Brigo and Mercurio~\cite{BrigoMercurio200009} did this by deriving a local volatility model consistent with the assumption of a (lognormal) mixture.
This approach resulted in dynamics consistent with the chosen parametric form of the risk-neutral density used for calibration.
This local volatility model is fully self-consistent.
Path-dependent and early-exercise products can be valued using Monte-Carlo simulation of the local volatility dynamics.

We also follow this route in this paper: the convex combination of processes and the corresponding valuation equation for plain vanilla options is only used for calibration purposes.
For pricing derivatives and simulating exposures, we use the SDE with state-dependent coefficients.

When setting up a RAnD model with various values for the volatility parameter in the diffusion term, the dynamics in Proposition~\ref{prop:randDynamics} are local-volatility dynamics.
This model is not the same as the `uncertain volatility (UV) model', which is a type of stochastic volatility model~\cite{Brigo200812}.
For the latter, it is known that for each model there exists a local-volatility version with the same marginals.
Hence, the European options prices at $t=0$ are consistent.

In the UV model, the volatility parameter follows a discrete distribution with probabilities and parameter values similar to the RAnD parameter weights and values.
The randomness in the volatility is independent of the Brownian shocks, and uncertainty is resolved at time $\epsilon > 0$.
When conditioning on time $u > \epsilon$, the volatility uncertainty is resolved, and it is known which parameter value is realized.
Hence, the transition density between $u$ and $t>u$ reduces to the known density used in the RAnD model with a parameter realization.
As a result, the UV transition density is a convex combination of transition densities.
So, both the marginals and transition densities are now known.
The mixture model is the local-volatility version of the UV model, and only the marginals are specified for this model.

The UV model seems to be Piterbarg's interpretation of mixture models~\cite{Piterbarg200307}, who states ``A mixture model only specifies what is going to happen at a fixed time in the future, and not what happens in between now and then...
Therefore, a mixture model is inherently under-specified.
There exist multiple dynamics that lead to the same mixture at a given future time.''
In other words, the mixture model only specifies the marginals but not the transition densities.
Hence, the model cannot be used to price path-dependent derivatives, whereas to European options only depend on the marginal distributions.

In summary, when setting up the local-volatility model consistent with a particular convex combination of dynamics, the model should not be interpreted as a UV model.
The two models only have the same marginals.
This line of reasoning extends to any parameter uncertainty.

\section{The Randomized Hull-White model} \label{sec:rHW}

For the application of computing exposures for Valuation Adjustments, we extend the general results from Section~\ref{sec:randGeneral} and derive the SDE consistent with the convex combination of several Hull-White one-factor (HW) models where one parameter is varied, i.e., either the mean-reversion or the model volatility parameter.
The RAnD technique is used to parameterize the combination of dynamics, resulting in the final model.
Going forward, we refer to this as the randomized Hull-White (rHW) model.

In~\cite{Grzelak202211}, the RAnD method was extended to Heath-Jarrow-Morton short-rate dynamics, focussing on the HW model.
In a single-period setting, the model was calibrated to the swaption implied market smile and skew.
When randomizing the mean-reversion parameter, a richer spectrum of implied volatility shapes can be generated than for the model volatility parameter~\cite{Grzelak202211}.
A randomized mean-reversion leads to the desired degree of curvature in implied volatilities, though at the same time, the implied volatility level is affected.
The model volatility can then be used to adjust the implied volatility level to achieve a perfect ATM fit.
Both the uniform and normal randomizing distributions yield good results.

We prove that the various realizations of the HW model at the quadrature points allow one to value non-path-dependent derivatives under the rHW model as the weighted sum of derivative prices under the underlying dynamics.
Hence, the rHW model allows for an efficient and accurate calibration where we profit from the analytic tractability of the underlying dynamics.
The model is calibrated not only to ATM co-terminal European swaptions but additionally to the information encoded in the market implied volatility surface, including skew and smile.
We extend the single-period calibration of the model from~\cite{Grzelak202211} by calibrating to data for multiple expiry-tenor pairs, which is required for $\xva$ modelling.
For exposure simulation in a Monte-Carlo setting, regression methods can be used to obtain the simulated Zero Coupon Bonds (ZCBs) along the future market scenarios, which in turn can be used to value the derivatives along these scenarios.

\subsection{The HW dynamics} \label{sec:HWDynamics}
We use the Gaussian one-factor formulation~\cite{BrigoMercurio2006} of the HW dynamics:
\begin{align}
  \shortRate(t)
    &= x(t) + b(t), \ \
  \dx(t)
    = -a_x x(t)\dt + \vol_x \d\brownian^{\Q_{\shortRate}}(t). \nonumber
\end{align}
Here, $a_x$ is the mean-reversion parameter, and $\vol_x$ is the model volatility; $b(t)$ is deterministic and used to fit the model to the market yield curve.

The HW model is often used for $\xva$ calculations due to its favourable properties: it falls within the AD class, allowing for analytic tractability; the short rate is normally distributed; calibration to swaptions can be done efficiently with a semi-analytic swaption pricing formula.
See~\ref{app:HW} for these well-known results.

For the time-homogeneous case, there are only two parameters to fit the whole market volatility surface, which is an impossible task.
Therefore, a time-dependent piece-wise constant volatility parameter is often used to calibrate the model to a strip of ATM co-terminal swaptions~\footnote{Co-terminal swaptions are swaptions where the sum of the expiry (of the option) and the maturity (of the underlying swap) are a constant number. In a matrix of implied swaption volatilities for a given strike, this amounts to selecting the counter diagonal.}.
The mean-reversion and volatility parameters have similar effects on the Black implied volatilities~\cite{OosterleeGrzelak201911}.
Therefore, due to the limited impact of the mean-reversion parameter, it is typically taken constant in practice and re-calibrated on an infrequent basis (e.g., weekly or monthly).
On the other hand, in practice, the volatility parameter is re-calibrated daily.

The forward rate under the HW model follows a shifted-lognormal distribution, such that the forward rate is displaced.
Hence, there is model skew by construction, but it does not necessarily match the market skew.
Alas, the model skew cannot be controlled, which is one of the model drawbacks.
In addition, the model does not generate a volatility smile.
Thus, it is only possible to calibrate the model to a limited part of the market volatility surface.

For ease of display, all results are presented for the case of constant model volatility $\vol_x$.
However, in the $\xva$ context, we require a time-dependent model volatility $\vol_x(t)$.
The results for this case can be found in~\ref{app:HWTimeDepVol}.

\subsection{The rHW dynamics} \label{sec:rHWDynamics}
The approach outlined in Section~\ref{sec:RAnDparametrization} can be used to derive the rHW dynamics: choose one of the model parameters to vary in the convex combination, which we refer to as randomization.
Due to the aforementioned effects on the shapes of implied volatility curves, we randomize mean reversion $a_x$ according to a discrete distribution, represented by $\NMix$ parameter pairs $\left\{\omega_n,\theta_n\right\}_{n=1}^{\NMix}$.
We work under the unique risk-neutral measure $\M = \Q_{\shortRate}$, i.e., the model is defined through the PDF from Definition~\ref{def:randPDF} under measure $\Q_{\shortRate}$.
The result is $N$ HW models $\shortRate_n(t)$ with mean-reversion parameter $\theta_n$, all driven by the same diffusion $\vol_x \brownian^{\Q_{\shortRate}}(t)$:
\begin{align}
  \shortRate_n(t)
    &= x_n(t) + b_n(t), \ 
  \d x_n(t)
    = - \theta_n x_n(t) \dt + \vol_x \dW^{\Q_{\shortRate}}(t). \label{eq:HWSDE1} 
\end{align}
The following integrated from for $\shortRate_n(t)$ can be derived using an integrating factor and Ito's lemma, i.e., for $s<t$:
\begin{align}
  \shortRate_n(t)
    &= x_n(s)\expPower{-\theta_n (t-s)} +  b_n(t) + \vol_{x} \int_{s}^{t}  \expPower{-\theta_n  (t-u)} \d\brownian^{\Q_{\shortRate}}(u). \label{eq:HWSDE1Integrated}
\end{align}
Clearly, $\shortRate_n(t) \sim \N \left( \condExpSmallGeneric{\shortRate_n(t)}{s}{\Q_{\shortRate}}, \condVarSmall{\shortRate_n(t)}{s} \right)$ conditional on $\F_s$, s.t.
\begin{align}
  &f_{\shortRate_n(t)}^{\Q_{\shortRate}}(y)
    = f_{\N \left( \bar{\mu}, \bar{\sigma}^2 \right)}^{\Q_{\shortRate}}(y), \label{eq:HWPDF} \\
  &\bar{\mu}
    \ldef \condExpSmallGeneric{\shortRate_n(t)}{s}{\Q_{\shortRate}}
    = x_n(s)\expPower{-\theta_n (t-s)} + b_n(t), \ 
  \bar{\sigma}^2
    \ldef \condVarSmall{\shortRate_n(t)}{s}
    = \frac{\vol_{x}^2 }{2\theta_n} \left( 1 - \expPower{-2\theta_n (t-s)} \right). \nonumber
\end{align}
The deterministic function $b_n(t)$ in Equation~\eqref{eq:HWSDE1} is obtained by fitting the model to the market yield curve, which is done in the derivation of the ZCB, see~\ref{app:HWZeroCouponBond}.

The dynamics from Equation~\eqref{eq:HWSDE1} do not yet fit in the required form of Equation~\eqref{eq:randUnderlyingSDE}.
The dynamics in Proposition~\ref{prop:randUnderlyingSDEHW} have the desired form.

\begin{prop}[HW dynamics in RAnD form] \label{prop:randUnderlyingSDEHW}
  The HW dynamics for $\shortRate_n(t)$ from Equation~\eqref{eq:HWSDE1} can be written in the form of Equation~\eqref{eq:randUnderlyingSDE} with drift and diffusion:
  \begin{align}
  \mu_{\shortRate_n}^{\Q_{\shortRate}}(t, \shortRate_n(t))
    &= \deriv{f^{\text{M}}(0,t)}{t} + \theta_n f^{\text{M}}(0,t) - \theta_n \shortRate_n(t) + \condVarSmall{\shortRate_n(t)}{0}, \label{eq:randUnderlyingSDEHWDrift} \\
  \eta_{\shortRate_n}(t, \shortRate_n(t))
    &= \vol_x. \label{eq:randUnderlyingSDEHWDiffusion}
  \end{align}
  The drift from Equation~\eqref{eq:randUnderlyingSDEHWDrift} is state dependent, whereas the diffusion in Equation~\eqref{eq:randUnderlyingSDEHWDiffusion} is of the same form as the diffusion of $\shortRate_n$ from Equation~\eqref{eq:HWSDE1}.
\end{prop}
\begin{proof}
  The SDE corresponding to the process in Equation~\eqref{eq:HWSDE1} is
  \begin{align}
    \d\shortRate_n(t)
      &= \left[ \deriv{b_n(t)}{t} + \theta_n b_n(t) - \theta_n \shortRate_n(t) \right]\dt + \vol_x \dW^{\Q_{\shortRate}}(t). \label{eq:HWSDE2}
  \end{align}
  Matching the terms with those in Equation~\eqref{eq:randUnderlyingSDE} yields
  \begin{align}
    \mu_{\shortRate_n}^{\Q_{\shortRate}}(t, \shortRate_n(t))
      &= \deriv{b_n(t)}{t} + \theta_n b_n(t) - \theta_n \shortRate_n(t), \ \ 
    \eta_{\shortRate_n}(t, \shortRate_n(t))
      = \vol_x. \nonumber
  \end{align}
  Using $b_n(t)$ from Equation~\eqref{eq:HWb} and the Leibniz integration rule we write:
  \begin{align}
    \deriv{b_n(t)}{t}
      &= \deriv{f^{\text{M}}(0,t)}{t} + x_n(0) \theta_n \expPower{-\theta_n t} + \vol_x^2 B_n(0,t) \expPower{-\theta_n t}, \ 
    B_n(s,t)
      \ldef \frac{1}{\theta_n}\left( 1 - \expPower{-\theta_n (t-s)} \right). \nonumber
  \end{align}
  The final form of the drift is obtained by combining this result with $b_n(t)$ from Equation~\eqref{eq:HWb} and simplifying the terms.
\end{proof}

With the HW dynamics from Proposition~\ref{prop:randUnderlyingSDEHW}, it is possible to extend the generic result in Proposition~\ref{prop:randDynamics} to obtain the rHW dynamics.
This result is given in Corollary~\ref{crllry:rHWDynamics}.

\begin{crllry}[rHW dynamics] \label{crllry:rHWDynamics}
  Let the marginal distribution of short rate $\shortRate(t)$ be as in Definition~\ref{def:randPDF}.
  Let the dynamics of $\shortRate_n(t)$ HW dynamics from Equation~\eqref{eq:HWSDE1}, and with density from Equation~\eqref{eq:HWPDF}. 
  Then, the rHW dynamics of $\shortRate(t)$ are given as in Equation~\eqref{eq:randSDE}, with
  \begin{align}
    \mu_{\shortRate}^{\Q_{\shortRate}}(t, y)
      &= \sum_{n=1}^{\NMix}\left[\deriv{f^{\text{M}}(0,t)}{t} + \theta_n f^{\text{M}}(0,t) - \theta_n y + \condVarSmall{\shortRate_n(t)}{0} \right]\Lambda_n^{\Q_{\shortRate}}(t,y), \label{eq:rHWDrift} \\
    \eta_{\shortRate}(t, y)
      &= \vol_x, \label{eq:rHWDiffusion}
  \end{align}
  with weighting function $\Lambda_n^{\Q_{\shortRate}}(t,y)$ from Equation~\eqref{eq:randSDEHelper}.
  Hence, the diffusions of $\shortRate(t)$ and $\shortRate_n(t)$ are equivalent, whereas the drift of $\shortRate(t)$ is state dependent in a non-linear way.
\end{crllry}
\begin{proof}
  The result is a direct consequence of Propositions~\ref{prop:randDynamics} and~\ref{prop:randUnderlyingSDEHW}.
\end{proof}

\begin{rem}[Numerical stability of $\Lambda_n^{\Q_{\shortRate}}(t,y)$]
  The term $\Lambda_n^{\Q_{\shortRate}}(t,y)$ from Equation~\eqref{eq:randSDEHelper} can in some occasions be numerically instable due to overflows.
  From the normality of $\shortRate_n(t)$, see Equation~\eqref{eq:HWPDF}, it is possible to write
  \begin{align}
    \omega_n  f_{\shortRate_n(t)}^{\Q_{\shortRate}}(y)
      &= \omega_n \frac{1}{\sqrt{\condVarSmall{\shortRate_n(t)}{s}}\sqrt{2\pi}} \expBrace{-\frac{\left( y - \condExpSmallGeneric{\shortRate_n(t)}{s}{\Q_{\shortRate}}\right)^2}{2 \condVarSmall{\shortRate_n(t)}{s}}}
      \rdef \expPower{ \gamma_n^{\Q_{\shortRate}}(t, y) }. \nonumber 
  \end{align}
  Apply this to $\Lambda_n^{\Q_{\shortRate}}(t,y)$ from Equation~\eqref{eq:randSDEHelper}:
  \begin{align}
    \Lambda_n^{\Q_{\shortRate}}(t,y)
       &= \frac{\omega_n  f_{\shortRate_n(t)}^{\Q_{\shortRate}}(y)}{\sum_{i=1}^{\NMix} \omega_i  f_{\shortRate_i(t)}^{\Q_{\shortRate}}(y)} 
       = \frac{\expPower{ \gamma_n^{\Q_{\shortRate}}(t, y) }}{\sum_{i=1}^{\NMix} \expPower{ \gamma_i^{\Q_{\shortRate}}(t, y) }}
       = \pderiv{}{\gamma_n^{\Q_{\shortRate}}(t, y)} \left( \ln \sum_{i=1}^{\NMix} \expPower{ \gamma_i^{\Q_{\shortRate}}(t, y) }\right), \nonumber
  \end{align}
  which can be recognized as the softmax function, which is the gradient of the log-sum-exp function.
  The following manipulation can be done to guarantee numerical stability of $\Lambda_n^{\Q_{\shortRate}}(t,y)$:
  \begin{align}
    \Lambda_n^{\Q_{\shortRate}}(t,y)
       &= \frac{C \expPower{ \gamma_n^{\Q_{\shortRate}}(t, y) }}{C\sum_{i=1}^{\NMix} \expPower{ \gamma_i^{\Q_{\shortRate}}(t, y) }} 
       = \frac{ \expPower{ \gamma_n^{\Q_{\shortRate}}(t, y) + \ln (C)}}{\sum_{i=1}^{\NMix} \expPower{ \gamma_i^{\Q_{\shortRate}}(t, y) + \ln (C)}}, \nonumber
  \end{align}
  where choosing $\ln(C) \ldef - \max_i \gamma_i^{\Q_{\shortRate}}(t, y)$ avoids overflows in the computation.
\end{rem}

\subsection{Swaption pricing for model calibration} \label{sec:rHWSwaptionPricing}

An efficient pricing function is required to calibrate the rHW model to swaptions.
For the underlying dynamics, swaptions can be priced semi-analytically.
See~\ref{app:HWSwaption} for more details.
Combining this with the pricing equation from Theorem~\ref{thrm:fastPricing} leads to Corollary~\ref{crllry:rHWSwaption}, which yields semi-analytic swaption pricing under the rHW model, such that this beneficial property of the underlying dynamics is preserved.

\begin{crllry}[rHW swaption pricing formula] \label{crllry:rHWSwaption}
  Assume the set-up of Theorem~\ref{thrm:fastPricing}.
  Consider a swaption with expiry date $T_M \leq T_0$, where $T_0$ is the first swap reset date.
  The swap is a unit notional swap starting at $T_0$, maturing at $T_m$ with intermediate payment dates $T_1 < T_2 < \ldots < T_{m-1} < T_m$, with strike $\strike$, and $\swapType = -1$ for a payer swap, $\swapType = 1$ for a receiver swap.

  The value of a swaption under the rHW model $\shortRate(t)$ can be computed as:
  \begin{align}   
    \tradeVal^{\text{Swaption}}_{\shortRate}(t;T_M,T_1,\ldots,T_m,\strike,\swapType)
      &= \sum_{n=1}^{\NMix} \omega_n \tradeVal^{\text{Swaption}}_{\shortRate_n}(t;T_M,T_1,\ldots,T_m,\strike,\swapType), \label{eq:rHWSwaption}
  \end{align}
  where $\tradeVal^{\text{Swaption}}_{\shortRate_n}(\cdot)$ is given in Proposition~\ref{prop:HWSwaption}.
\end{crllry}
\begin{proof}
  This result is a direct consequence of Theorem~\ref{thrm:fastPricing} and Proposition~\ref{prop:HWSwaption}.
\end{proof}

This pricing result is only used at $t=0$ for calibrating the model to European options (see Section~\ref{sec:rHWCalibration}).
The result from Theorem~\ref{thrm:fastPricing} is restricted to derivatives with payoffs that can be written as a function of the state variable at the payoff date.
As such, path-dependent derivatives, e.g., early-exercise products, require alternative valuation methods.
Therefore, we consider a generic simulation-based framework to value any derivative and exposures of those derivatives, see Sections~\ref{sec:rHWSimulation} and~\ref{sec:rHWGenericPricing}.

\subsection{Model calibration} \label{sec:rHWCalibration}

The model must now be calibrated to market data.
As instruments, we use swaptions, which contain information about the correlation between different points on the discount curve~\cite{PuetterRenzitti202011a}.
The market implied swaption volatility surfaces are three dimensional, as volatilities are quoted in the market for $\NExp$ different expiries, $\NTen$ different tenors, and $\NStrike$ different strikes.
We denote the sets of expiries, tenors and strikes used for calibration as $\TExpCalib$, $\TTenCalib$ and $\strikeCalib$, respectively.
The model implied volatilities are shifted Black volatilities, see the remark below for more details.

\begin{rem}[Market implied volatilties]
The swaptions quoted in the market are typically quoted as (shifted) Black volatilities, which are (shifted) Black-Scholes volatilities with zero interest rate.
The value of a unit notional swap can be written in terms of the swap rate $\swapRate(t)$ and annuity $\annuity(t)$, see~\ref{app:HWSwaption} for more details:
\begin{align}
  &\tradeVal^{\text{Swap}}(t;T_1,\ldots,T_m,\strike,\swapType)
    = \swapType \annuity(t) \left[ \strike - \swapRate(t) \right], \nonumber \\
  &\swapRate(t)
    \ldef \frac{\zcb(t,T_0) - \zcb(t,T_m)}{\annuity(t)}, \ \
  \annuity(t)
    \ldef \sum_{k=1}^m \dct_k \zcb(t,T_k), \ \
  \dct_k
    \ldef T_k - T_{k-1}. \nonumber
\end{align}
For the corresponding swaption, we change to the measure associated with annuity $\annuity(t)$:
\begin{align}
  \tradeVal^{\text{Swaption}}(t;T_M,T_1,\ldots,T_m,\strike,\swapType)
    &= \condExpSmall{\expPower{- \int_t^{T_M} \shortRate(s) \ds}\maxOperator{ \swapType \annuity(T_M) \left[ \strike - \swapRate(T_M) \right]} }{t} \nonumber \\
    &=  \annuity(t) \condExpSmallGeneric{\maxOperator{\swapType \left[ \strike - \swapRate(t) \right]} }{t}{\annuity}. \nonumber
\end{align}
When assuming that the swap rate $\swapRate(T_M)$ is log-normally distributed, the swaption can be priced using the Black 76 formula.
However, with interest rates possibly being negative, we shift the distribution to avoid this issue.
Therefore, we assume that the swap rate $\swapRate(T_M)$ is shifted lognormal under the annuity measure, i.e.,
\begin{align}
  \d \swapRate(t)
    &= \vol \left( \swapRate(t) + \shift \right) \dW(t). \nonumber
\end{align}
The volatilities associated with this shifted lognormal model are called shifted Black volatilities:
\begin{align}
  \tradeVal^{\text{Swaption}}(t;T_M,T_1,\ldots,T_m,\strike,\swapType)
    &=  \annuity(t) \condExpSmallGeneric{\maxOperator{\swapType\left[ \strike - \swapRate(T_M) \right]} }{t}{\annuity} \nonumber \\
    &= \annuity(t) \cdot \bs(t, T, S(t) + \shift, \strike + \shift, 0, \vol, -\swapType), \nonumber
\end{align}
where
\begin{align}
  &\bs(t, T, S, \strike, \shortRate, \vol, \optType)
    = \optType \left[ S \normCDF(\optType d_1) - \strike \expPower{-\shortRate (T - t)} \normCDF(\optType d_2)  \right], \nonumber \\
  &d_1
    \ldef \frac{\ln \left( \frac{S}{\strike}\right) + \left(r + \half \vol^2\right) (T-t)}{\vol \sqrt{T - t}}, \ 
  d_2
    \ldef d_1 - \vol \sqrt{T - t}, \ 
  \optType
    \ldef \left\{
       \begin{array}{cc}
    	1, & \text{call}, \\
    	-1, & \text{put}.
       \end{array}
       \right. \nonumber
\end{align}
\end{rem}

The result from Corollary~\ref{crllry:rHWSwaption} gives semi-analytic swaption pricing, which allows for an efficient calibration phase.
When calibrating a model for $\xva$, multiple expiry-tenor combinations need to be included to end up with a useful model to compute exposure profiles.
The ability of the rHW model to be calibrated to market smile and skew is demonstrated in~\cite{Grzelak202211}, where the model is calibrated to the swaption implied volatilities for a single expiry-tenor pair.
Here, we extend this to a multiple expiry-tenor calibration through a HW model with piece-wise constant model volatility.

First, recall the calibration of HW dynamics.
Due to the high dimensionality of the market data ($\NExp \times \NTen \times \NStrike$), this model cannot be calibrated perfectly to the whole volatility surface.
Hence, particular strips of the volatility surface need to be emphasized.
Using the co-terminal strip (also known as the counter diagonal) is a common choice in an $\xva$ context, as co-terminal swaptions match the amount of counterparty credit risk as seen in the exposure profiles used for $\xva$~\cite{PuetterRenzitti202011a}.
This strip is specified by the combination of expiries and tenors that sum to the same swaption maturity $\TCot$.
The mean-reversion parameter can be used to get a global fit to the ATM instruments for all expiry-tenor combinations of the volatility surface.
For more details on the HW calibration, see~\ref{app:HWCalibration}.

\subsubsection{Moneyness, smile and skew} \label{sec:marketConcepts}
Moneyness, smile and skew are crucial when calibrating a model to market implied volatilities.
Hence, we summarize them one by one.

\paragraph{Moneyness} 
Moneyness indicates the relative position of the current strike $\strike$ or price w.r.t. the at-the-money (ATM) strike $\strikeATM$.
For strikes below the ATM strike, i.e., $\strike < \strikeATM$, puts are out of the money (OTM) and calls are in the money (ITM).
For strikes above the ATM strike, i.e., $\strike > \strikeATM$, puts are ITM and calls are OTM.
See Figure~\ref{fig:CallPutMoneyness} for a visualization.
\begin{figure}[H]
  \centering
  \includegraphics[width=0.45\linewidth]{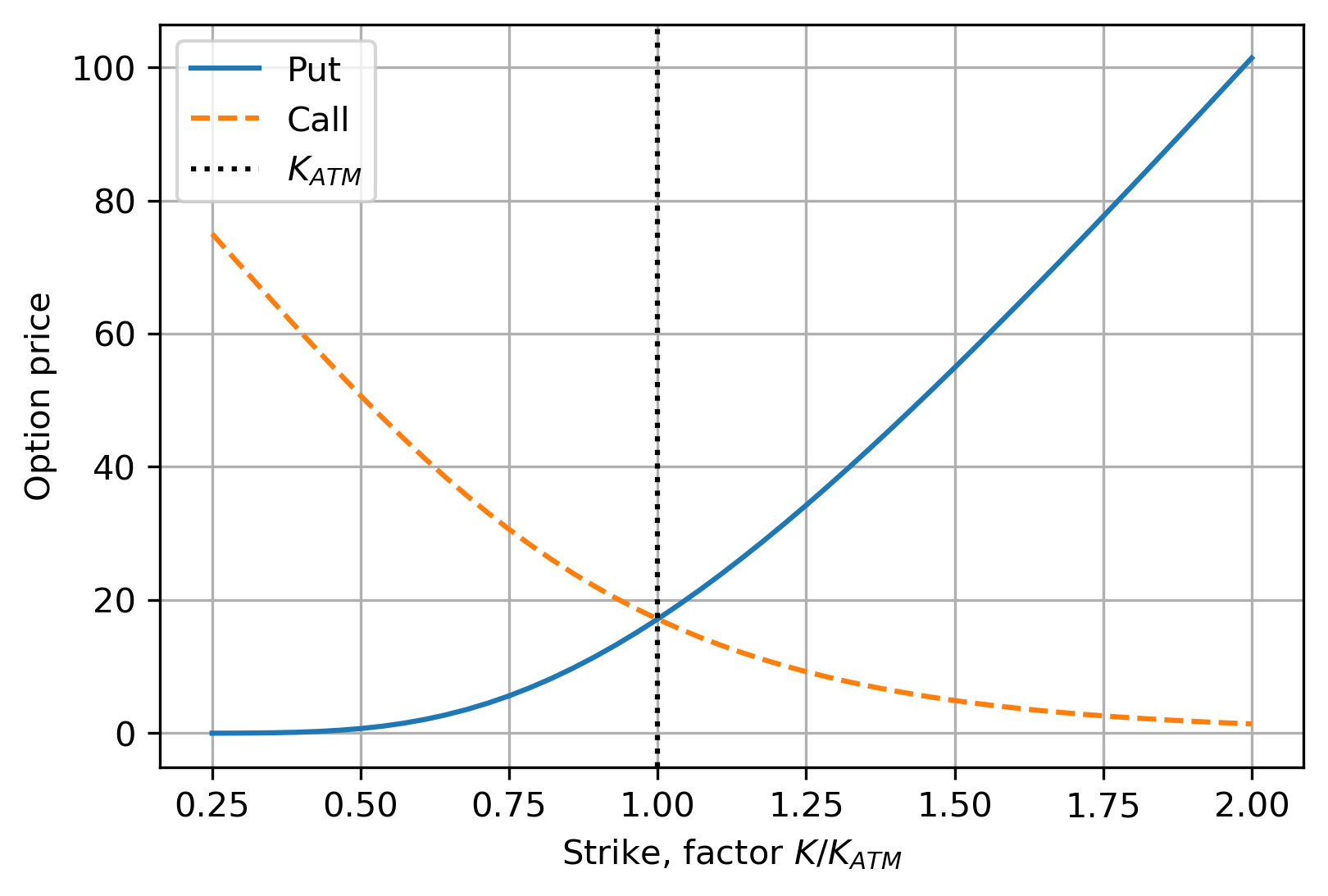}
  \caption{European option at $t = 0$ on a stock $\stock$ that follows a GBM dynamics, with $S(0) = 100$, $\shortRate = 0.03$, $\vol = 0.25$, and $T = 3$.}
  \label{fig:CallPutMoneyness}
\end{figure}

\paragraph{Smile} 
Volatility smile occurs when a higher probability is attributed to extreme events, represented by fatter tails of the market-implied distribution than implied by a lognormal distribution.
For a fixed option maturity, implied volatilities are lowest at the ATM strike and rise when the underlying is further ITM/OTM.
A higher demand for these options away from the ATM point means the market assigns a higher probability to extreme events.
This higher probability results in higher prices, such that the implied volatilities are also higher.
The market smiles are typically less pronounced as the option maturity increases.
Asymmetry of the smile tends to increase during periods of market turbulence.

Since OTM options are typically more liquid than ITM options, the former are often used for model calibration.
The OTM options imply an implicit weighting when using the option value in the calibration objective function, as OTM options are cheapest, see Figure~\ref{fig:CallPutMoneyness}.
 
\paragraph{Skew} 
Volatility skew is present when implied volatility goes down as the strike increases.
This phenomenon is caused by market participants' willingness to `overpay' for lower strike options.
In other words, investors have market concerns and buy puts to compensate for perceived risks.
Market implied volatility is the market's forecast of a likely movement in underlying.
In this case, more volatility is assigned to the downside than the upside, which is a possible indicator that downside protection is more valuable than upside speculation (for the stock market).

In terms of the implied distribution tails, the left tail is heavier than the lognormal distribution, and the right tail is less heavy.
Skew results from a negative correlation between the underlying and its volatility.
As the underlying moves down (up), the volatility tends to move up (down), making even greater (smaller) declines (increases) possible.
This mechanism explains the heavy left and thin right implied distribution tail.

\subsubsection{Market data} \label{sec:rHWCalibrationMarketData}

The flexibility of the rHW model allows more market data to be included in the calibration than for the HW model.
However, calibrating to the whole market surface is still not possible.
Once more, we need to consider what parts of the volatility surface are vital for calibration and where to put the emphasis.

Whereas the HW model can only calibrate to a single strike, the rHW model can capture the market skew and smile.
Hence, it is crucial to understand what information is captured in the wings of the volatility surface.
The left wing, i.e., $\strike < \strikeATM$, gives information about how valuable downward protection is.
The right wing, i.e., $\strike > \strikeATM$, provides information about the market perception of the upward potential.
Preferably, all this information is captured through calibration.
The natural extension of the HW calibration is to use the full range of strikes for each co-terminal expiry-tenor pair rather than a single point.

\subsubsection{Calibration procedure} \label{sec:rHWCalibrationProcedure}

For the rHW model calibration, we use a normal distribution $\N (\hat{a}, \hat{b}^2 )$ to generate the quadrature points, which have proven successful in calibrating the $1y$ market swaption smile for multiple tenors~\cite{Grzelak202211}.
In this case, the rHW model has one additional degree of freedom compared to the HW model.
We obtain $\NMix$ different models $\shortRate_n(t)$ with mean reversion $\theta_n$ and corresponding weights $\omega_n$.

We calibrate the rHW model to co-terminal swaptions, using all strikes for each expiry-tenor combination rather than just the ATM point.
The swaptions are valued using the result from Corollary~\ref{crllry:rHWSwaption}.
The calibration steps are summarized in Algorithm~\ref{algo:rHWCalib}.
During the calibration, we only consider OTM swaptions for a natural `weighting' in the target function in terms of the option values.
Let $\strike_{\ATM}$ denote the ATM strike, which is the swap rate such that the payer and receiver swap both have zero value.
Then a call (put) swaption is a receiver (payer) swaption, which is an option to enter a swap where the floating (fixed) leg is paid.
This swaption is ITM when $\strike > \strike_{\ATM}$ ($\strike < \strike_{\ATM}$) and OTM when $\strike < \strike_{\ATM}$ ($\strike > \strike_{\ATM}$).
For the HW model, the swaption pricing with Jamshidian's decomposition, see Proposition~\ref{prop:HWSwaption}, is done using ZCB calls (puts).
The Black-Scholes implied volatility is then computed with a Black-Scholes put (call) option.

\begin{algorithm}[ht!]
    \footnotesize
    \SetAlgoLined 
    \DontPrintSemicolon 


    \KwIn{Market implied volatility surface, co-terminal maturity $\TCot$, number of quadrature points $\NMix$}
    \KwOut{Calibrated parameters $\hat{a}$, $\hat{b}$ and $\vol_{x}(t)$}
    \BlankLine
	Initialize parameters $\hat{a}$ and $\hat{b}$ \;
    Choose the set of calibration instruments, and define calibration objective function $\objFun(\hat{a}, \hat{b}, \vol_{x})$ \;

    Set calibration tolerance $\varepsilon$ \;
    Initialize the volatility pillar dates\;

    \While{$\objFun(\hat{a}, \hat{b}, \vol_{x})  > \varepsilon$}{
        Numerical optimization step to update $\hat{a}$ and $\hat{b}$ \;
        Compute the Gauss-quadrature weights $\left\{\omega_n,\theta_n\right\}_{n=1}^{\NMix}$ corresponding to $\N (\hat{a}, \hat{b}^2 )$, see~\cite{Grzelak202211} \;     
        Initialize $N$ HW processes $\shortRate_n$ with mean reversion $\theta_n$ as per Equation~\eqref{eq:HWSDE1} \;
        Bootstrap calibration of piece-wise constant model volatility $\vol_{x}$ to get a good ATM fit to the co-terminal swaption strip, see Algorithm~\ref{algo:HWCalib} in~\ref{app:HWCalibration} but then with valuation from Corollary~\ref{crllry:rHWSwaption} to compute the model value\;
        \ForEach{calibration instrument $\tradeValMkt_{i,j,k}$}{
            \For{$n \in \{1,2,\ldots,\NMix\}$}{
                Compute $\tradeValMdl_{\shortRate_n;i,j,k}$ for $\shortRate_n$ as in Proposition~\ref{prop:HWSwaption} \;
            }
            Compute $\tradeValMdl_{\shortRate;i,j,k}$ using the valuation from Corollary~\ref{crllry:rHWSwaption} using $\tradeValMdl_{\shortRate_n;i,j,k}$ \;
            Convert $\tradeValMdl_{\shortRate;i,j,k}$ to model implied volatility $\impliedVolMdl_{i,j,k}$ (using Black's formula) \;
        }
        Compute error metric $\objFun(\hat{a}, \hat{b}, \vol_{x})$ to assess the model fit to the market data \;
    }
    \caption{rHW calibration algorithm}
    \label{algo:rHWCalib}
\end{algorithm}

During the calibration, one must be cautious with negative mean-reversion ($\theta_n$ following from the quadrature points).
These negative values are bad for extrapolation, as model-implied volatilities will no longer decay with swaption expiry and length but can increase ~\cite{PuetterRenzitti202011b}.
Clearly, as exposures also have a decaying pattern, this decay is desired for $\xva$ calculations.
Hence, one must be cautious when using negative mean-reversion values for the $\shortRate_n(t)$ dynamics.

\subsection{Simulating the rHW dynamics} \label{sec:rHWSimulation}

The short-rate $\shortRate_n(t)$ from Equation~\eqref{eq:HWSDE1Integrated} can be simulated in an exact manner, such that large time steps can be made from time $s$ to $t$.
Ideally, there is a similar result for the rHW dynamics from Corollary~\ref{crllry:rHWDynamics}.
Hence, we integrate to obtain an expression for $\shortRate(t)$ conditional on $\shortRate(s)$ for $s < t$, i.e.,
\begin{align}
  \shortRate(t)
    &= \shortRate(s) + \int_s^t \mu_{\shortRate}^{\Q_{\shortRate}}(u, \shortRate(u)) \du + \sigma_x \int_s^t \dW^{\Q_{\shortRate}}(u), \nonumber \\
  \int_s^t \mu_{\shortRate}^{\Q_{\shortRate}}(u, \shortRate(u)) \du 
    &= f^{\text{M}}(0,t) - f^{\text{M}}(0,s) \nonumber \\
    &\quad + \int_s^t \sum_{n=1}^{\NMix}\left[\theta_n f^{\text{M}}(0,u) - \theta_n \shortRate(u) + \condVarSmall{x_n(u)}{0}\right]\Lambda_n(u,\shortRate(u)) \du.\nonumber 
\end{align}
Due to the state-dependent nature of the integrated drift, it is impossible to make large time steps in an analytic fashion.

Therefore, we use the Euler-Maruyama discretization scheme for the Monte-Carlo simulation, i.e., for $t_{i}<t_{i+1}$
\begin{align}
  \shortRate(t_{i+1})
    &= \shortRate(t_{i}) + \mu_{\shortRate}^{\Q_{\shortRate}}(t_i, \shortRate(t_i)) \Delta t + \eta_{\shortRate}(t, \shortRate(t_i)) \sqrt{\Delta t} Z, \ \ \shortRate(0) = f^{\text{M}}(0, 0),\label{eq:rHWEuler}
\end{align}
where $\Delta t = t_{i+1} - t_{i}$, $Z \sim \N(0,1)$ and with drift and diffusion from Corollary~\ref{crllry:rHWDynamics}.
Here, for the drift and diffusion terms, $\shortRate(u)$ is frozen at the initial value at $t_i$.
Hence, many small time steps are needed to avoid a sizable discretization error and, hence, to get a proper numerical approximation to the SDE from Equation~\eqref{eq:randSDE}.

Alternatively to the Euler-Maruyama scheme, one could attempt to use the predictor-corrector method~\cite{KloedenPlaten199904,HunterJackelJoshi200107}, which is commonly used for efficient simulation of the Libor Market Model through a two-step drift approximation.
Finally, one could use machine learning techniques to `learn' the large time steps offline and then reproduce them online, e.g., using the 7-league scheme~\cite{LiuGrzelakOosterlee202202}.
For the current work, we stick to the Euler-Maruyama scheme.

\begin{rem}[Consistency of calibration and simulation]
    The dynamics used for $\tradeVal_{\shortRate_n}(t;T)$ in the result from Theorem~\ref{thrm:fastPricing} are based on the dynamics from Equation~\eqref{eq:randUnderlyingSDE} with $\M = \Q_{\shortRate_n}$.
    On the other hand, for simulation purposes, we use $\M = \Q_{\shortRate}$ in the drift of Equation~\eqref{eq:randSDE} for $\d\shortRate(t)$ as per Proposition~\ref{prop:randDynamics}.
    The use of different measures raises the question of how $\mu_{\shortRate_n}^{\Q_{\shortRate_n}}$ and $\mu_{\shortRate_n}^{\Q_{\shortRate}}$ are related.
    Since the former is used for calibration and the latter for simulation, consistency between the two is essential.
    Due to the uniqueness of the risk-neutral measure, $\dW^{\Q_{\shortRate_n}}(t) = \dW^{\Q_{\shortRate}}(t)$, such that $\mu_{\shortRate_n}^{\Q_{\shortRate_n}} = \mu_{\shortRate_n}^{\Q_{\shortRate}}$.
    Hence, the calibration and simulation are consistent.
\end{rem}

\subsection{Generic derivative pricing} \label{sec:rHWGenericPricing}
To price derivatives and exposures under the rHW dynamics, we use a Monte-Carlo framework using the simulation scheme from Section~\ref{sec:rHWSimulation}.
We simulate $M_{\tradeVal}$ paths of the dynamics as per Equation~\eqref{eq:rHWEuler}.
Then, the goal is to value along all these simulated trajectories, for which (future) ZCB values are required to compute the (future) cash flows.
Let $\zcb_{\shortRate}(t,T; x)$ denote the ZCB conditional on reaching $\shortRate(t)=x$, i.e., 
\begin{align}
  \zcb_{\shortRate}(t,T; x)
     &= \E\left[\left.\expPower{- \int_{t}^{T} \shortRate(s) \ds} \right| \shortRate(t) = x\right]. \label{eq:zcbMC1}
\end{align}

For the HW model we can compute $\zcb_{\shortRate_n}(t,T; x)$ analytically conditional on reaching future state $\shortRate_n(t)=x$, $t\geq 0$.
However, for the rHW model, this is not the case, and $\zcb_{\shortRate}(t,T; x)$ needs to be obtained in a different way.
As can be seen from Equation~\eqref{eq:zcbMC1}, the ZCB is a conditional expectation, which could be naively approximated using a nested Monte-Carlo simulation.
To avoid this nested simulation, we propose to do a separate and independent Monte-Carlo simulation with $M_{\zcb}$ paths to obtain a functional approximation of $\zcb_{\shortRate}(t,T;x)$ based on the cross-sectional information of $\shortRate(t)$ at $t>0$ through a least-squares regression.
These regression-based methods lend themselves naturally for $\xva$ calculations, also known as American Monte Carlo~\cite{Green201511}.
For example, an $n$-th order polynomial can be used to regress $\zcb_{\shortRate}(t,T;x)$ on $\shortRate(t)=x$, i.e., find $\beta^{\top} = [\beta_0, \ldots, \beta_n]$ such that
\begin{align}
  \zcb_{\shortRate}(t,T; x)
    &\approx \beta_{0} + \beta_{1} \cdot x + \beta_{2} \cdot x^2 + \cdots + \beta_{n} \cdot x^n
    \rdef \beta^{\top} \cdot \psi(x), \nonumber
\end{align}
where along the path, we can approximate the integral inside the expectation in Equation~\eqref{eq:zcbMC1} numerically using trapezoidal integration.
Given the calibrated ZCB approximations, the pricing of derivatives in a Monte-Carlo framework is as fast as for the HW case where $\zcb_{\shortRate_n}(t,T;x)$ can be obtained analytically conditional on reaching $\shortRate_n(t)=x$, as the evaluation of the regression polynomial is cheap.

For a low number of paths $M_{\zcb}$, the path-wise approximation of the ZCB as input for the regression will likely result in a substantial numerical error.
In this case, we do not cover the full spectrum of $\shortRate(t)$ and consequently only a handful of $\expPower{- \int_{t}^{T} \shortRate(s) \ds} | \shortRate(t) = x$, such that the approximation error in
\begin{align}
  \zcb_{\shortRate}(t,T;x)
     &\approx \left.\expPower{- \int_{t}^{T} \shortRate(s) \ds}\right| \shortRate(t) = x \label{eq:zcbMC2}
\end{align}
is significant due to high variance.
To deal with this, one can use a nested simulation with a low number of nested paths $M_{\text{nested}}\in[10,30]$ as a variance reduction technique~\cite{KrepkiyLakhanyZhang202105}.
Alternatively, use a large number of paths $M_{\zcb}$ to cover a larger spectrum of simulated $\shortRate(t)$, which we use as a regression basis.
We choose to do the latter.

This generic valuation framework can be used to price derivatives in a Monte-Carlo setting, but also extends naturally to compute exposure profiles of derivatives: 
\begin{enumerate}
  \item Identify all the dates $T_k$ for which ZCBs $\zcb_{\shortRate}(t,T_k)$ are required to value the derivative(s) in question in a Monte-Carlo setting.
  \item Identify all important dates $0 = t_0 < t_1 < \ldots < t_n = T$ at which ZCBs are required to value the derivative or exposure profile.
  \item Run a Monte-Carlo simulation with a large number of paths $M_{\zcb}$.
  Include at least all the aforementioned important dates where at each important date $t_i$ and each date $T_k > t_i$ a least-squares regression is done to obtain a polynomial approximation of $\zcb_{\shortRate}(t_i,T_k; x)$.~\footnote{If this amount of dates becomes too large, one can pick a subset of the relevant dates, and interpolate over these spines as done for a regular yield curve that is calibrated to the market at $t=t_0$.}
  \item Run a new, independent Monte-Carlo simulation with $M_{\tradeVal}$ paths, including at least all the aforementioned important dates.
  At each date $t_i$ for which ZCBs are required for derivative valuation, use the regression from the previous step.
\end{enumerate}

\section{Numerical results} \label{sec:results}

We demonstrate the ability of the rHW model introduced in Section~\ref{sec:rHW} to capture the market-implied smile and skew and show the effect on exposures and $\xva$s of swaps and Bermudan swaptions.
Remarkably, whereas in other asset classes the main impact of smile and skew was on tail exposures, we demonstrate that the effect on average exposures and $\xva$ metrics of all IR derivatives considered in this paper are significant.
Furthermore, we will see that the smile impact will be more substantial when considering Bermudan swaptions compared to the linear IR swap derivative.

As market data, we use market implied volatility surfaces as of 02/12/2022 and the corresponding market yield curve, see Figure~\ref{fig:VolSurfaceCoterminal}.
The surfaces show smile on the short end, which transforms into skew over time.
The smile and skew are most significant for the USD data.
Therefore, we only present the USD results in this section.
For EUR results, see~\ref{app:resultsEUR}.
The conclusions that are drawn in this section extend to those results.
\begin{figure}[ht!]
  \centering
  \begin{subfigure}[b]{\resultFigureSize}
    \includegraphics[width=\linewidth]{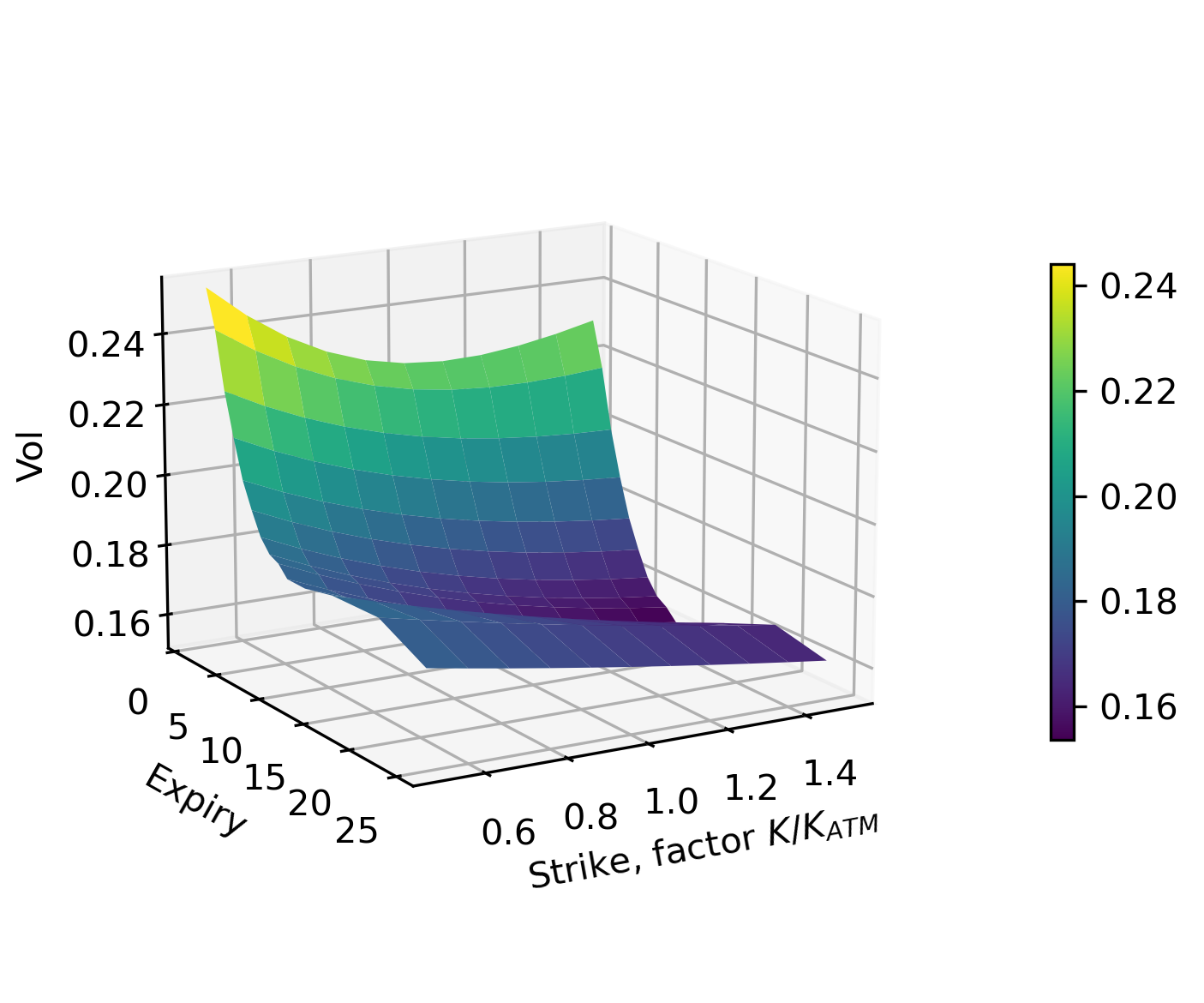}
    \caption{EUR surface.}
    \label{fig:EURVolSurfaceTenor10Y}
  \end{subfigure}
  \begin{subfigure}[b]{\resultFigureSize}
    \includegraphics[width=\linewidth]{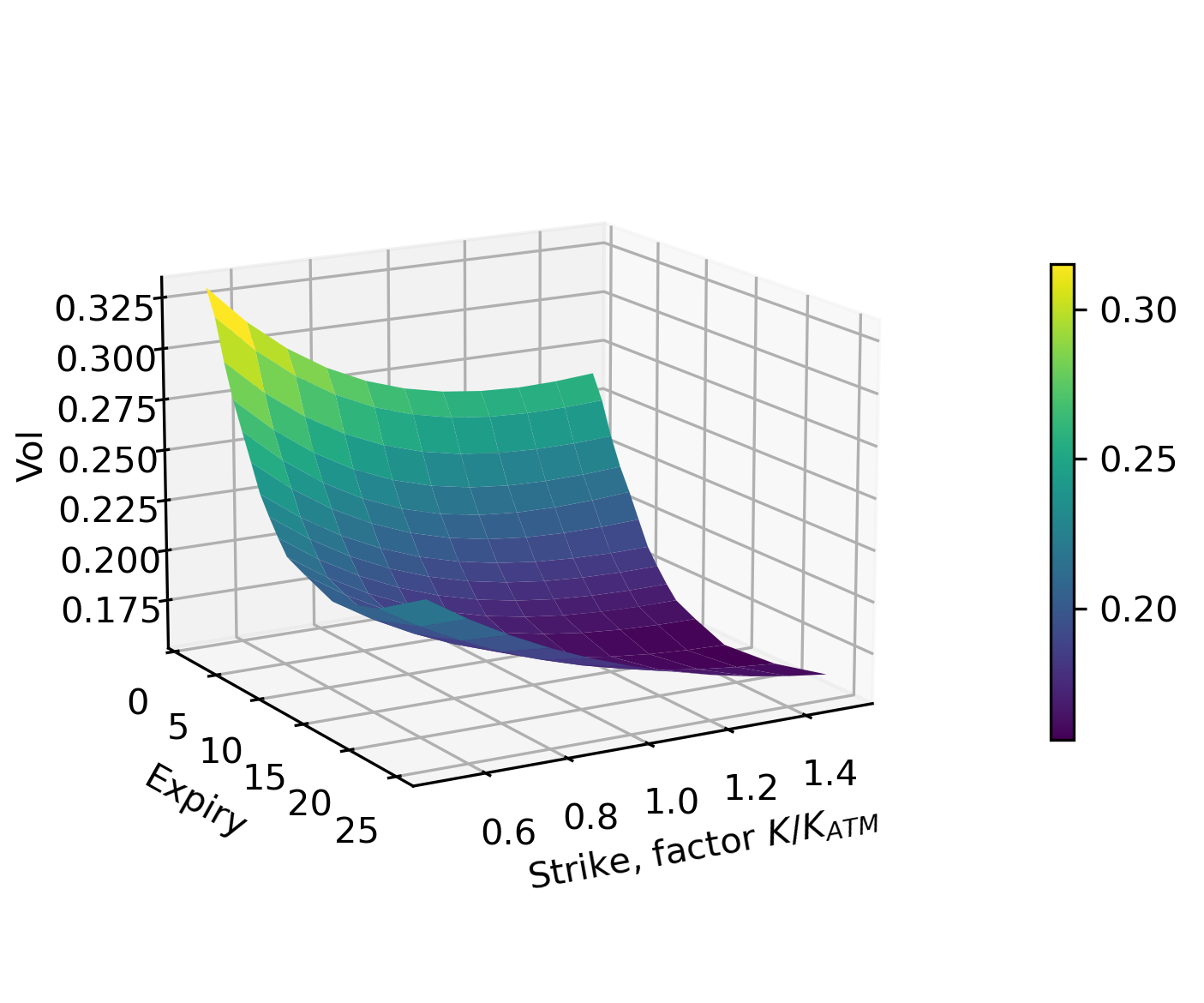}
    \caption{USD surface.}
    \label{fig:USDVolSurfaceTenor10Y}
  \end{subfigure}
  \caption{
  Swaption implied volatility surfaces with a 10Y tenor.
  As both surfaces are constructed using a different shift (EUR uses a shift of $3$, USD a shift of $1$), the surfaces are not directly comparable, especially in terms of level.
  The strike is given as a factor times the at-the-money (ATM) strike $\strikeATM$, e.g., 1.2 means a strike of $1.2\cdot \strikeATM$.
  This ATM strike is different for every tenor-expiry combination.
  }
  \label{fig:VolSurfaceCoterminal}
\end{figure}

In the results that follow, $\NMix=5$ quadrature points are used.
This number has shown to be appropriate when calibrating to market-implied swaption smiles for multiple tenors~\cite{Grzelak202211}.
Mind that different values for $\NMix$ represent different models, not a worse/better approximation, i.e., we are not affected by the convergence to the continuous case as in Equation~\eqref{eq:fastPricingContinuousApprox}.
The lower the value of $\NMix$, the faster the calibration is.
Furthermore, the simulation speed will also be slightly affected as the computational time of the state-dependent drift increases in $\NMix$, see Corollary~\ref{crllry:rHWDynamics}.
The choice of $\NMix=5$ is sufficiently large to have enough flexibility in the model while remaining sufficiently low for a suitable computational speed.

In Section~\ref{sec:resultsCalibration}, the rHW model is calibrated to the strip of co-terminal swaptions.
Given this calibrated model, we simulate example paths in Section~\ref{sec:resultsSimulation} and compare the resulting trajectories and distributions with the HW model. 
In Section~\ref{sec:resultsZCBRegression}, we demonstrate that pricing derivatives with a regression-based ZCB expression is appropriate by showing this for the HW case and comparing it to the analytic benchmark available for this model.
Finally, in Section~\ref{sec:resultsExposures}, we assess the effect of smile and skew on exposure profiles and $\xva$s of IR swaps and Bermudan swaptions using the rHW model.

\subsection{Calibration} \label{sec:resultsCalibration}

We calibrate the rHW model to co-terminal swaptions with $\TCot = 30y$, using all strikes for each expiry-tenor combination, through the calibration approach described in Section~\ref{sec:rHWCalibration}.
For comparison, we calibrate the HW model to the ATM point of the same co-terminal swaptions.
The HW parameters $a_x$ and $\vol_x$ have opposite effects on the model-implied swaption volatility, i.e., for larger values of $a_x$, the implied volatility is lower, whereas the implied volatility is higher for larger values of $\vol_x$.
The calibrated model parameters are presented in Figure~\ref{fig:calibrationUSD}.
The rHW mean reversions $\theta_n$ at the quadrature points $\left\{\omega_n,\theta_n\right\}_{n=1}^{\NMix}$ are significantly larger than that of the HW model, and therefore the same holds for the model volatilities.
The high model volatility values can be explained in a similar way as the opposite effects of the HW mean-reversion and model-volatility parameters on implied volatilities: high mean-reversion values are used in the distribution, so the scaling of the diffusion part needs to be increased to generate sufficient randomness with the model (achieved by increasing the model volatility).

\begin{figure}[ht!]
  \centering
  \begin{subfigure}[b]{0.435\linewidth}
    \includegraphics[width=\linewidth]{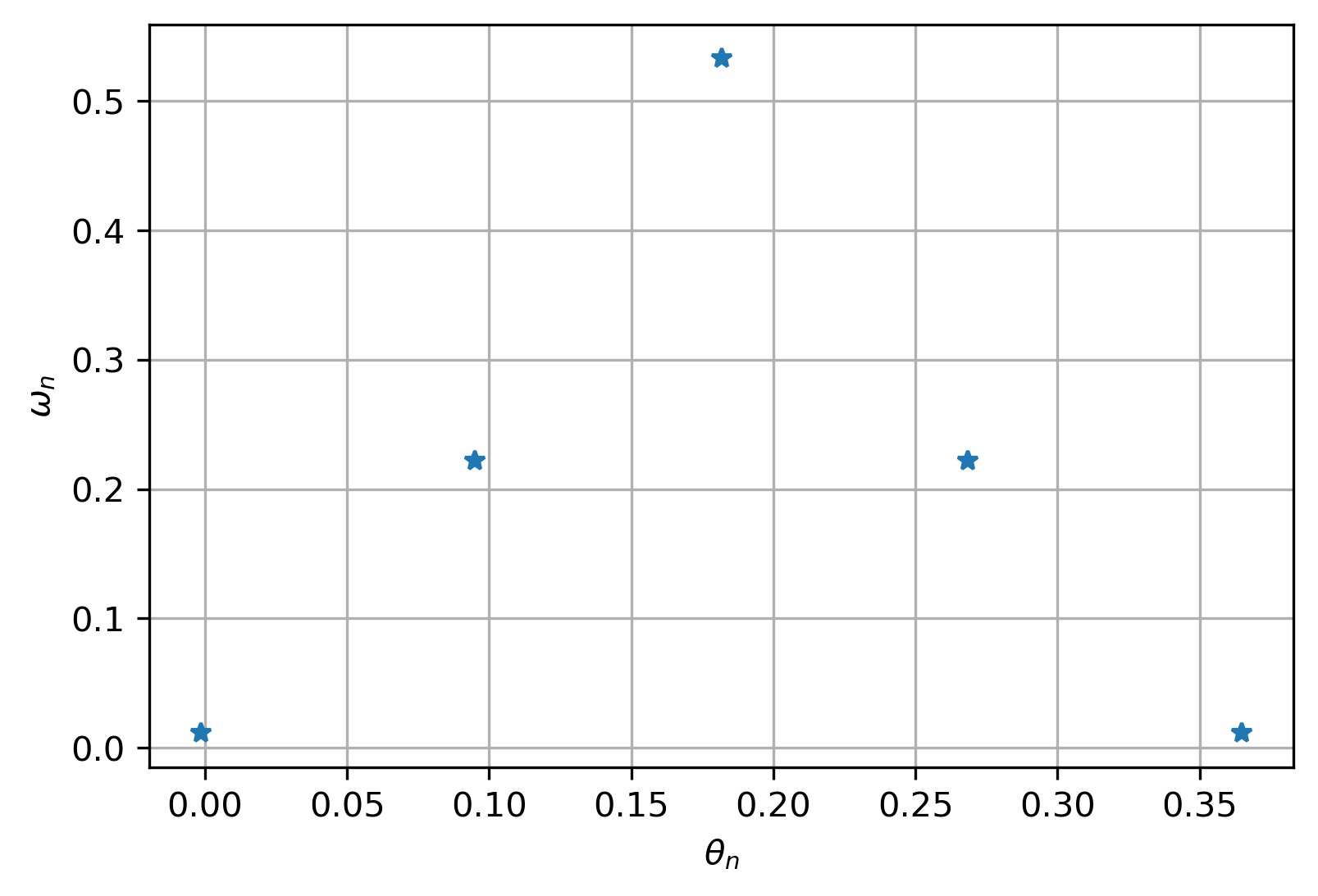}
    \caption{Quadrature points for rHW.}
    \label{fig:quadraturePointsUSD}
  \end{subfigure}
  \begin{subfigure}[b]{0.45\linewidth}
    \includegraphics[width=\linewidth]{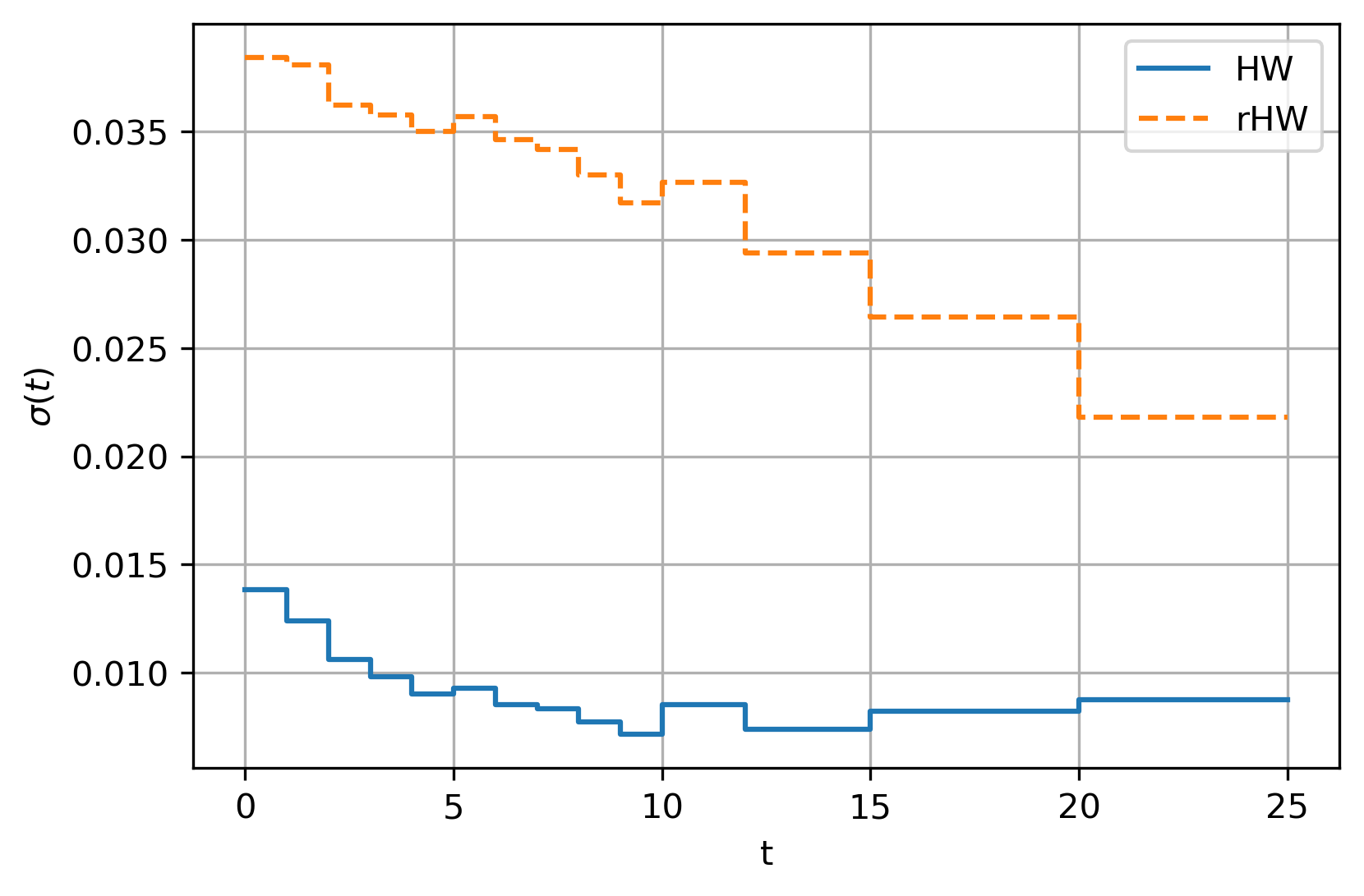}
    \caption{Model volatilities.}
    \label{fig:modelVolatilitiesUSD}
  \end{subfigure}
  \caption{Calibration of the USD market data from Figure~\ref{fig:VolSurfaceCoterminal} for $\NMix=5$ and quadrature points $\left\{\omega_n,\theta_n\right\}_{n=1}^{\NMix}$ for $\N (\hat{a}, \hat{b}^2 )$ with $\hat{a} = 0.181711$ and $\hat{b} = 0.064055$.
    The calibrated HW mean-reversion is $a_x = 0.030228$.}
  \label{fig:calibrationUSD}
\end{figure}

From the implied swaption volatilities in Figure~\ref{fig:USDVolSlicesNormalRandCotsmilesCalib}, it can be seen that the rHW model is sufficiently flexible to be calibrated well to the smile and skew of the set of co-terminal swaptions.
Clearly, the rHW model outperforms the HW model, which only has a good fit at $\strikeATM$, and the model skew cannot be controlled.
\begin{figure}[ht!]
  \centering
  \begin{subfigure}[b]{0.45\linewidth}
    \includegraphics[width=\linewidth]{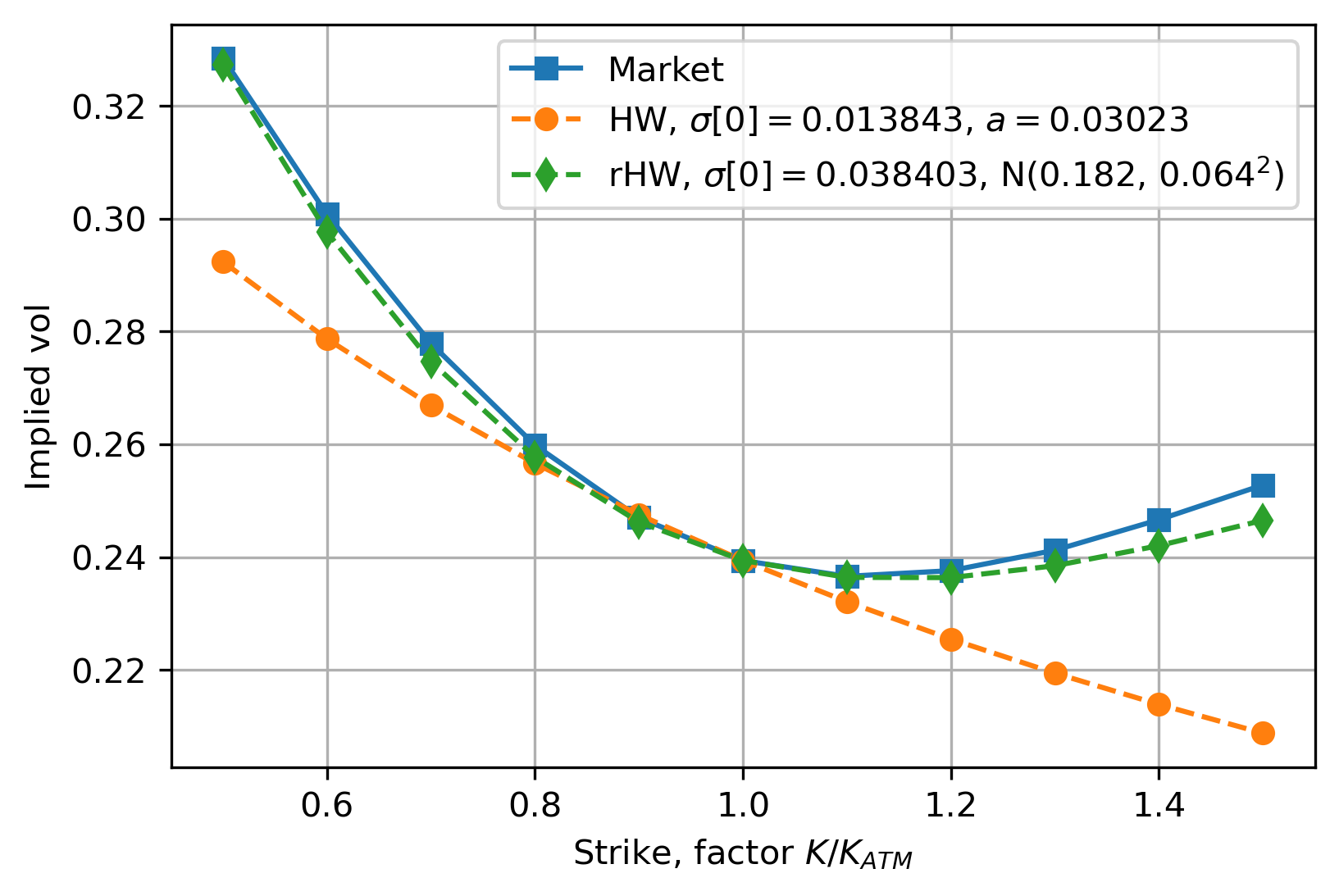}
    \caption{1Y expiry, 29Y tenor.}
    \label{fig:USD1yExpiry29yTenorVolSliceNormalRandCotsmilesCalib}
  \end{subfigure}
  \begin{subfigure}[b]{0.45\linewidth}
    \includegraphics[width=\linewidth]{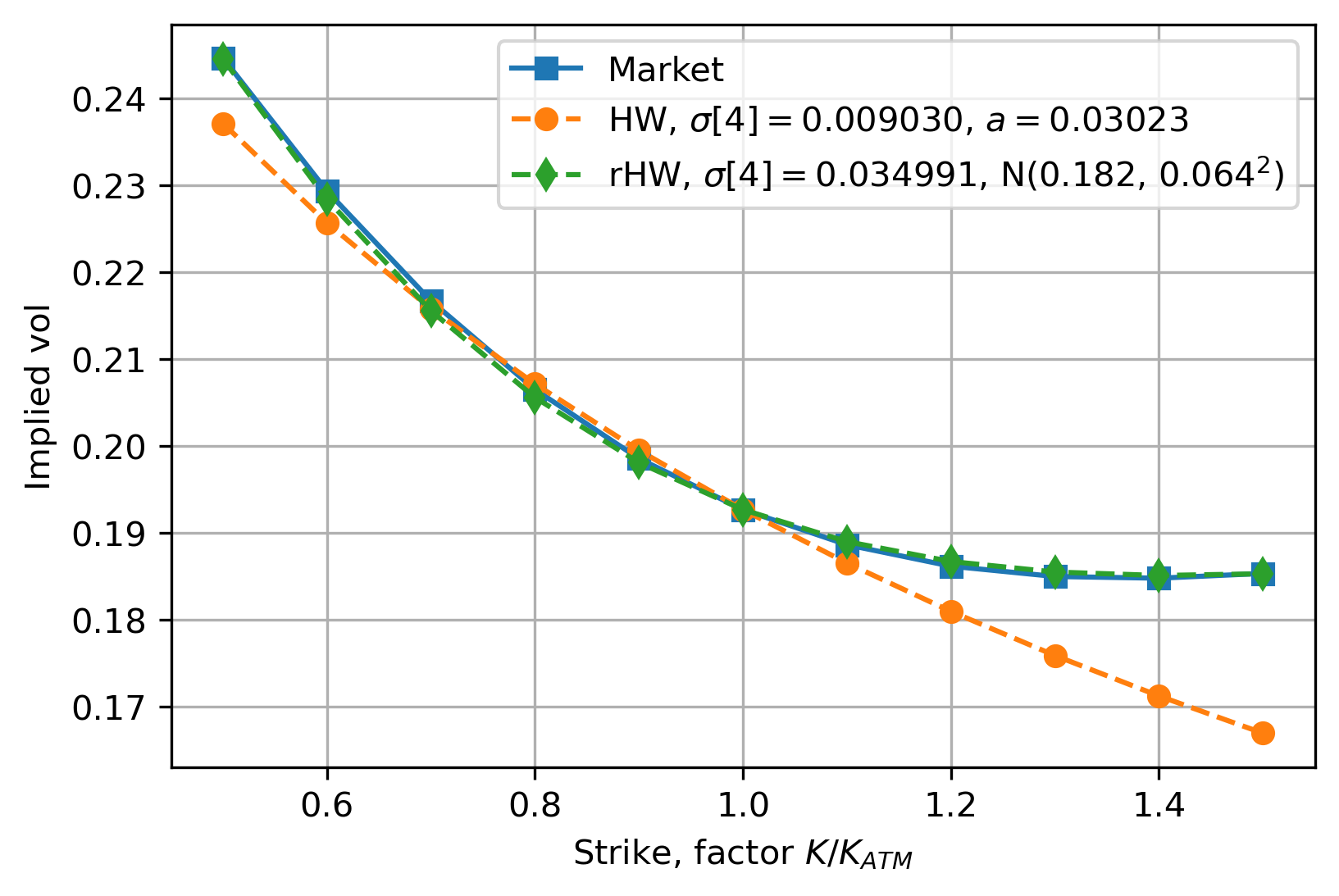}
    \caption{5Y expiry, 25Y tenor.}
    \label{fig:USD5yExpiry25yTenorVolSliceNormalRandCotsmilesCalib}
  \end{subfigure}
  \begin{subfigure}[b]{0.45\linewidth}
    \includegraphics[width=\linewidth]{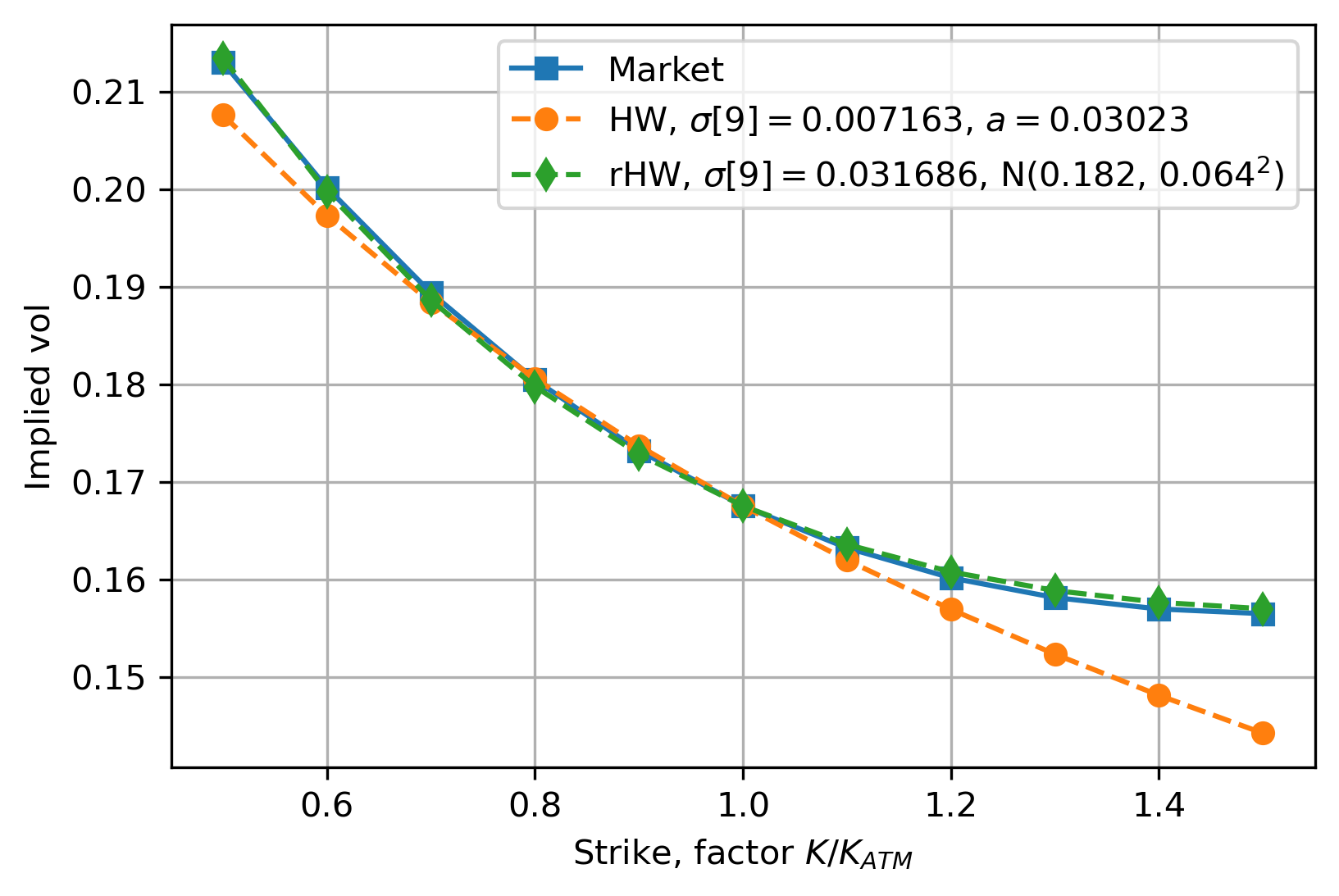}
    \caption{10Y expiry, 20Y tenor.}
    \label{fig:USD10yExpiry20yTenorVolSliceNormalRandCotsmilesCalib}
  \end{subfigure}
  \begin{subfigure}[b]{0.45\linewidth}
    \includegraphics[width=\linewidth]{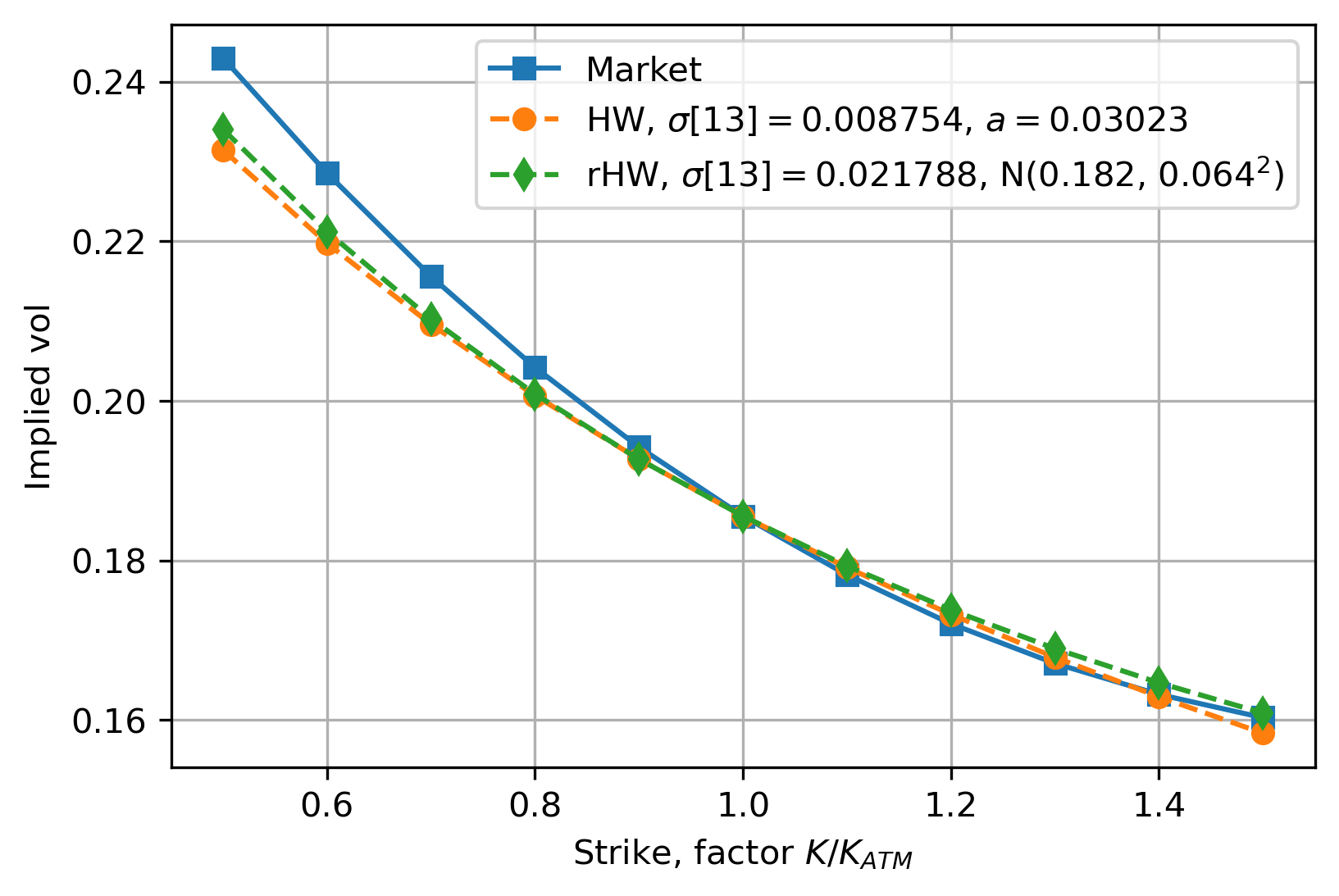}
    \caption{25Y expiry, 5Y tenor.}
    \label{fig:USD25yExpiry5yTenorVolSliceNormalRandCotsmilesCalib}
  \end{subfigure}
  \caption{Market and model swaption implied volatilities for the calibrated parameters from Figure~\ref{fig:calibrationUSD}.}
  \label{fig:USDVolSlicesNormalRandCotsmilesCalib}
\end{figure}

For the non-co-terminal expiry-tenor combinations, the rHW ATM model volatilities are further away from the market volatilities than for the HW model, see Figure~\ref{fig:USDNormalRandCotsmilesCalibImpvolError}.
The good fit for these points with the HW model results from the mean-reversion calibration, which takes into account all the ATM quotes for the entire volatility cube and minimizes the Mean Squared Error for all these points.
If a good ATM fit is required for all expiry-tenor pairs of the volatility surface, the error metric in the calibration procedure from Section~\ref{sec:rHWCalibration} can be updated accordingly by taking into account these extra points.
However, since our numerical results focus on the co-terminal strip, the calibration procedure from Section~\ref{sec:rHWCalibration} is suitable.

\begin{figure}[ht!]
  \centering
  \begin{subfigure}[b]{0.45\linewidth}
    \includegraphics[width=\linewidth]{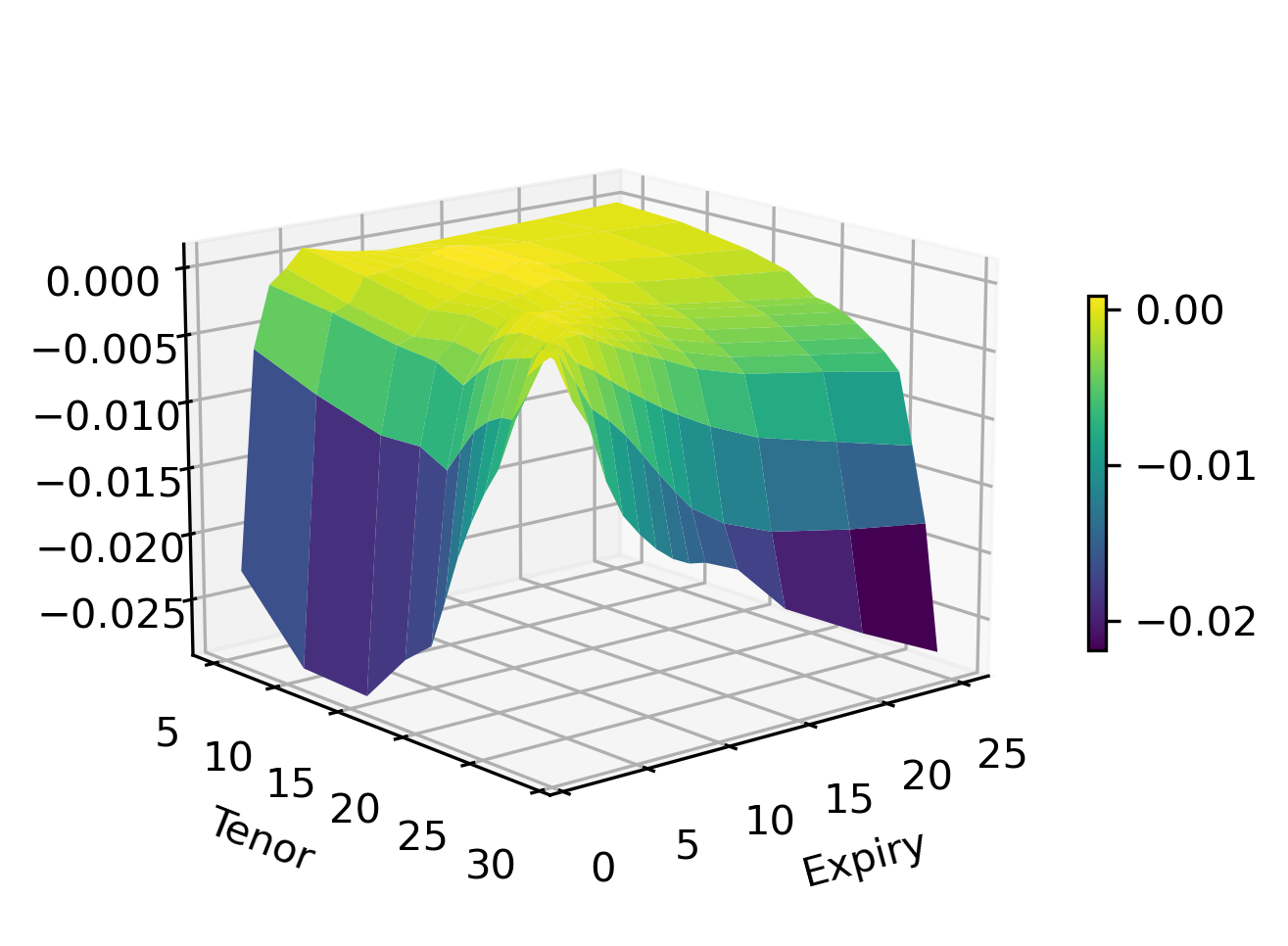}
    \caption{HW model.}
    \label{fig:USDNormalRandCotsmilesCalibImpvolErrorHW}
  \end{subfigure}
  \begin{subfigure}[b]{0.45\linewidth}
    \includegraphics[width=\linewidth]{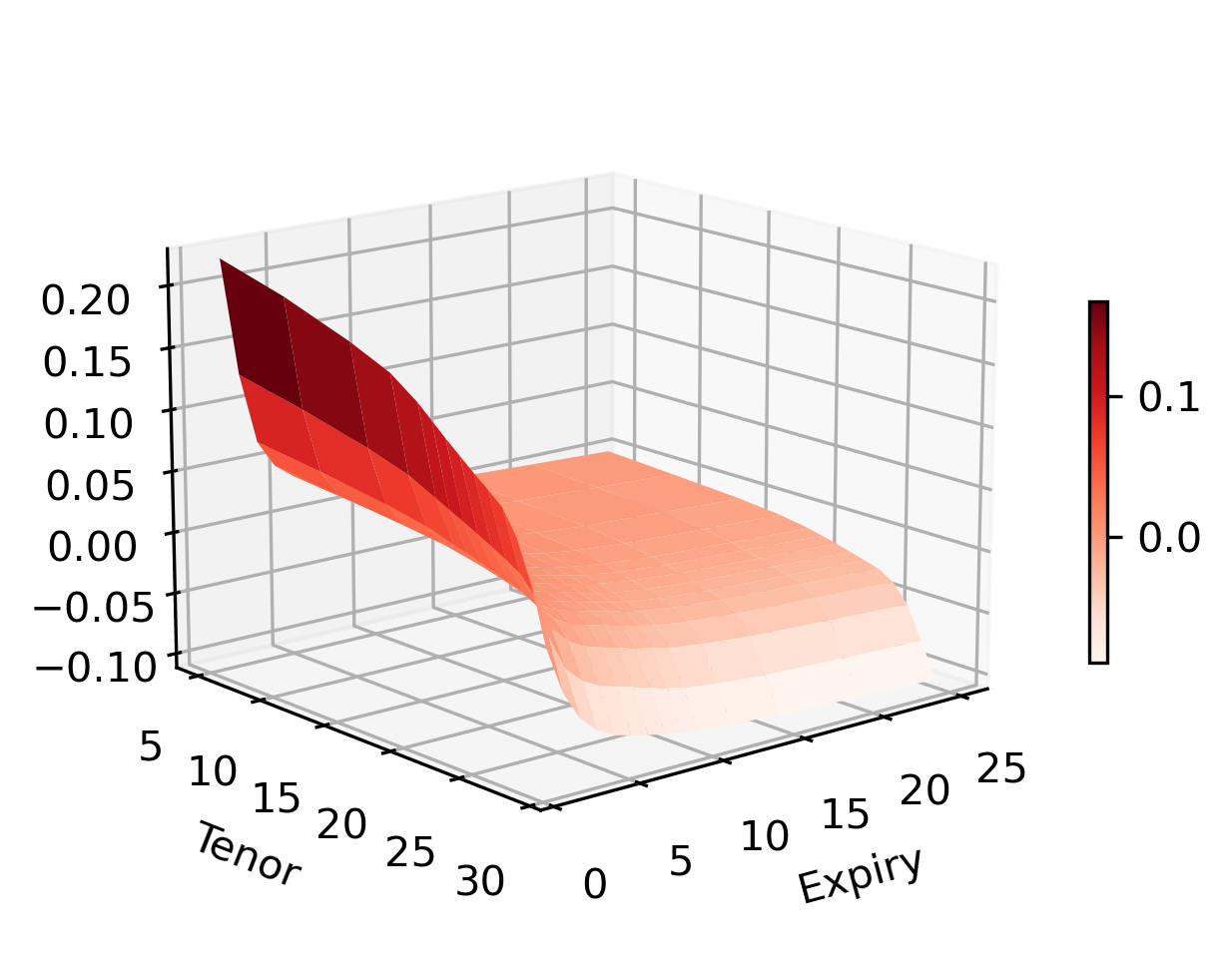}
    \caption{rHW model.}
    \label{fig:USDNormalRandCotsmilesCalibImpvolErrorrHW}
  \end{subfigure}
    \caption{Implied volatility calibration error for all ATM points.}
  \label{fig:USDNormalRandCotsmilesCalibImpvolError}
\end{figure}

In the results presented so far, we have chosen to calibrate to all co-terminal smiles simultaneously.
As a result, the model fit to the first co-terminal smile is not perfect, see Figure~\ref{fig:USD1yExpiry29yTenorVolSliceNormalRandCotsmilesCalib}.
However, the rHW model is sufficiently flexible to allow for a perfect fit when merely calibrating to this smile.
See Figure~\ref{fig:USDVolSlicesNormalRandInitSmileCalib} for an example.
\begin{figure}[ht!]
  \centering
  \includegraphics[width=0.45\linewidth]{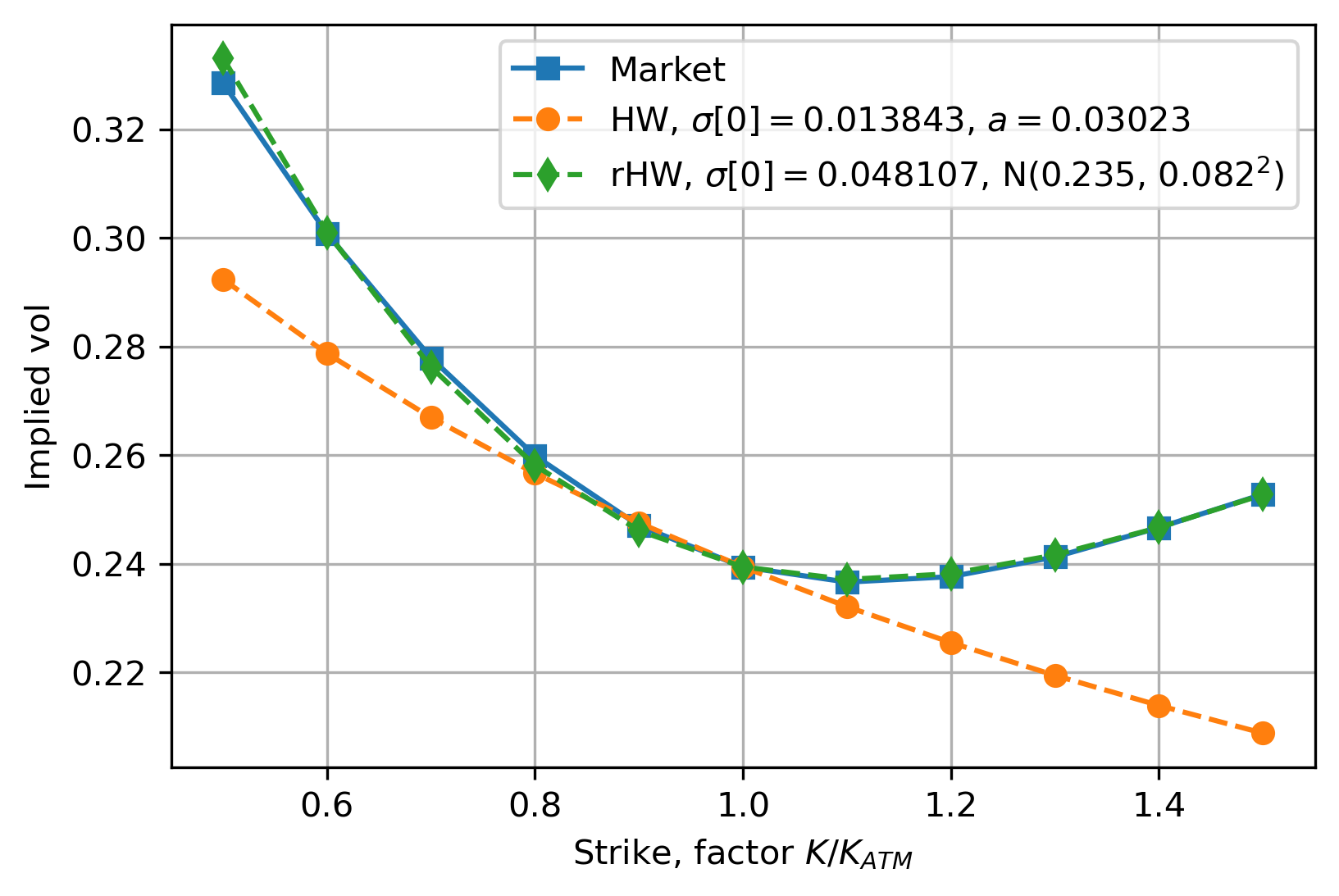}
  \caption{USD calibration to the initial market smile.}
  \label{fig:USDVolSlicesNormalRandInitSmileCalib}
\end{figure}

Mind that this affects the fit for the other co-terminal strips, as well as the fit for other instruments of the volatility surface, as in Figure~\ref{fig:USDNormalRandCotsmilesCalibImpvolError}.
Hence, it is a tradeoff, and one must select which market data is most relevant for a particular application.
For $\xva$ purposes, a proper fit to all co-terminal smiles is desired.

In conclusion, whereas the HW model can only be fitted to the ATM point of all co-terminal smiles and the model skew cannot be controlled; the rHW skew and smile can be controlled to fit the market-observed implied volatilities for all strikes of the co-terminal strips.
However, one needs to be aware of the quality of the fit for non-co-terminal ATM points.

\subsection{Simulation} \label{sec:resultsSimulation}

Given the calibrated dynamics as presented in Section~\ref{sec:resultsCalibration}, and using the simulation as described in Section~\ref{sec:rHWSimulation}, we present example rHW paths and compare them with the paths of the underlying models.
Since all dynamics are driven by the same source of randomness, see Section~\ref{sec:randGeneral}, the same random samples must be used for all simulated paths.

Figure~\ref{fig:RAnDUSDNormalCotsmilesDistributionPathsComparisonNormal} shows that a typical $\shortRate(t)$ path is similar to the $\shortRate_2(t)$ path.
Recall that the quadrature points are given in Figure~\ref{fig:quadraturePointsUSD}, where $\shortRate_2(t)$ has the largest weight of all underlying processes.
Hence, it makes sense that the $\shortRate(t)$ paths are similar to this $\shortRate_2(t)$ path.
However, if we consider a $\shortRate(t)$ path that ends at a high value, see Figure~\ref{fig:RAnDUSDNormalCotsmilesDistributionPathsComparisonHigh} for an example, the $\shortRate(t)$ trajectory looks much more like that of $\shortRate_0(t)$, which has a very low mean-reversion parameter, such that the short-rate path is quite extreme.
So, in general, the $\shortRate(t)$ paths look a lot like the highly mean-reverting paths of $\shortRate_2(t)$.
However, in some extreme cases, we diverge from this, such that the tails of the $\shortRate(t)$ distribution should be fatter than that of $\shortRate_2(t)$. 

\begin{figure}[ht!]
  \centering
  \begin{subfigure}[b]{0.45\linewidth}
    \includegraphics[width=\linewidth]{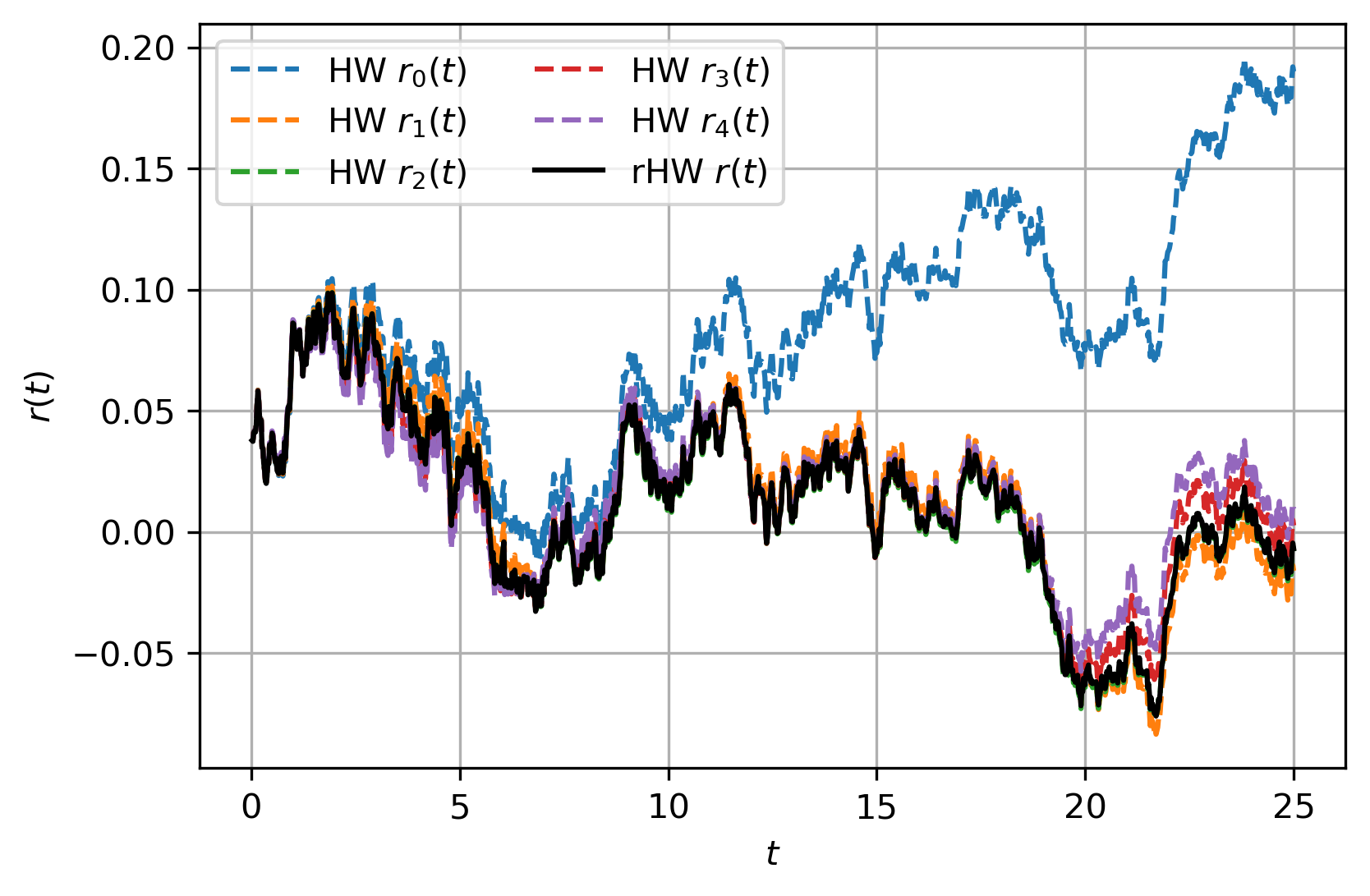}
    \caption{Example path.}
    \label{fig:RAnDUSDNormalCotsmilesDistributionPathsComparisonNormal}
  \end{subfigure}
  \begin{subfigure}[b]{0.435\linewidth}
    \includegraphics[width=\linewidth]{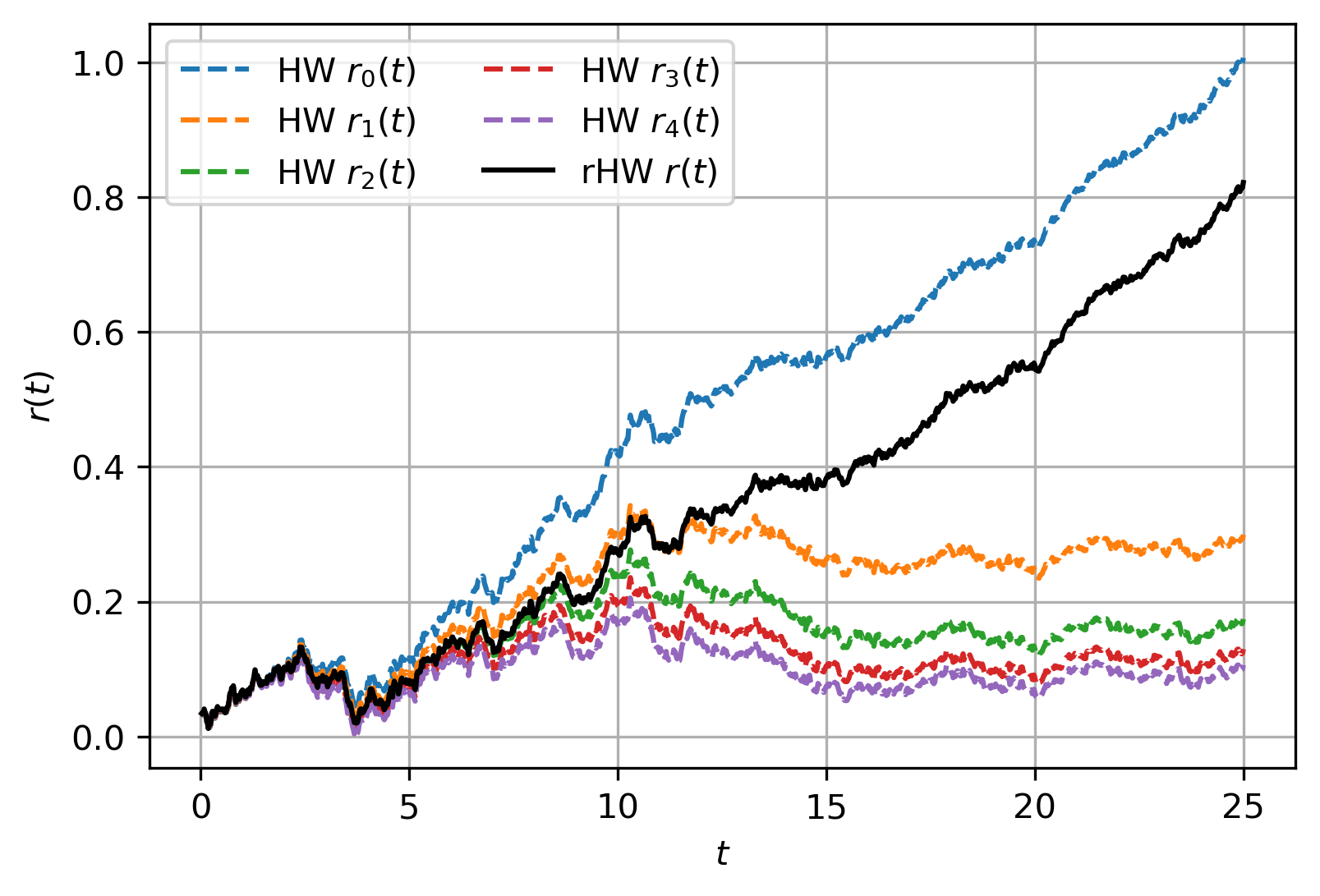}
    \caption{Example path ending high.}
    \label{fig:RAnDUSDNormalCotsmilesDistributionPathsComparisonHigh}
  \end{subfigure}
  \caption{Comparing the paths of rHW $\shortRate(t)$ with all HW processes $\shortRate_n(t)$.}
  \label{fig:RAnDUSDNormalCotsmilesDistributionPathsComparison}
\end{figure}

Looking at the rHW PDF and CDF at $t=25$ in Figure~\ref{fig:RAnDUSDNormalCotsmilesDistributionHWDistrCheck2} and comparing them with the HW process $\shortRate_2(t)$, the right tail of $\shortRate(t)$ is significantly fatter than the normal distribution of $\shortRate_2(t)$.
This observation is in line with the fat tails implied by the market smile and skew.

\begin{figure}[ht!]
  \centering
  \begin{subfigure}[b]{0.58\linewidth}
    \includegraphics[width=\linewidth]{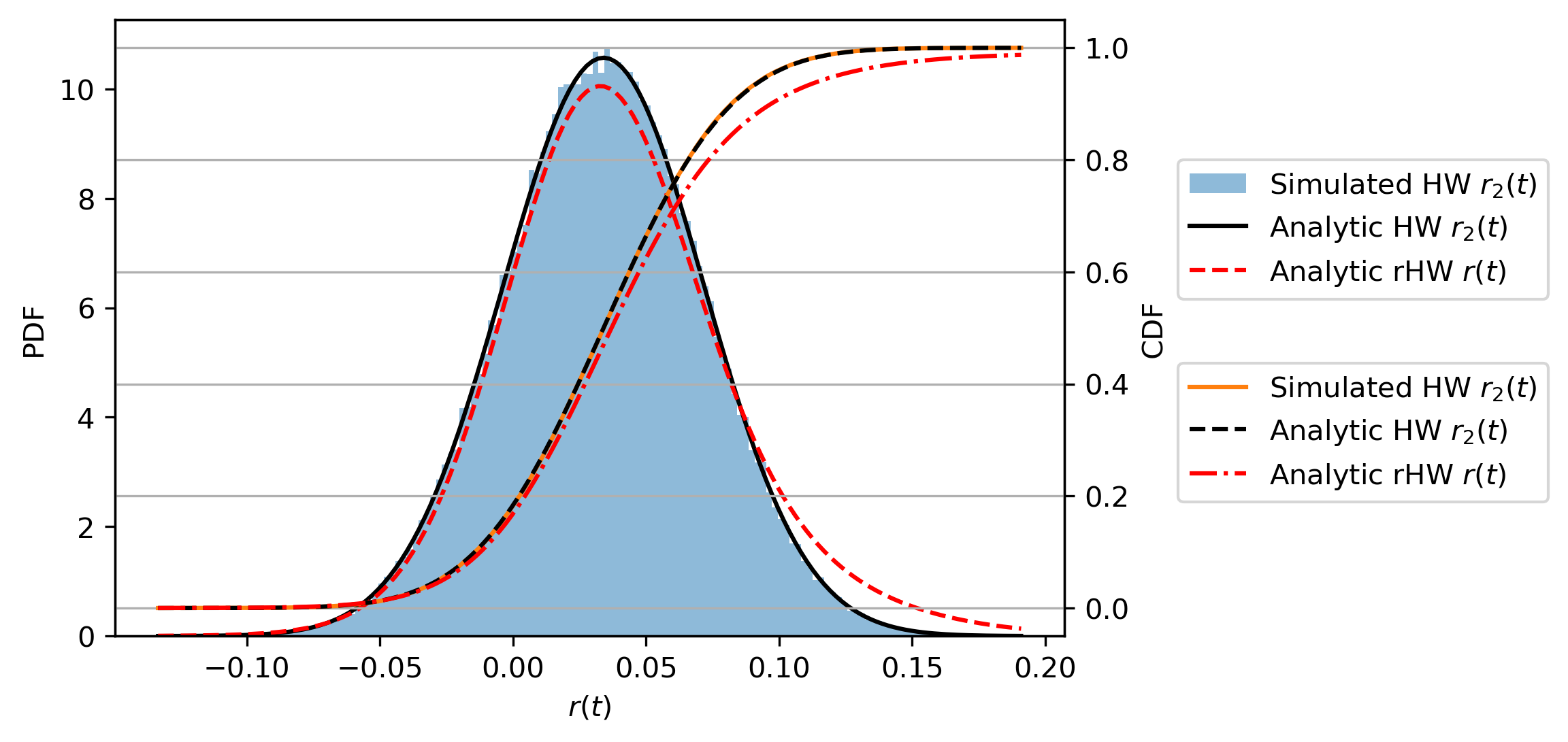}
    \caption{CDF and PDF.}
    \label{fig:RAnDUSDNormalCotsmilesDistributionHWDistrCheck2}
  \end{subfigure}
  \begin{subfigure}[b]{0.41\linewidth}
    \includegraphics[width=\linewidth]{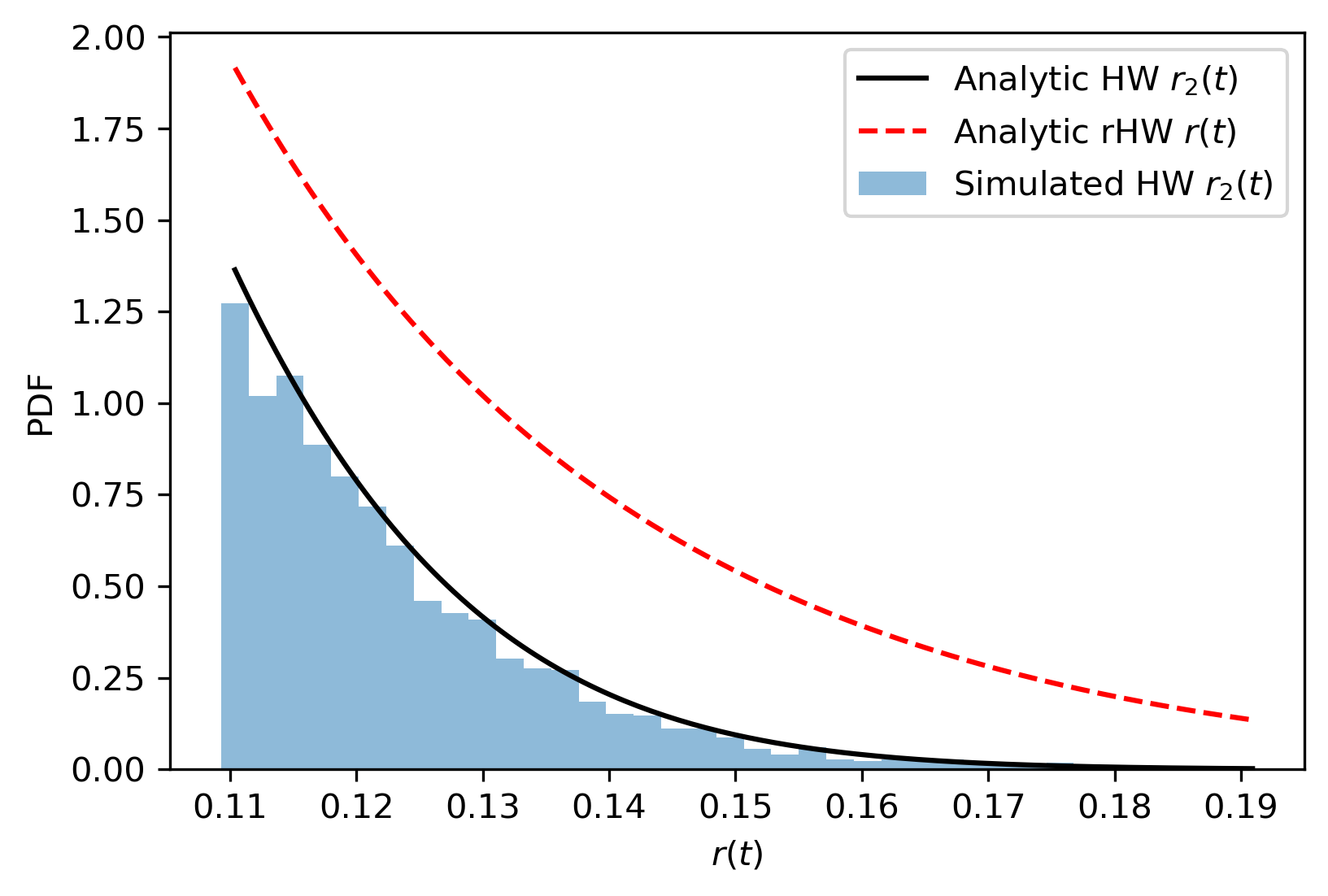}
    \caption{Right tail of the PDF.}
    \label{fig:RAnDUSDNormalCotsmilesDistributionHWRightTailCheck2}
  \end{subfigure}
  \caption{Comparing the rHW PDF and CDF at $t=25$ with the HW process $\shortRate_2(t)$.}
  \label{fig:RAnDUSDNormalCotsmilesDistributionComparison}
\end{figure}

\begin{rem}[Parameter intuition]
    The parameters $\hat{a}$ and $\hat{b}$ do not have a clear financial meaning, which is not a problem, as these parameters are merely a modelling aspect to introduce additional degrees of freedom to calibrate to the market smile and skew.
    In other words, a clear financial meaning of these parameters is not required as long as there is a fit to the market.
\end{rem}

\subsection{ZCB regression} \label{sec:resultsZCBRegression}
    
We demonstrate that pricing derivatives with a regression-based ZCB expression, as introduced in Section~\ref{sec:rHWGenericPricing}, is appropriate.
For the HW model, we have several (semi-)analytic results which can serve as a benchmark.
First of all, the HW model allows for an analytic ZCB expression, see Proposition~\ref{prop:HWZCB}.
Based on this result, we derive a semi-analytic swaption pricing formula in Proposition~\ref{prop:HWSwaption} using Jamshidian's decomposition.
Alternatively, we use a Monte-Carlo simulation with $M_{\tradeVal}$ paths to approximate the swaption value, where we can use either the analytic ZCB expression or a regression-based ZCB that has been calibrated using an independent simulation with $M_{\zcb}$ paths.
At the same time, we can construct an instance of the rHW model with $\NMix=1$, $\theta_1 = a_x$ and $\omega_1=1$, such that the rHW model collapses to the HW model.
We can then use a Monte-Carlo swaption valuation, with regression-based ZCBs.
The only difference between the two models is that the HW model allows for exact simulation, whereas the rHW model requires an Euler-Maruyama discretization that results in an additional error, see Section~\ref{sec:rHWSimulation}.

As an example, consider a swaption with expiry $T_M = 5$, on an ATM payer swap with $\notional = 10k$, start $T_0 = T_M = 5$, maturity $T_m = 10$.
For both the ZCB regression simulation and the outer Monte-Carlo simulation, we use various numbers of paths $M_{\tradeVal}=M_{\zcb}$ and 25 simulation dates per year.
We use a 2nd order polynomial for the regression, where the regression variable is the short rate.
The results are presented in Figure~\ref{fig:HW1FUSDSwaptionVariousNrPathsSwaptionComparisonNrPathsRAnDReSim}.
We conclude that the ZCB regression performs well.

To assess the impact of the Euler-Maruyama discretization error w.r.t. the exact simulation, in Figure~\ref{fig:HW1FUSDSwaptionVariousNrDatesSwaptionComparisonHW1FvsRANDReSim}, we present the difference in swaption value between the HW model with regression-based ZCB and the rHW model that collapses to the HW model.
This is done for $M_{\tradeVal}=M_{\zcb}=10^4$ paths with various number of simulation dates per year.
We observe convergence in the number of dates per year.

\begin{figure}[h]
  \centering
  \begin{subfigure}[b]{0.44\linewidth}
    \includegraphics[width=\linewidth]{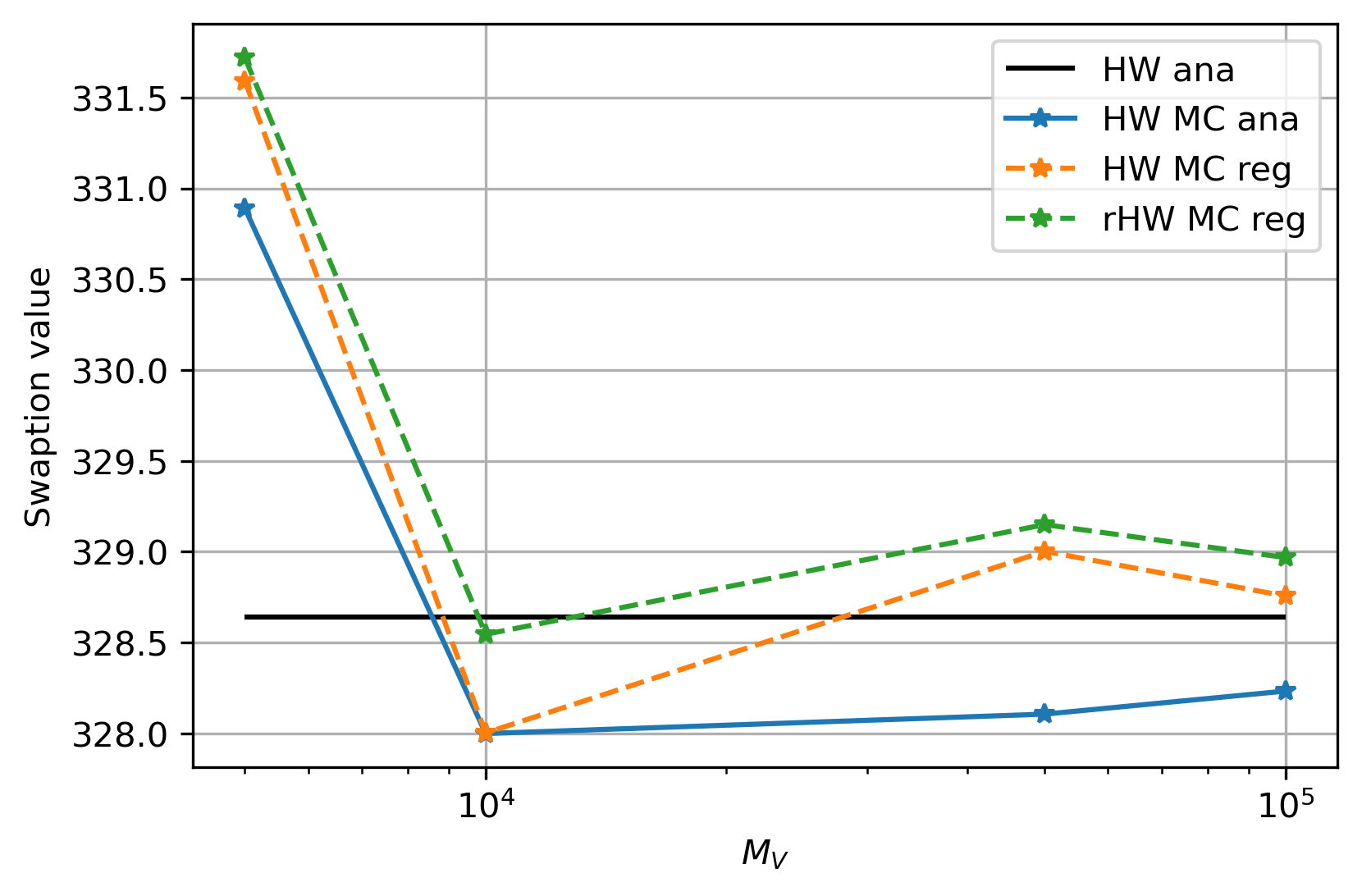}
    \caption{Effect of number of MC paths $M_{\tradeVal}$.}
    \label{fig:HW1FUSDSwaptionVariousNrPathsSwaptionComparisonNrPathsRAnDReSim}
  \end{subfigure}
  \begin{subfigure}[b]{0.43\linewidth}
    \includegraphics[width=\linewidth]{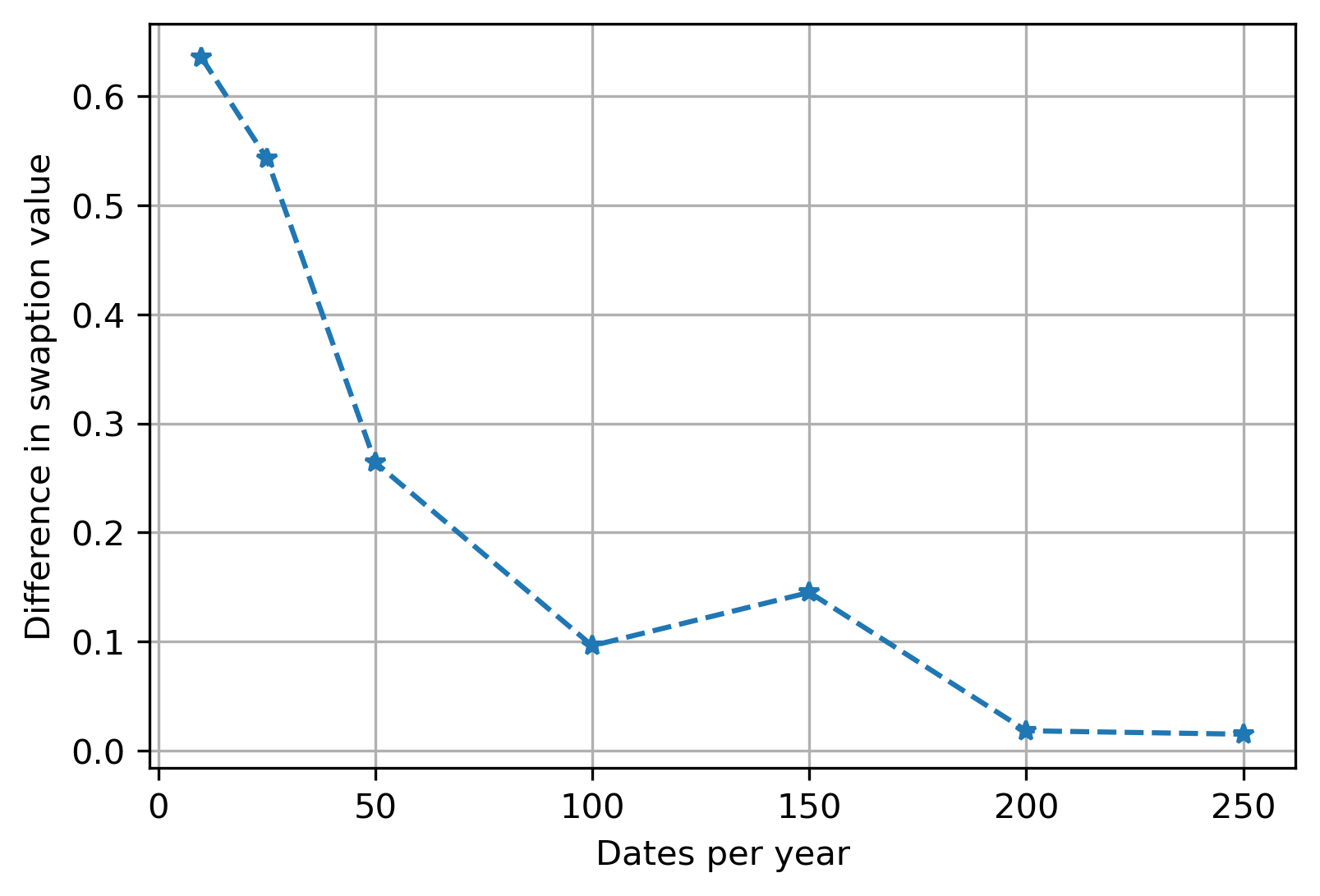}
    \caption{Effect of number of simulation dates.}
    \label{fig:HW1FUSDSwaptionVariousNrDatesSwaptionComparisonHW1FvsRANDReSim}
  \end{subfigure}
  \caption{Assessing the ZCB regression effect and the Euler discretization error.}
  \label{fig:HW1FUSDSwaptionVariousNrPathsSwaptionComparisons}
\end{figure}

\subsection{Exposures and $\xva$s} \label{sec:resultsExposures}

We use the industry standard Monte-Carlo simulation framework to compute exposures and $\xva$ metrics~\cite{Green201511}.
In particular, we build upon the pricing algorithm outlined in Section~\ref{sec:rHWGenericPricing}.
When examining exposures and $\xva$s, we assume that all the relevant ZCB functions at relevant future dates have been obtained through a least-squares regression based on an independent Monte-Carlo simulation.

All $\xva$ metrics are considered from the perspective of institution $I$, and all the trades are done with counterparty $C$.
Consider the $\CVA$ definition under the assumption of no Wrong-Way Risk, which is the weighted sum of Expected Positive Exposures ($\EPE$) at monitoring dates $t_0 < \ldots < t_n$:
\begin{align}
  \CVA(t)
    &\approx (1-\RR) \sum_{i=1}^{n}   \EPE(t, t_i) \left[\PD_C(t_i) - \PD_C(t_{i-1})\right] , \nonumber \\
  \EPE(t, t_i)
    &\ldef \condExpSmall{\expPower{-\int_t^{t_i} \shortRate(v) \dv} \maxOperator{\tradeVal(t_i)}}{t} 
    \approx \frac{1}{M_{\tradeVal}} \sum_{j=1}^{M_{\tradeVal}} \expPower{-\int_t^{t_i} \shortRate_j(v) \dv} \maxOperator{\tradeVal(t_i; \shortRate_j(t_i))},  \label{eq:EPE}
\end{align}
where in the latter step, we denote the Monte-Carlo approximation of the exposure using $M_{\tradeVal}$ paths.~\footnote{For exposures, we only include future payments w.r.t. the monitoring dates (not the payments at the monitoring dates themselves, they are assumed to have taken place already).}
$\CVA$ is computed with recovery rate $\recovRate = 0$ and counterparty $C$'s default probability $\PD_C(t)$ with constant hazard rate $\hazardRate_C = 0.02$.
$\DVA$ is computed similarly but with $\ENE$ rather than $\EPE$, i.e., the max-operator in Equation~\eqref{eq:EPE} is replaced by the min-operator, and using the institution's default probability $\PD_I(t)$ with constant hazard rate $\hazardRate_I = 0.01$.~\footnote{This hazard rate implies that the creditworthiness of institution $I$ is higher than that of counterparty $C$.}:
\begin{align}
  \DVA(t)
    &\approx (1-\RR) \sum_{i=1}^{n}   \ENE(t, t_i) \left[\PD_I(t_i) - \PD_I(t_{i-1})\right]. \nonumber
\end{align}
Finally, $\BCVA$ is computed with a survival adjustment, such that exposures are only considered if the other party did not yet go into default:
\begin{align}
  \BCVA(t)
    &\approx (1-\RR) \sum_{i=1}^{n}   \EPE(t, t_i) \left[\PD_C(t_i) - \PD_C(t_{i-1})\right] \left[1 - \PD_I(t_{i-1})\right] \nonumber \\
    &\quad + (1-\RR) \sum_{i=1}^{n}   \ENE(t, t_i) \left[\PD_I(t_i) - \PD_I(t_{i-1})\right] \left[1 - \PD_C(t_{i-1})\right] . \nonumber
\end{align}

In addition, we consider a tail metric, Potential Future Exposure ($\PFE$), which is the $\alpha$-quantile of the positive exposure $\maxOperator{\tradeVal(t_i)}$, i.e., at time $t$ this metric is defined as
\begin{align}
  \PFE(t, t_i; \alpha)
    &\ldef \inf\left\{ x\in\R : \frac{\alpha}{100} \leq F_{\maxOperator{\tradeVal(t_i)}}(x) \right\},  \label{eq:PFE}
\end{align}
where $F_{\maxOperator{\tradeVal(t_i)}}(x)$ is the CDF of the positive exposure at time $t_i$.
Analogously, the Potential Future Loss ($\PFL$) is the $\alpha$-quantile of the negative exposure $\minOperator{\tradeVal(t_i)}$.

The key to computing (potential future) exposures is to obtain derivative values at monitoring dates $t_i$ for all $j$ Monte-Carlo scenarios, i.e., $\tradeVal(t_i; \shortRate_j(t_i))$, where $\shortRate_j(t_i)$ denotes the $j$-th simulated short-rate $\shortRate$ at time $t_i$.
It is crucial to have an efficient valuation of ZCBs for all relevant dates $T_k > t_i$, i.e., $\zcb_{\shortRate}(t_i,T_k; x)$ conditional on reaching state $x$.
For the HW model, this expression is analytic, see~\ref{app:HWZeroCouponBond}, whereas for the rHW model, we use the polynomial approximation obtained from regression.
Hence, we have an efficient way to evaluate the future ZCBs for both models. 

See Section~\ref{sec:resultsExposuresSwaps} for IR swap exposures and the assessment of skew and smile effects on these exposure profiles and the effects on the corresponding $\xva$ metrics.
In Section~\ref{sec:resultsExposuresBermudanSwaption}, we look at Bermudan swaption exposures and $\xva$s to assess the impact of smile and skew on these metrics when the derivative does not depend on the underlying ZCBs in a linear way.

\subsubsection{Swaps} \label{sec:resultsExposuresSwaps}
To examine the smile and skew effects on the (potential future) exposure profiles of IR swaps, we consider swaps under a single curve setup starting at $T_0$, maturing at $T_m$ with intermediate payment dates $T_1 < T_2 < \ldots < T_{m-1} < T_m$, with strike $\strike$.~\footnote{See Definition~\ref{def:swap} for the swap pricing $\tradeVal^{\text{Swap}}$.}
The examples are limited to co-terminal exposures, i.e., $T_m = 30y$, to properly compare the exposures under both models.
This choice is motivated by the calibration to co-terminal swaptions.
From Figure~\ref{fig:USDNormalRandCotsmilesCalibImpvolError} it is clear that for the quotes that were not used for calibration, e.g., short expiries and tenors, the two models can imply significantly different volatilities, with a similar effect on exposures.

Puetter and Renzitti state that the moneyness~\footnote{Moneyness refers to the relation of the strike and the current swap rate. If the two are equivalent, the swap value is zero and is said to be at-the-money (ATM). Furthermore, the swap is in-the-money (ITM) if the swap value is positive, and out-of-the-money (OTM) if the swap value is negative.} of IR swaps influences how much the $\xva$ values depend on the calibration~\cite{PuetterRenzitti202011a}.
Since ITM swaps are largely driven by the centers of the simulated swap rate distributions, it is expected that the tail effects that follow from the market smile and skew are the smallest in this case.
On the other hand, OTM swaps are much more sensitive to the tails of the swap rate distributions.
Despite the large relative smile impact on $\xva$ numbers in the OTM case, the absolute impact is small.
The authors report that these observations are ``largely in line with the Totem's $\xva$ consensus pricing results, where the CVA submissions for ITM swaps consistently exhibit less variation than those for ATM and OTM swaps.''
Despite the useful insights on $\xva$ metrics, the impact on tail exposures such as $\PFE$ is not addressed in this paper.

The examples presented here are for a receiver swap, starting at $T_0 = 0$, ending at $T_m = 30$, with payments every two years using $M_{\zcb} = 10^4$ paths and 200 dates per year, with a regression polynomial of degree 3, where the short rate is used as the regression variable.
For the exposure simulation, $M_{\tradeVal} = 10^4$ paths and 20 monitoring dates per year are used (with 200 simulation dates per year for the Euler discretized simulation).

The simulated swap rate at $t=20$ is presented in Figure~\ref{fig:RAnDUSDNormalCotsmilesExpiry0SwapRateDistribution}.
The rHW model including smile has a significantly fatter right tail compared to the HW model and is a result of the volatility smile.
The translation of this fat tail to average and tail exposures of a swap will depend on the strike and type of the swap.

\begin{figure}[ht!]
  \centering
  \includegraphics[width=0.5\linewidth]{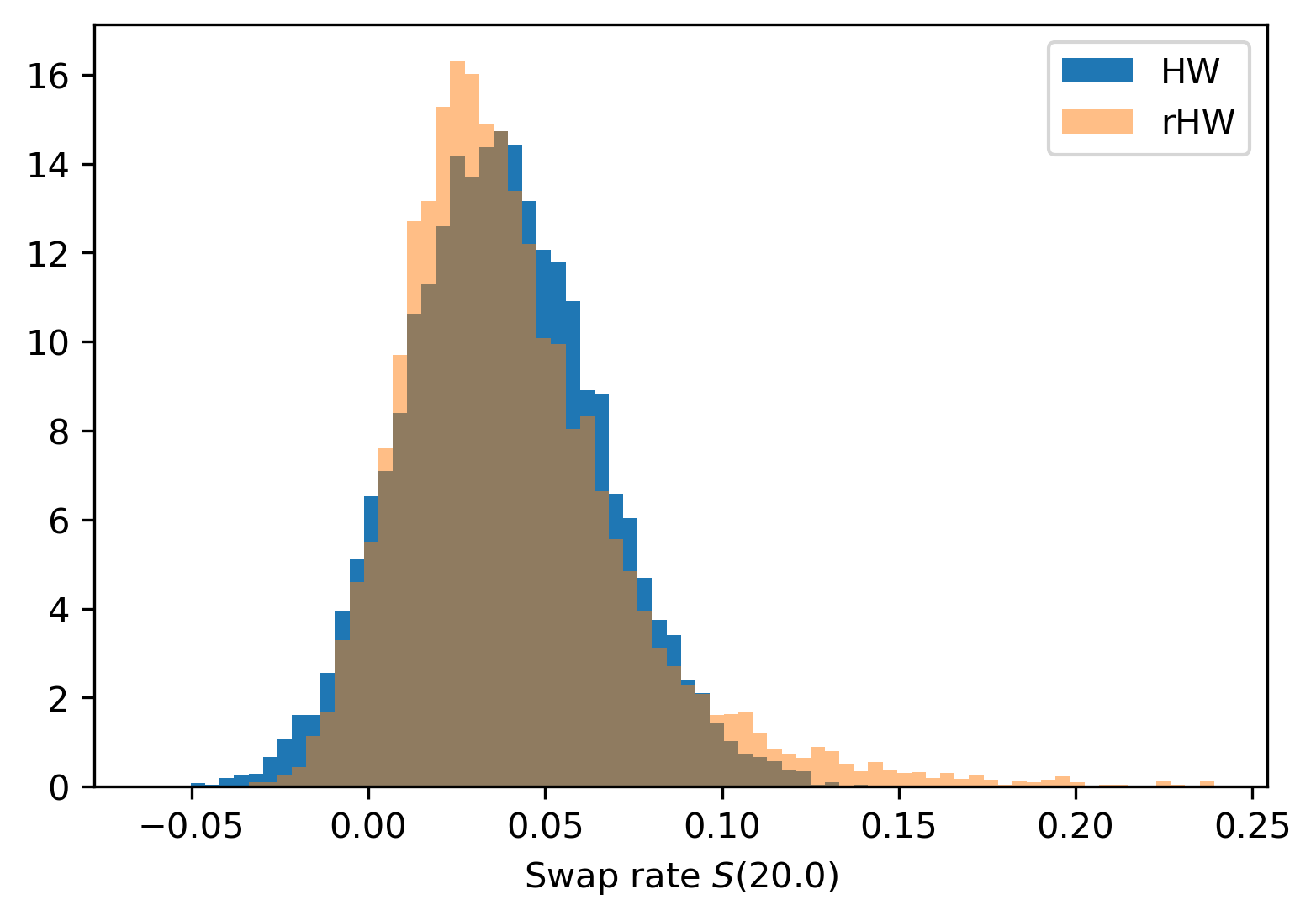}
  \caption{Swap rate distribution at $t=20$.}
  \label{fig:RAnDUSDNormalCotsmilesExpiry0SwapRateDistribution}
\end{figure}

See Figure~\ref{fig:RAnDUSDNormalCotsmilesReceiverSwapATMExposureExpiry0} for ATM receiver swap exposures.
The corresponding $\CVA$ and $\DVA$ numbers for various strikes are reported in Table~\ref{tab:RAnDUSDNormalCotsmilesReceiverSwapExposureExpiry0}.
In this example, the $\EPE$ is significantly affected by the smile, whereas the $\ENE$ of the two models is comparable.
These effects directly translate to the reported $\xva$ metrics, where the $\CVA$ and $\BCVA$ are significantly impacted, with up to a 15\% decrease in $\BCVA$ when smile is included in the model.
In addition, the tail exposures in Figure~\ref{fig:RAnDUSDNormalCotsmilesReceiverSwapATMExposureExpiry0} are affected by the smile, in the sense that the HW model underestimates the tails, with the largest impact on $\PFL$ in this example.
From the $\xva$ metrics in Table~\ref{tab:RAnDUSDNormalCotsmilesReceiverSwapExposureExpiry0}, we conclude that the smile in the rHW model has a significant effect in all cases.
The results are in line with the aforementioned observations from Puetter and Renzitti~\cite{PuetterRenzitti202011a}, in the sense that the impact is relatively the smallest and absolutely the largest in the ITM case and relatively the largest and absolutely the smallest in the OTM case.

\begin{figure}[ht!]
  \centering
  \begin{subfigure}[b]{0.45\linewidth}
    \includegraphics[width=\linewidth]{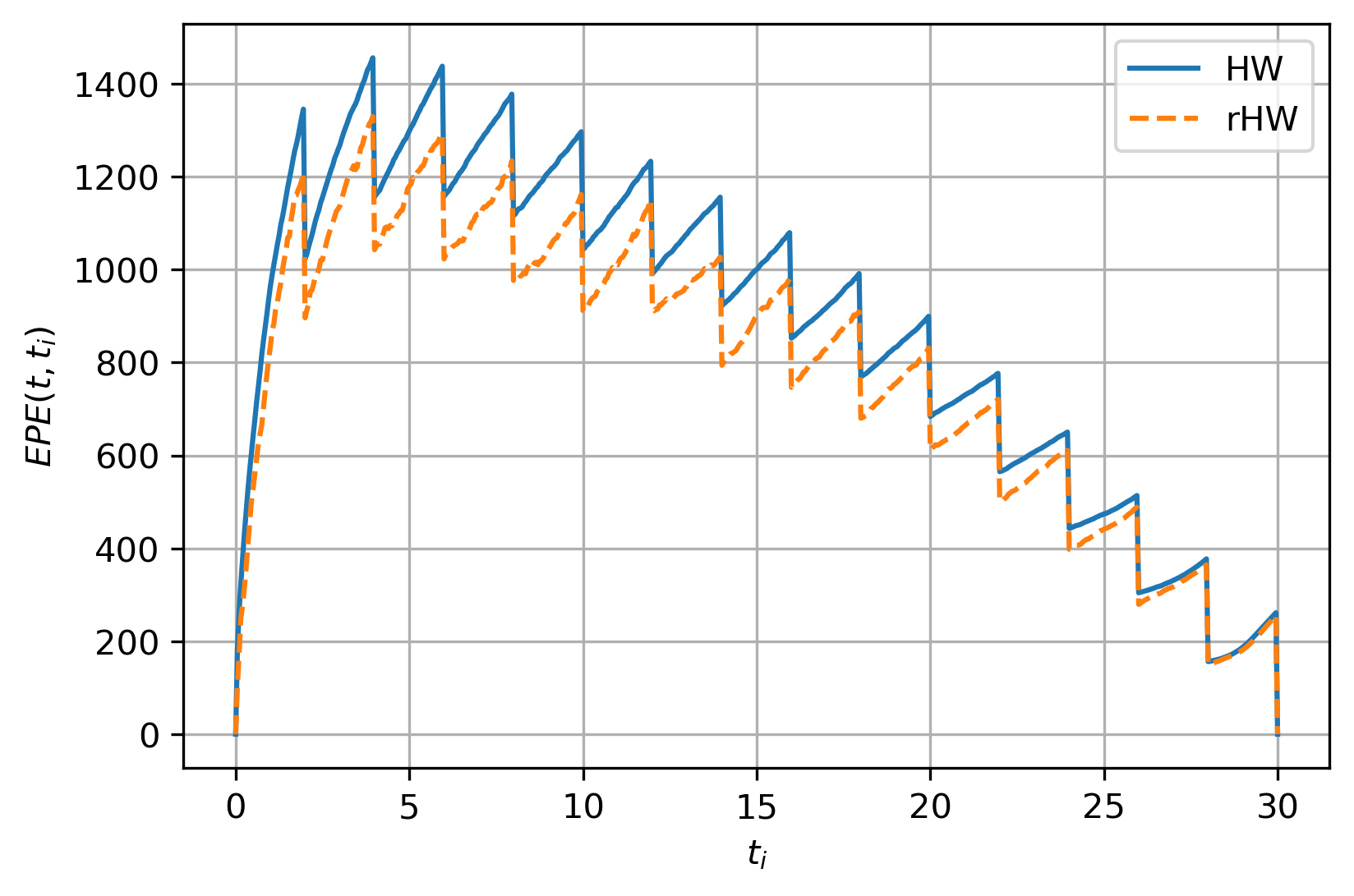}
    \caption{$\EPE(t, t_i)$}
    \label{fig:RAnDUSDNormalCotsmilesReceiverSwapATMExposureExpiry0EPE}
  \end{subfigure}
  \begin{subfigure}[b]{0.45\linewidth}
    \includegraphics[width=\linewidth]{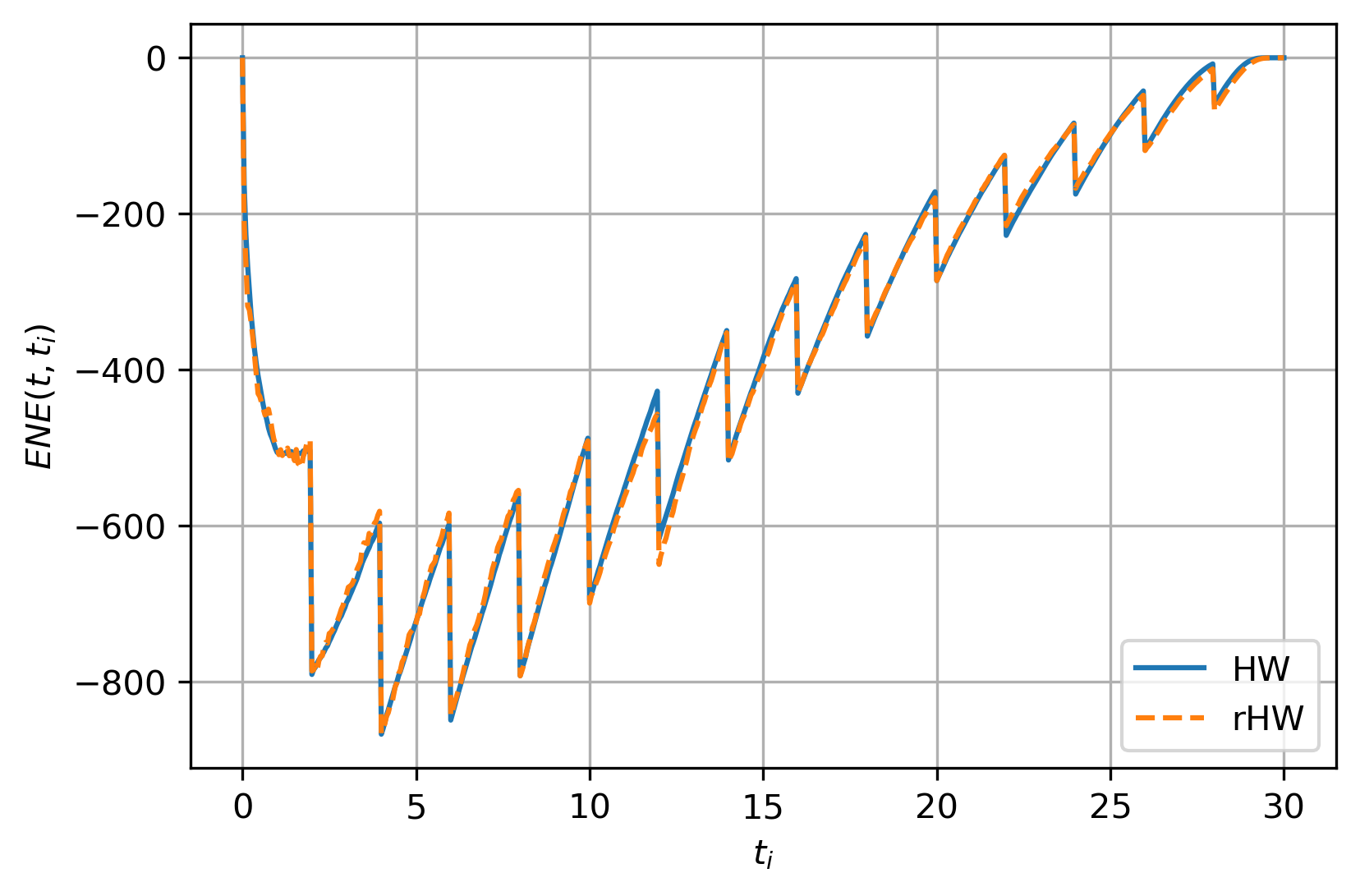}
    \caption{$\ENE(t, t_i)$}
    \label{fig:RAnDUSDNormalCotsmilesReceiverSwapATMExposureExpiry0ENE}
  \end{subfigure}
  \begin{subfigure}[b]{0.45\linewidth}
    \includegraphics[width=\linewidth]{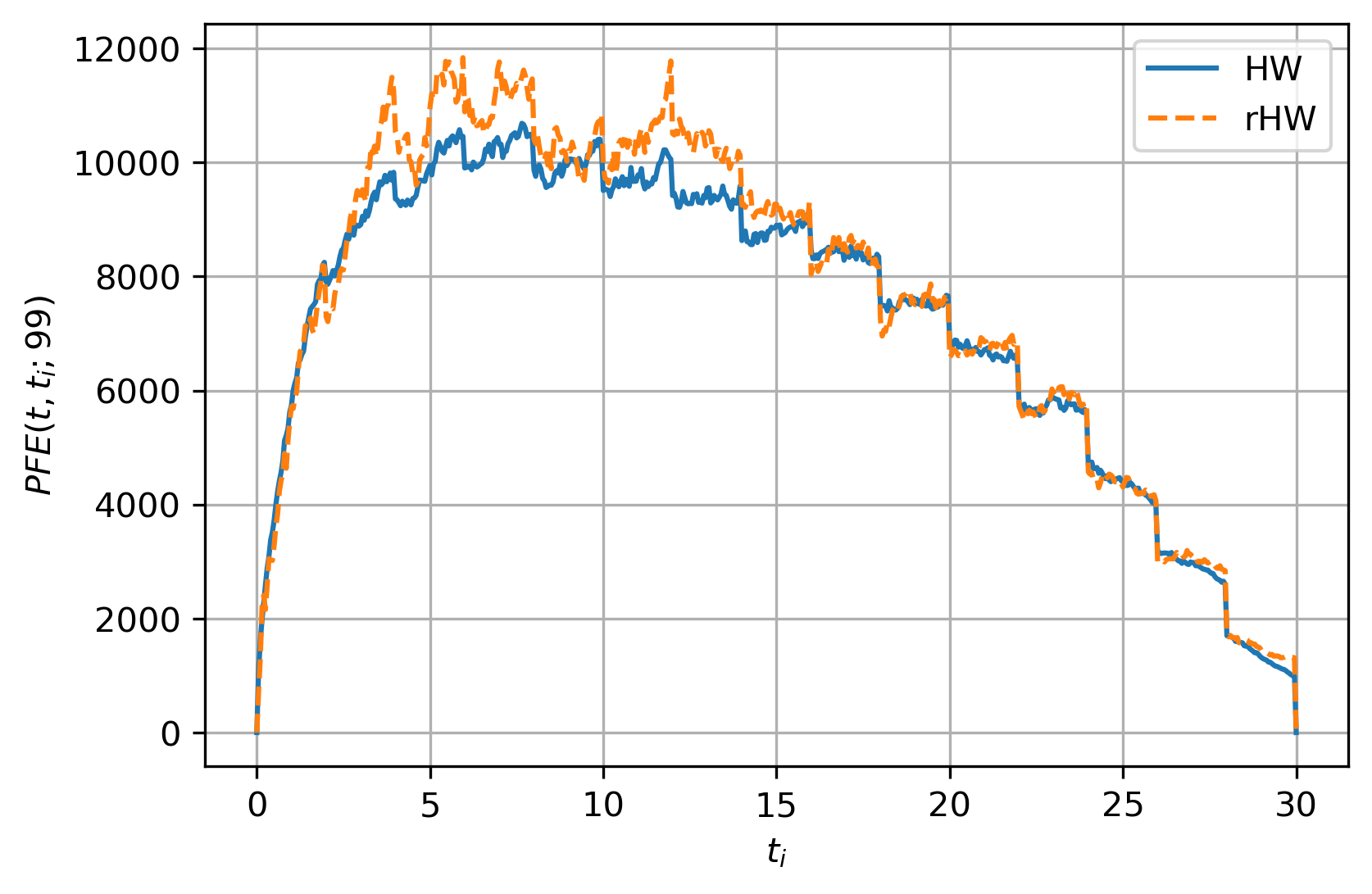}
    \caption{$\PFE(t, t_i; 99)$}
    \label{fig:RAnDUSDNormalCotsmilesReceiverSwapATMExposureExpiry0PFE(99)}
  \end{subfigure}
  \begin{subfigure}[b]{0.45\linewidth}
    \includegraphics[width=\linewidth]{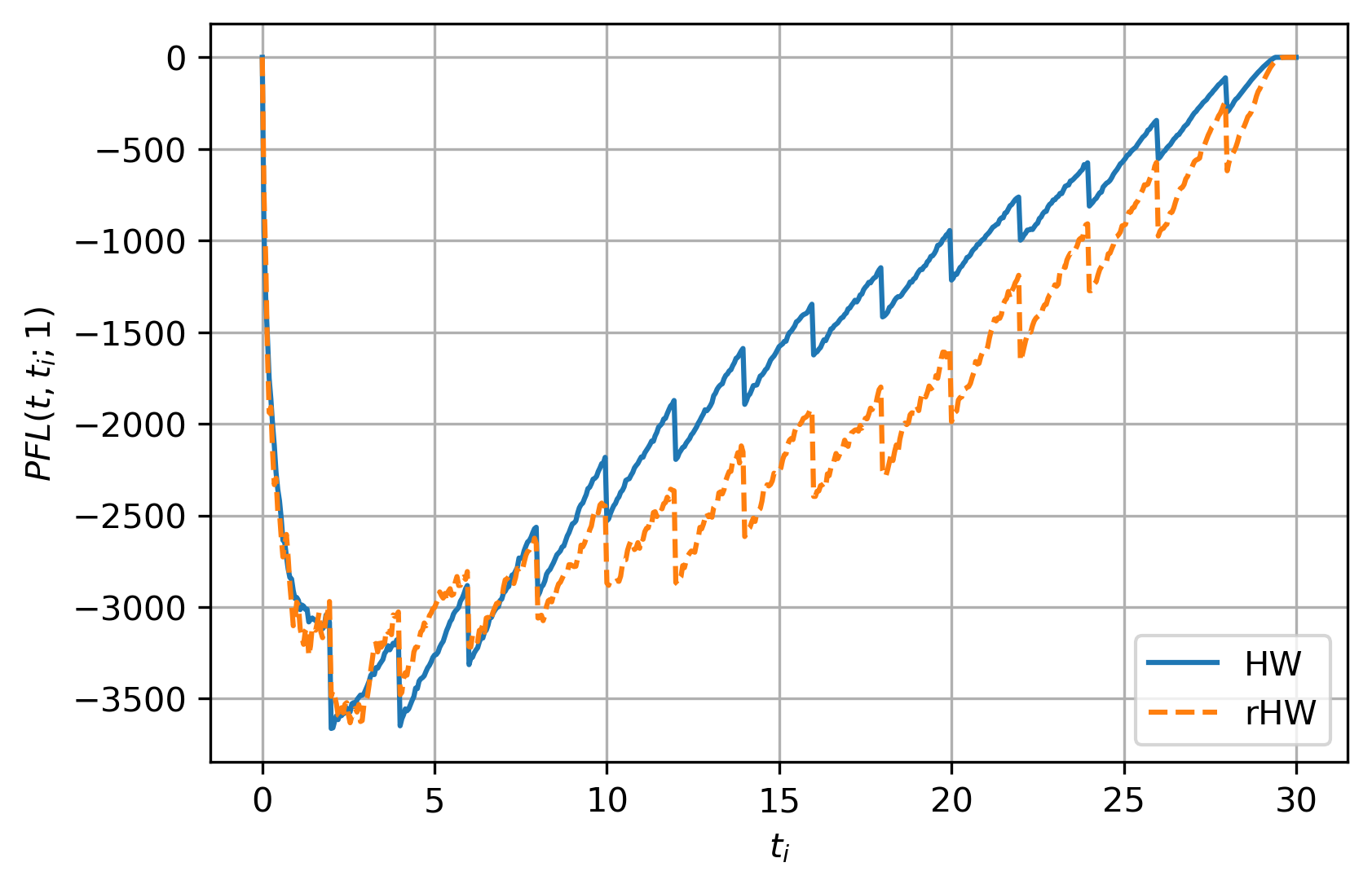}
    \caption{$\PFL(t, t_i; 1)$}
    \label{fig:RAnDUSDNormalCotsmilesReceiverSwapATMExposureExpiry0PFL(1)}
  \end{subfigure}
  \caption{Comparing exposures for an ATM receiver swap ($\strike = \strikeATM$).
  }
  \label{fig:RAnDUSDNormalCotsmilesReceiverSwapATMExposureExpiry0}
\end{figure}

\begin{table}[ht!]
    \centering
    \begin{tabular}{l|ll|rrr}
        Model   & $\strike$             & Moneyness     & $\CVA(t_0)$   & $\DVA(t_0)$   & $\BCVA(t_0)$  \\ \hline
        HW      & $\strikeATM$          & ATM           & 419.098       & -104.201      &  289.612      \\
        rHW     &                       &               & 375.367       & -104.459      &  249.831      \\ \hline
        HW      & $1.5\cdot\strikeATM$  & ITM           & 967.818       &  -23.569      &  862.803      \\
        rHW     &                       &               & 924.550       &  -27.847      &  820.243      \\ \hline
        HW      & $0.5\cdot\strikeATM$  & OTM           & 122.268       & -319.868      & -165.661      \\
        rHW     &                       &               &  92.005       & -323.449      & -196.049
    \end{tabular}
    \caption{$\xva$ metrics for the receiver swap example, for various strikes.}
    \label{tab:RAnDUSDNormalCotsmilesReceiverSwapExposureExpiry0}
\end{table}

The smile effects on the tail exposures are in line with previous observations in the literature on the relevance of skew and smile for $\xva$ computations for FX and equity derivatives~\cite{GraafFengKandhaiOosterlee201406,FengOosterlee201709,SimaitisGraafHariKandhai201605}.
However, in those works, there were no clear conclusions on a significant smile and skew impact on average exposures and $\xva$ metrics of linear derivatives of all moneyness types.

The significant effect of smile and skew on the $\xva$ metrics is particularly relevant from a practical perspective.
Consider an uncollateralized receiver swap between a bank and a client hedging IR risk on a floating-rate loan.
See~\cite{ZwaardGrzelakOosterlee202210} for further background on this use case.
Typically, these swaps are entered at par, i.e., $\strike = \strikeATM$.
From Table~\ref{tab:RAnDUSDNormalCotsmilesReceiverSwapExposureExpiry0} it is clear that when including smile, the $\CVA$ is significantly lower, which allows the bank to offer the client a better rate.
Alternatively, consider the case where the swap was entered in the past, and the trade is now OTM due to market movements, e.g., $\strike = 0.5 \cdot \strikeATM$.
The results in Table~\ref{tab:RAnDUSDNormalCotsmilesReceiverSwapExposureExpiry0} tell that there is net $\DVA$ for this trade ($\BCVA$ is negative).
If the client wants to unwind the swap, the bank needs to charge the client to compensate for the loss made from the $\DVA$ benefit that will disappear.
When including smile in the model, the client needs to be charged more than the amount computed with the HW model.
Furthermore, next to the mispricing effect, there will also be an impact on $\xva$ sensitivities used for hedging.
Ignoring smile in the $\xva$ model will cause under- or over-hedging.

The most computationally expensive part of $\xva$ calculations is the valuation of derivatives at future monitoring dates for all Monte-Carlo scenarios, as this has to be done for each derivative separately.
For linear derivatives such as IR swaps, the future derivative value only depends on the future ZCBs such that the corresponding exposures can be easily obtained.
The exposure calculations for the HW model take 2.03 seconds~\footnote{All the reported runtimes are averages over 20 runs.}, whereas the rHW exposure calculations take 1.90 seconds.
Both implementations are not fully optimized in terms of computational speed, but we can conclude that the runtimes of both models are approximately as fast as one another.

\begin{rem}[Forward smile]
    Extra caution is required for forward-smile sensitive exposures or derivatives.
    Since both models are not calibrated to the forward smile evolution as would be done for a stochastic volatility model, the model-implied forward smile from the two models is likely to be different.
    In other words, one is `left at the mercy' of the smile evolution that the model of choice implies.
    However, the smile-aware rHW model is still preferred over the HW model, as for the former, the marginal information that the market smiles encode is included during the model calibration.
    
    For example, the exposure for a forward-starting swap starting at $T_0>0$ is a forward-starting option, such that smile evolution is relevant when computing this exposure.   
    Since a typical portfolio of IR swaps will not be similar to a forward-starting swap but to a spot-starting swap, the aforementioned results for the swap starting at $T_0=0$ are most relevant from a practical perspective.
\end{rem}

Results for the EUR market data can be found in~\ref{app:resultsEURExposuresSwaps}.
The conclusions drawn here extend to this other market data set.
Due to the smaller implied volatility smile in this example, the smile effects are smaller but non-negligible.

\subsubsection{Bermudan swaption} \label{sec:resultsExposuresBermudanSwaption}

In this experiment, we demonstrate that in addition to the effects of volatility smile and skew on $\PFE$ as observed for IR swaps, there is a significant impact on the $\EPE$ and $\CVA$ of Bermudan swaptions.

A receiver (payer) Bermudan swaption $\tradeVal^{\text{BS}}$ is an early-exercise contract where at a predetermined set of exercise dates $\mathcal{T}_E = \{T_{E,1}, \ldots, T_{E,n}\}$ the option holder has the right to enter into a receiver (payer) swap.
We consider a fixed maturity Bermudan swaption, i.e., the swap that is entered into upon exercise starts straight after the exercise date and always at the same pre-specified maturity date.
As a consequence, the lifetime of the swap is decreasing as time progresses.
Furthermore, we assume that the exercise dates $\mathcal{T}_E$ coincide with the swap reset dates.
So, the option holder can enter the underlying swap at the first reset date up and till the last reset date of the swap.
We value the Bermudan swaption in a Least Squares Monte Carlo (LSMC) setting using the Longstaff-Schwartz method~\cite{LongstaffSchwartz200101}.
This approach is a Monte-Carlo based algorithm where the exercise rule at each exercise date is approximated iteratively using a least-squares regression.
For further background, see~\ref{app:bermudanSwaptionPricing}.

The computation of Bermudan swaption exposures requires additional considerations.
We assume that all exercise dates ($\mathcal{T}_E = \{T_{E,1}, \ldots, T_{E,n}\}$) are also monitoring dates ($t_0 < t_1 < \ldots < t_n$).
We consider the case where there are monitoring dates between exercise dates to examine the exposure profile between consecutive exercise dates.
The results are limited to cash-settled Bermudan swaptions, such that the exposure of the Bermudan swaption after exercise is zero.~\footnote{For physical settlement, the Bermudan swaption exposure after the exercise is given by the swap exposure.}
The key to computing exposures is to obtain future market values $\tradeVal^{\text{BS}}(t_i)$ at all monitoring dates $t_i$ for all Monte-Carlo scenarios.
We compute the path-wise $\tradeVal^{\text{BS}}(t_i)$ using the collocation-based approximation technique from Deelstra~\etal\cite{DeelstraGrzelakWolf202304}, using Lagrangian interpolation to approximate the future market value over the entire domain.
For further background, see~\ref{app:bermudanSwaptionExposures}.
Finally, we assume that the exercise strategy of the Bermudan swaption is not affected by the possibility of default.

We compare exposures for a receiver Bermudan swaption.
The underlying swap starts at $T_0=0$, ends at $T_m=30$, has payments every two years and early-exercise dates at every swap payment date until the swap maturity.
The Monte-Carlo setup for the ZCB regression and exposure simulation is the same as in Section~\ref{sec:resultsExposuresSwaps}.
For the regression phase of the Bermudan swaption, we use a polynomial of degree 2.
Finally, for computing the exposures with the collocation method using 5 nodes, we use $M_{\tradeVal} = 5\cdot 10^3$ paths and 20 dates per year.

Exposures of a receiver Bermudan swaption on an ATM swap are presented in Figure~\ref{fig:RAnDUSDNormalCotsmilesReceiverBermSwaptionATMExposureExpiry0}.
For both models, the effect of the early-exercise is clearly visible in the $\EPE$ profile in Figure~\ref{fig:RAnDUSDNormalCotsmilesReceiverBermSwaptionATMExposureExpiry0EPE}, where at every exercise date, the exposure drops significantly.
The corresponding $\CVA$ numbers for all cases are reported in Table~\ref{tab:RAnDUSDNormalCotsmilesReceiverBermSwaptionExposureExpiry0}.
The results show a significant effect of the volatility smile and skew on the $\EPE$, $\PFE$, and $\CVA$.
For the latter, the effect is between 48\% and 67\%, and is thereby significantly larger than for the IR swap from Section~\ref{sec:resultsExposuresSwaps}.
Hence, the smile effects from the rHW model are more pronounced for this more exotic derivative.
This conclusion is in line with the increasing smile effect for more exotic FX derivatives in~\cite{SimaitisGraafHariKandhai201605}.

\begin{figure}[ht!]
  \centering
  \begin{subfigure}[b]{0.45\linewidth}
    \includegraphics[width=\linewidth]{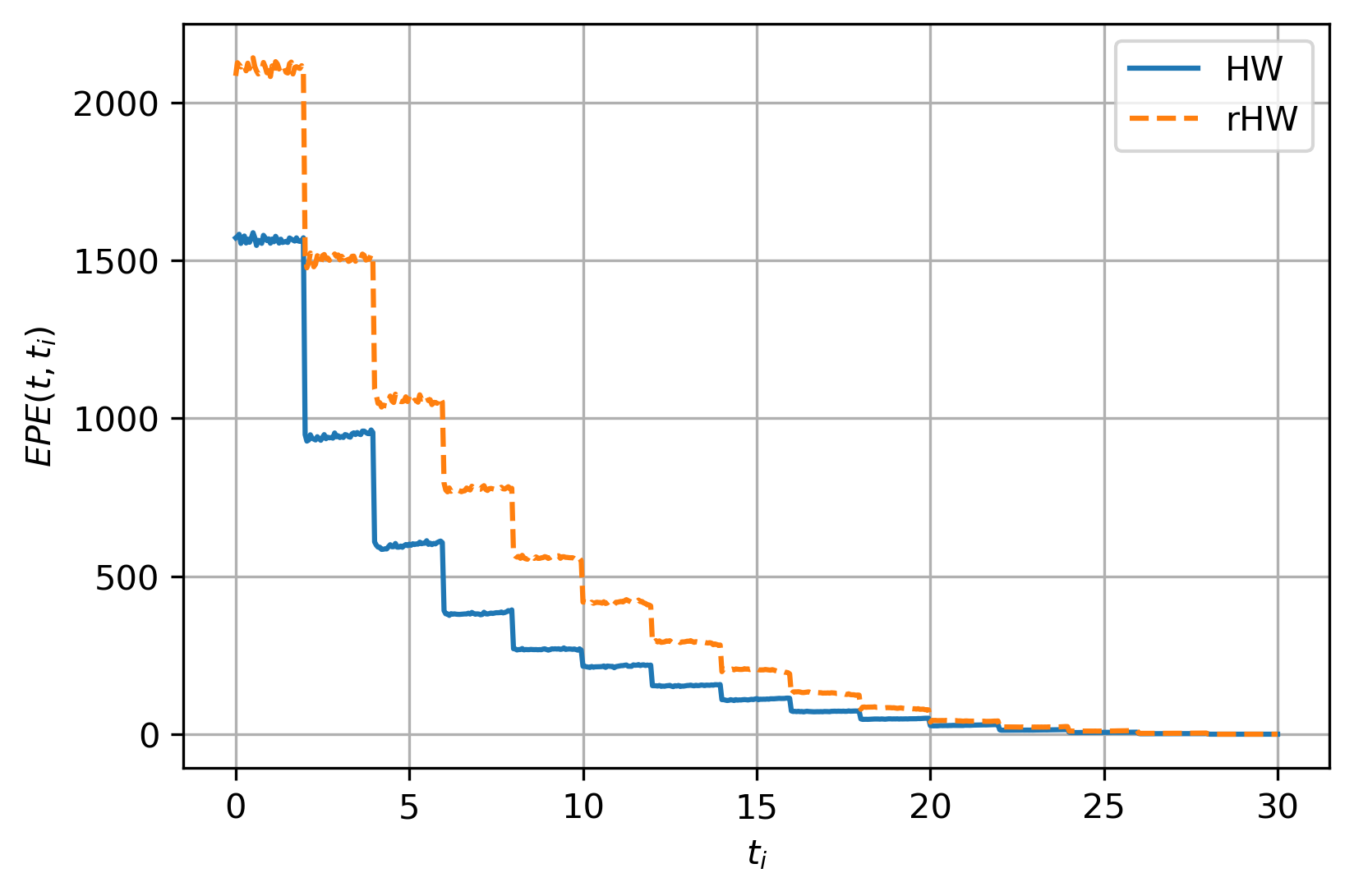}
    \caption{$\EPE(t, t_i)$}
    \label{fig:RAnDUSDNormalCotsmilesReceiverBermSwaptionATMExposureExpiry0EPE}
  \end{subfigure}
  \begin{subfigure}[b]{0.45\linewidth}
    \includegraphics[width=\linewidth]{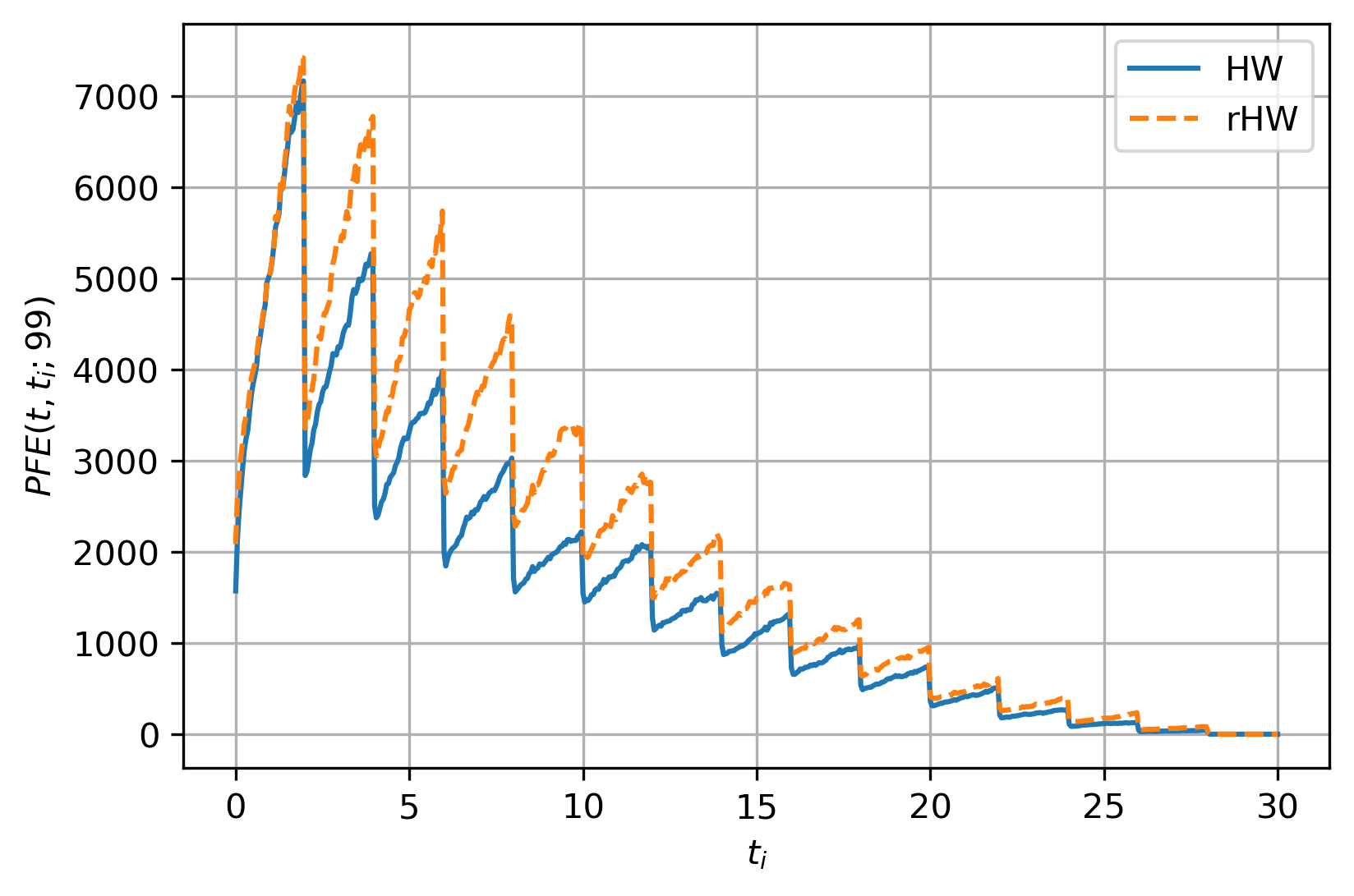}
    \caption{$\PFE(t, t_i; 99)$}
    \label{fig:RAnDUSDNormalCotsmilesReceiverBermSwaptionATMExposureExpiry0PFE(99)}
  \end{subfigure}
  \caption{Comparing exposures for a receiver Bermudan swaption on an ATM swap ($\strike = \strikeATM$).
  }
  \label{fig:RAnDUSDNormalCotsmilesReceiverBermSwaptionATMExposureExpiry0}
\end{figure}

\begin{table}[ht!]
    \centering
    \begin{tabular}{l|ll|r}
        Model   & $\strike$             & Moneyness     & $\CVA(t_0)$   \\ \hline
        HW      & $\strikeATM$          & ATM           & 159.156       \\
        rHW     &                       &               & 258.035       \\ \hline
        HW      & $0.5\cdot\strikeATM$  & OTM           & 106.292       \\
        rHW     &                       &               & 157.305       \\ \hline
        HW      & $1.5\cdot\strikeATM$  & ITM           & 190.855       \\
        rHW     &                       &               & 320.358   
    \end{tabular}
    \caption{$\CVA$ metrics for a receiver Bermudan swaption on a swap, for various strikes.}
    \label{tab:RAnDUSDNormalCotsmilesReceiverBermSwaptionExposureExpiry0}
\end{table}

The results for EUR Bermudan swaptions are presented in~\ref{app:resultsEURExposuresBermudanSwaption}, where again the smile effects are of a similar nature, but of a lower magnitude.

\section{Conclusion} \label{sec:conclusion}

We have proposed a novel way to capture market implied smile and skew for $\xva$ calculations of IR derivatives.
This approach extends the current market practice of HW-type dynamics merely calibrated to the ATM point of market implied co-terminal swaptions.
The inclusion of smile and skew in $\xva$ calculations are relevant both in terms of mispricing, as well as risk management of $\xva$s.

We have demonstrated that the rHW model is sufficiently flexible to be calibrated to the market-implied co-terminal swaption smiles efficiently, preserving the semi-analytic calibration of the underlying models.
Furthermore, the rHW model allows for fast pricing of derivatives in a generic Monte-Carlo framework, which is typically used for exposure calculations.
The future rHW ZCBs that are required to compute future market values are retrieved using an independent regression-based Monte-Carlo simulation, such that evaluating a ZCB in an exposure framework boils down to evaluating a polynomial with pre-computed coefficients.
Evaluating the rHW ZCBs is of similar computational cost evaluating the analytic ZCB expression for the HW model.

Using the rHW model calibrated to market data, it can be shown that there is a significant effect of smile and skew on exposures and $\xva$ metrics of IR derivatives.
For IR swaps, we have shown that in line with previous observations in the literature on exposures for FX and equity derivatives~\cite{GraafFengKandhaiOosterlee201406,FengOosterlee201709,SimaitisGraafHariKandhai201605}, the tails of the exposure distributions are affected by the smile.
Remarkably, we have also found that the average exposures and corresponding $\xva$ metrics are also significantly affected for all moneyness types, which has profound practical implications for $\xva$ pricing and risk management.
For early-exercise derivatives such as Bermudan swaptions, the same effects as for IR swaps are present, but with a larger magnitude of the smile impact.

These results motivate to investigate the extension of the rHW model to a multi-currency setting, which will undoubtedly come with additional challenges due to the high-dimensional nature of the underlying system of model dynamics.

\section*{Acknowledgements} This work has been financially supported by Rabobank.
The authors would like to thank Natalia Borovykh for her valuable input and fruitful discussions.

{\footnotesize
\bibliographystyle{abbrv}
\bibliography{bib/MacroStrings,bib/Articles,bib/Books,bib/Misc,bib/Regulation} 
}

\appendix
\setlength\parindent{0pt}
\renewcommand{\thesection}{\appendixname\Alph{section}}
\renewcommand{\thedefinition}{\Alph{section}.\arabic{definition}}
\renewcommand{\thethrm}{\Alph{section}.\arabic{thrm}}
\renewcommand{\theprop}{\Alph{section}.\arabic{prop}}
\renewcommand{\thecrllry}{\Alph{section}.\arabic{crllry}}
\renewcommand{\thelem}{\Alph{section}.\arabic{lem}}

\section{Density measure change and invariance} \label{app:densityResults}

\begin{prop}[Density measure change] \label{prop:densityMeasureChange}
  Let $(\Omega, \F, \M)$ be a probability space where $\Omega$ is the set of outcomes (sample space), $\F$ is the filtration (events) and $\M$ is the probability measure.
  Consider real-valued random variable $X$ which is a measurable function $X : \Omega \to \R$.
  An $\M$-density of $X$ is a function $f^{\M}: \Omega \to [0,\infty)$, which can be recognized as the Radon-Nikodym derivative w.r.t. the Lebesgue measure $\mu$, i.e.,
  \begin{align}
    f^{\M} (y)
      &= \deriv{\M(y)}{\mu(y)}
      = \deriv{\M(y)}{y}, \nonumber
  \end{align}
  where the latter follows from the restriction to the real line.
  Assume that measures $\M$ and $\Nmeas$ are both absolutely continuous w.r.t. $\mu$, denoted as $\M \ll \mu$ and $\Nmeas \ll \mu$. 
  This assumption is not restrictive, as for continuous random variables this simply means that the densities are well defined.
  Assume that additionally $\M$ is absolutely continuous w.r.t. $\Nmeas$, i.e., $\M \ll \Nmeas \ll \mu$.
  Then, $\mu$-a.e.,
  \begin{align}
    f^{\M} (y)
      &= \deriv{\M(y)}{\Nmeas(y)} f^{\Nmeas} (y), \label{eq:randPDFtoN}
  \end{align}
  where $\deriv{\M(y)}{\Nmeas(y)}$ is the Radon-Nikodym derivative to change between measures $\M$ and $\Nmeas$.
\end{prop}
\begin{proof}
  Under the assumption of $\M \ll \Nmeas \ll \mu$, the result follows from the iterative application of the Radon-Nikodym derivative, i.e., the product rule for Radon-Nikodym derivatives:
  \begin{align}
      f^{\M} (y)
        &= \deriv{\M(y)}{\mu(y)}
        = \deriv{\M(y)}{\Nmeas(y)} \deriv{\Nmeas(y)}{\mu(y)} 
        = \deriv{\M(y)}{\Nmeas(y)} f^{\Nmeas} (y). \label{eq:randPDFtoNProof}
  \end{align}
\end{proof}

\begin{crllry}[Density measure invariance] \label{crllry:densityMeasureInvariance}
  Let the density equation~\eqref{eq:randPDF} be given under measure $\M$.
  If $\M \ll \Nmeas$, then the density equation is measure invariant, i.e., when the density equation holds under measure $\M$, it also holds for $\Nmeas$.
\end{crllry}
\begin{proof}
  This result follows directly from applying Proposition~\ref{prop:densityMeasureChange} to both sides of Equation~\eqref{eq:randPDF}, and the linearity of the Radon-Nikodym derivative.
\end{proof}

\section{The Hull-White model} \label{app:HW}

\subsection{Dynamics with time-dependent volatility} \label{app:HWTimeDepVol}
For ease of display, the results in Section~\ref{sec:rHW} are presented for the case of constant model volatility $\vol_x$.
However, in the $\xva$ context we require a time-dependent model volatility $\vol_x(t)$.
Here we present the formulae from Section~\ref{sec:rHW} for $\vol_x(t)$ rather than $\vol_x$.

The majority of the formulae in Section~\ref{sec:rHWDynamics} stay the same, but Equation~\eqref{eq:HWSDE1Integrated} changes into
\begin{align}
  \shortRate_n(t)
    &= x_n(s)\expPower{-\theta_n (t-s)} +  b_n(t) + \int_{s}^{t} \vol_{x}(u) \expPower{-\theta_n  (t-u)} \d\brownian^{\M}(u). \label{eq:HWSDE1IntegratedApp}
\end{align}
The variance of this process, see Equation~\eqref{eq:HWPDF}, becomes:
\begin{align}
  \condVarSmall{\shortRate_n(t)}{s}
    &= \int_{s}^{t} \vol_{x}^2(u) \expPower{-2\theta_n (t-u)}\du. \nonumber
\end{align}
Furthermore, the following expressions change in Proposition~\ref{prop:HWZCB}:
\begin{align}
  V_n(s,t) 
    &\ldef \int_{s}^{t} \vol_{x}^2(u) B_n^2(u,t) \du, \nonumber \\
  b_n(T)
    &= f^{\text{M}}(0,T) - x_n(0) \expPower{-\theta_nT} + \int_{0}^{T} \vol_{x}^2(u) B_n(u,T) \expPower{-\theta_n (T-u)}\du. \label{eq:HWb} 
\end{align}

\subsection{Zero-Coupon Bond} \label{app:HWZeroCouponBond}

The deterministic function $b_n(t)$ in Equation~\eqref{eq:HWSDE1} is obtained by fitting the model to the market yield curve, which is done in the derivation of the Zero-Coupon Bond (ZCB) under the HW dynamics, see Proposition~\ref{prop:HWZCB}.

\begin{prop}[HW ZCB] \label{prop:HWZCB}
  Let $\shortRate_n(t)$ be the HW dynamics from Equation~\eqref{eq:HWSDE1}.
  The ZCB under these dynamics is given by
  \begin{align}
    \zcb_{\shortRate_n}(t,T)
      &= \expPower{\half V_n(t,T) - \int_{t}^{T} b_n(u) \du - x_n(t)B_n(t,T) }, \nonumber \\
    V_n(s,t) 
      &\ldef \vol_{x}^2 \int_{s}^{t}  B_n^2(u,t) \du, \ \ 
    B_n(s,t)
      \ldef \frac{1}{\theta_n}\left( 1 - \expPower{-\theta_n (t-s)} \right), \nonumber \\
    b_n(T)
        &= f^{\text{M}}(0,T) - x_n(0) \expPower{-\theta_nT} +  \half \vol_{x}^2  B_n^2(0,T). \label{eq:HWb} 
  \end{align}
  For $t=0$ and $x_n(0)=0$ we have $\zcb_{\shortRate_n}(0,T) = \zcb^{\text{M}}(0,T)$ $\forall T$.
\end{prop}    
\begin{proof}
    Integrating the dynamics from Equation~\eqref{eq:HWSDE1Integrated} yields an expression for $\int_{s}^{t}\shortRate_n(u)\du$:
    \begin{align}
      \int_{s}^{t} \shortRate_n(u) \du
        &= x_n(s) B_n(s,t)  + \int_{s}^{t} b_n(u) \du + \vol_{x}  \int_{s}^{t} B_n(u,t)\d\brownian^{\M}(u), \label{eq:HW2b}
    \end{align}
    which follows a normal distribution with mean and variance
    \begin{align}
      \condExpSmallGeneric{\int_{s}^{t}\shortRate_n(u) \du}{s}{\Q_{\shortRate_n}}
        &=  x_n(s) B_n(s,t) + \int_{s}^{t} b_n(u) \du, \nonumber \\
      V_n(s,t)
        &\ldef \condVarSmall{\int_{s}^{t} \shortRate_n(u) \du}{s}
        = \vol_{x}^2 \int_{s}^{t}  B_n^2(u,t) \du. \nonumber
    \end{align}
    Using the definition of the ZCB we write~\footnote{Using the property of exponential normals: if $z \sim \N\left(\mu_z,\sigma_z^2\right)$ then $\E\left[\expPower{z}\right] = \expPower{\mu_z + \half\sigma_z^2}$.}
    \begin{align}
      \zcb_{\shortRate_n}(t,T)
        &=  \expPower{\condExpSmallGeneric{-\int_{t}^{T} \shortRate_n(u) \du  }{t}{\Q_{\shortRate_n}}  + \half \condVarSmall{-\int_{t}^{T} \shortRate_n(u) \du  }{t}}
        = \expPower{ \half V_n(t,T) - \int_{t}^{T} b_n(u) \du - x_n(t) B_n(t,T)}. \label{eq:HW5}
    \end{align}
    In this equation, $b_n(t)$ is the only unknown.
    Writing down the term structure of the model can be done using Equation~\eqref{eq:HW5} and the Leibniz integration rule~\footnote{For $a$ constant, $\deriv{}{x}\left( \int_a^x f(x,t) \dt \right) = f(x,x) + \int_a^x \pderiv{}{x} f(x,t) \dt$.}:
    \begin{align}
      f_{\shortRate_n}(t,T)
        &= \frac{-\partial \left(\ln \zcb_{\shortRate_n}(t,T) \right)}{\partial T}
        =  - \half\pderiv{V_n(t,T)}{T} + x_n(t) \pderiv{B_n(t,T)}{T} + b_n(T). \label{eq:HW6}
    \end{align}
    We want the model's initial term structure $f_{\shortRate_n}(0,T)$ to fit the initial term structure from the market $f^{\text{M}}(0,T)$.
    As a consequence:
    \begin{align}
      b_n(T)
        &= f^{\text{M}}(0,T)  + \half \pderiv{V_n(0,T)}{T} - x_n(0) \pderiv{B_n(0,T)}{T} 
        = f^{\text{M}}(0,T) - x_n(0) \expPower{-\theta_nT} +  \half \vol_{x}^2  B_n^2(0,T). \nonumber
    \end{align}
    This concludes the proof.
\end{proof}

\subsection{Zero-Coupon Bond Option pricing} \label{app:HWZCBO}
Using the ZCB expression from Proposition~\ref{prop:HWZCB}, we look at the pricing of a Zero-Coupon Bond Option (ZCBO), see Proposition~\ref{prop:HWZCBO}.
\begin{prop}[HW ZCBO pricing] \label{prop:HWZCBO}
  Let $\shortRate_n(t)$ be the HW dynamics from Equation~\eqref{eq:HWSDE1}.
  A ZCBO is an option expiring at $T$ on a ZCB $\zcb_{\shortRate_n}(T, S)$ with maturity $S > T$, with strike $\strike$, and option type $\optType=1$ for a call and $\optType = -1$ for a put.
  The value of a unit notional ZCBO under the $\shortRate_n(t)$ dynamics is given by:
  \begin{align}
    \tradeVal^{\text{ZCBO}}_{\shortRate_n}(t;T,S,\strike, \optType)
      &= \optType \left[\zcb_{\shortRate_n}(t,S) \normCDF(\optType d_1) - \strike  \zcb_{\shortRate_n}(t,T) \normCDF(\optType d_2)\right], \label{eq:zcboHW} \\
    d_1
      &\ldef \frac{\ln \left( \frac{\zcb_{\shortRate_n}(t,S)}{\strike \cdot \zcb_{\shortRate_n}(t,T)}\right)}{\Sigma_{\shortRate_n}(t,T,S)} + \half\Sigma_{\shortRate_n}(t,T,S), \ 
    d_2
      \ldef d_1 - \Sigma_{\shortRate_n}(t,T,S), \nonumber \\
    \Sigma_{\shortRate_n}^2(t,T,S)
      &\ldef \condVarSmall{\ln \zcb_{\shortRate_n}(T,S)}{t}. \nonumber
  \end{align}
\end{prop}
\begin{proof}
    The ZCBO payoff is given by $\payoff^{\text{ZCBO}}(T; \shortRate_n(T)) \ldef \maxOperator{\optType \left[\zcb_{\shortRate_n}(T,S) - \strike \right] }$, such that the value is given by
      \begin{align}
        \tradeVal^{\text{ZCBO}}_{\shortRate_n}(t;T,S,\strike,\optType)
          &= \zcb_{\shortRate_n}(t, T)\condExpSmallGeneric{\payoff^{\text{ZCBO}}(T; \shortRate_n(T))}{t}{\Q_{\shortRate_n}^T}. \label{eq:zcbo}
      \end{align}
      The ZCB $\zcb_{\shortRate_n}(T,S)$ is stated in Proposition~\ref{prop:HWZCB}.
      We move to the corresponding $T$-forward measure $\Q_{\shortRate_n}^T$, so we write down the following Radon-Nikodym derivative:
      \begin{align}
        \left.\deriv{\Q_{\shortRate_n}^T}{\Q_{\shortRate_n}}\right|_{\F(T)}
          &= \frac{\bank_{\shortRate_n}(0)}{\bank_{\shortRate_n}(T)} \frac{\zcb_{\shortRate_n}(T,T)}{\zcb_{\shortRate_n}(0,T)} 
          = \expBrace{ -\half \vol_{x}^2 \int_{0}^{T}  B_{n}^2(u,T) \du - \vol_{x} \int_{0}^{T}  B_{n}(u,T)\d\brownian^{\Q_{\shortRate_n}}(u)}. \nonumber
      \end{align}
      Girsanov then gives the $\Q_{\shortRate_n}^T$-Brownian motion: $\d \brownian^{\Q_{\shortRate_n}^T}(t) = \d \brownian^{\Q_{\shortRate_n}}(t) + \vol_{x} B_{n}(t,T) \dt$.
      The HW dynamics from Equation~\eqref{eq:randUnderlyingSDE} for $\M = \Q_{\shortRate_n}^{T}$ are then given by
      \begin{align}
        \d\shortRate_n(t)
          &= \mu_{\shortRate_n}^{\Q_{\shortRate_n}^T}(t, \shortRate_n(t)) \dt + \eta_{\shortRate_n}(t, \shortRate_n(t)) \dW^{\Q_{\shortRate_n}^T}(t), \nonumber \\
        \mu_{\shortRate_n}^{\Q_{\shortRate_n}^T}(t, \shortRate_n(t)) 
          &= \mu_{\shortRate_n}^{\Q_{\shortRate_n}}(t, \shortRate_n(t)) - \vol_x^2 B_{n}(t,T), \label{eq:hw1fDriftTFwd}
      \end{align}
      where the diffusion is as in Proposition~\ref{prop:randUnderlyingSDEHW}.
      Using the change of measure, we write the integrated HW dynamics~\eqref{eq:HWSDE1Integrated} under the $T$-forward measure $\Q_{\shortRate_n}^T$:
      \begin{align}
        \shortRate_n(t)
          &= x_n(s)\expPower{-\theta_n (t-s)} + b_n(t) - M_{n}^T(s,t) + \vol_{x}  \int_{s}^{t} \expPower{-\theta_n (t-u)} \d \brownian^{\Q_{\shortRate_n}^T}(u), \label{eq:HWSDE1IntegratedTforward} \\
        M_{n}^T(s,t)
          &\ldef \vol_{x}^2 \int_{s}^{t}  \expPower{-\theta_n (t-u)} B_{n}(u,T) \du. \label{eq:HWSDE1IntegratedTforwardDriftAdjustment}
      \end{align}
      Under $\Q_{\shortRate_n}^T$, $\shortRate_n(t)$ still follows a normal distribution, but with an adjusted mean $\condExpSmallGeneric{\shortRate_n(t)}{s}{\Q_{\shortRate_n}^T} = \condExpSmallGeneric{\shortRate_n(t)}{s}{\Q_{\shortRate_n}} - M_{n}^T(s,t)$.
      Hence, under $\Q_{\shortRate_n}^T$, conditional on $\F_t$, $\ln \zcb_{\shortRate_n}(T,S)$ is normally distributed with mean and variance
      \begin{align}
        \condExpSmallGeneric{\ln \zcb_{\shortRate_n}(T,S) }{t}{\Q_{\shortRate_n}^T}
          &= \half V_n(T,S) - \int_{T}^{S} b_n(u) \du - B_n(T,S)\left[ x_n(t)\expPower{-\theta_n (T-t)} - M_{n}^T(t,T) \right] \nonumber \\
          &= \condExpSmallGeneric{\ln \zcb_{\shortRate_n}(T,S) }{t}{\Q_{\shortRate_n}} + B_n(T,S)  M_{n}^T(t,T) , \label{eq:HW15a} \\
        \Sigma_{\shortRate_n}^2(t,T,S)
          &\ldef \condVarSmall{\ln \zcb_{\shortRate_n}(T,S)}{t}
          = B_n^2(T,S)  \condVarSmall{\shortRate_n(T)}{t}. \label{eq:HW15b}
      \end{align}
      As a consequence, the ZCBO pricing from Equation~\eqref{eq:zcbo} under the HW model reduces to a Black-type formula.
\end{proof}

\subsection{Swaption pricing} \label{app:HWSwaption}
Using the ZCBO pricing equation from Proposition~\ref{prop:HWZCBO}, we can now shift the attention to pricing swaptions in a similar setup.
First, in Definition~\ref{def:swap}, we introduce the valuation of a swap, after which in Proposition~\ref{prop:HWSwaption} the swaption pricing is discussed.

\begin{definition}[Swap] \label{def:swap}
  Let the value of a unit notional swap under a single curve setup starting at $T_0$, maturing at $T_m$ with intermediate payment dates $T_1 < T_2 < \ldots < T_{m-1} < T_m$, with strike $\strike$ be given by:
  \begin{align}
    \tradeVal^{\text{Swap}}(t;T_1,\ldots,T_m,\strike,\swapType)
      &= \swapType \left(-\zcb(t,T_0) + \zcb(t,T_m) + \strike \sum_{k=1}^m \dct_k \zcb(t,T_k)\right), \label{eq:swap1} \\
    \dct_k
      &\ldef T_k - T_{k-1}, \ 
    \swapType
      \ldef \left\{
           \begin{array}{cc}
        	-1, & \text{payer swap}, \\
        	1, & \text{receiver swap}.
           \end{array}
           \right. \nonumber
  \end{align}
\end{definition}

\begin{prop}[HW swaption pricing] \label{prop:HWSwaption}
  Let $\shortRate_n(t)$ be the HW dynamics from Equation~\eqref{eq:HWSDE1} under $\Q_{\shortRate_n}$.
  Consider a swaption on a swap as per Definition~\ref{def:swap}, with swaption expiry date $T_M \leq T_0$, where $T_0$ is the first swap reset date.
  The value of the swaption under the $\shortRate_n(t)$ dynamics and under the measure $\Q_{\shortRate_n}^T$ is given by:
  \begin{align}
    &\tradeVal^{\text{Swaption}}_{\shortRate_n}(t;T_M,T_1,\ldots,T_m,\strike,\swapType)
      = \sum_{k=1}^m w_k  \tradeVal^{\text{ZCBO}}_{\shortRate_n}(t;T_M,T_k,\tilde{\strike}_k,\swapType), \label{eq:swaption10HW} \\
    &w_k
        \ldef \left\{
        \begin{array}{ll}
          -1, & k=0, \\
          \strike\dct_k, & k=1,\ldots,m-1, \\
          1 + \strike \dct_k, & k=m,
        \end{array}
        \right. \ 
    \tilde{\strike}_k
        \ldef \left\{
        \begin{array}{ll}
          1, & k=0, \\
          \zcb_{\shortRate_n}(T_M,T_k;\shortRate_n^{*}(T_M)), & k \in \{1,\ldots,m\}.
        \end{array}
        \right. \nonumber
  \end{align}
  Here, $\shortRate_n^{*}(T_M)$ is obtained numerically by solving $\sum_{k=0}^m w_k \zcb_{\shortRate_n}(T_M,T_k;\shortRate_n^{*}(T_M)) = 0$.
  Furthermore, ZCB option $\tradeVal^{\text{ZCBO}}_{\shortRate_n}(t;T,S,\strike,\swapType)$ under the HW model and the $\Q_{\shortRate_n}^T$ measure is given in Proposition~\ref{prop:HWZCBO}.
\end{prop}
\begin{proof}
    The swaption payoff $\payoff^{\text{Swaption}}(T_M;\shortRate_n(T_M))$ reads
    \begin{align}
      &\payoff^{\text{Swaption}}(T_M;\shortRate_n(T_M))
      \ldef \maxOperator{\tradeVal^{\text{Swap}}_{\shortRate_n}(T_M;T_1,\ldots,T_m,\strike,\swapType)} \nonumber \\
      & = \maxOperator{\swapType \left[-\zcb_{\shortRate_n}(T_M,T_0;\shortRate_n(T_M)) + \zcb_{\shortRate_n}(T_M,T_m;\shortRate_n(T_M)) + \strike \cdot \sum_{k=1}^m \dct_k \zcb_{\shortRate_n}(T_M,T_k;\shortRate_n(T_M))\right]}, \label{eq:swaption1HW}
    \end{align}
    where we explicitly denote the dependence of the ZCB on $\shortRate_n(T_M)$.
    We define $\shortRate_n^{*}(T_M)$ such that the following holds:
    \begin{align}
      &\payoff^{\text{Swaption}}(T_M;\shortRate_n(T_M)) = 0, \nonumber \\
      \Leftrightarrow \ & - \zcb_{\shortRate_n}(T_M,T_0;\shortRate_n^{*}(T_M)) + \zcb_{\shortRate_n}(T_M,T_m;\shortRate_n^{*}(T_M)) + \strike \cdot \sum_{k=1}^m \dct_k \zcb_{\shortRate_n}(T_M,T_k;\shortRate_n^{*}(T_M)) = 0, \nonumber \\
      \Leftrightarrow \ & \sum_{k=0}^m w_k \zcb_{\shortRate_n}(T_M,T_k;\shortRate_n^{*}(T_M)) = 0. \label{eq:swaption7HW}
    \end{align}
    Here, $\shortRate_n^{*}(T_M)$ is obtained numerically such that the above holds.
    Since $\zcb_{\shortRate_n}(T_M,T_k;y)$ from Proposition~\ref{prop:HWZCB} is monotonically decreasing in $y$, the swaption only pays out if $\shortRate_n(T_M) > \shortRate_n^{*}(T_M)$.
    Using the monotonicity and the result from Equation~\eqref{eq:swaption7HW} allows us to rewrite the swaption payoff from Equation~\eqref{eq:swaption1HW}:
    \begin{align}
      \payoff^{\text{Swaption}}(T_M;\shortRate_n(T_M))
        &=  \swapType \left(\zcb_{\shortRate_n}(T_M,T_m;\shortRate_n(T_M)) - \tilde{\strike}_m \right) \indicator{\shortRate_n(T_M) > \shortRate_n^{*}(T_M)}  \nonumber \\
        &\quad + \left(\strike \cdot \sum_{k=1}^m \dct_k \swapType \left[  \zcb_{\shortRate_n}(T_M,T_k;\shortRate_n(T_M)) - \tilde{\strike}_k \right]\right) \indicator{\shortRate_n(T_M) > \shortRate_n^{*}(T_M)} \nonumber \\
        & = \sum_{k=1}^m w_k \maxOperator{\swapType \left[\zcb_{\shortRate_n}(T_M,T_k;\shortRate_n(T_M)) - \tilde{\strike}_k \right]}. \label{eq:swaption9HW}
    \end{align}
    Rewriting the swaption payoff in this particular form is known as Jamshidian's decomposition~\cite{Jamshidian198903}.
    Using the payoff from Equation~\eqref{eq:swaption9HW} we write the decompose the swaption payoff into a sum of ZCB put options.
    \begin{align}
      \tradeVal^{\text{Swaption}}_{\shortRate_n}(t;T_M,T_1,\ldots,T_m,\strike,\swapType)
        &= \zcb_{\shortRate_n}(t,T_M) \condExpSmallGeneric{\payoff^{\text{Swaption}}(T_M;\shortRate_n(T_M))}{t}{\Q_{\shortRate_m}^{T_M}}  \nonumber \\
        &= \zcb_{\shortRate_n}(t,T_M) \sum_{k=1}^m w_k  \condExpSmallGeneric{\maxOperator{\swapType \left[\zcb_{\shortRate_n}(T_M,T_k) - \tilde{\strike}_k \right]}}{t}{\Q_{\shortRate_m}^{T_M}} \nonumber \\
        &= \sum_{k=1}^m w_k  \tradeVal^{\text{ZCBO}}_{\shortRate_n}(t,T_M,T_k,\tilde{\strike}_k,\swapType),\label{eq:swaption10HW}
    \end{align}
    where ZCB option $\tradeVal^{\text{ZCBO}}_{\shortRate_n}(\cdot)$ is given in Proposition~\ref{prop:HWZCBO}.
\end{proof}

\subsection{Model calibration} \label{app:HWCalibration}
As market data for the HW volatility calibration we use a strip of ATM co-terminal swaptions with maturity $\TCot$.
We assume that the model volatility $\vol_{x}(t)$ is piece-wise constant on the intervals between the co-terminal expiries with value $\vol_{x,c}$, i.e.,
\begin{align}
  \vol_{x}(t)
    &\ldef \left\{
    \begin{array}{ll}
      \vol_{x,1} & \text{for } t\in\left(0, \TExpCot_1\right], \\
      \vol_{x,2} & \text{for } t\in\left(\TExpCot_1, \TExpCot_2\right], \\
      \vdots &  \\
      \vol_{x,\NCot} & \text{for } t\in\left(\TExpCot_{\NCot-1}, \TExpCot_{\NCot}\right].
    \end{array}
    \right. \label{eq:piecewiseVolHW}
\end{align}
Furthermore, we use flat extrapolation in volatility on both ends.

The calibration of the piece-wise constant model volatility is typically done for a given, constant, mean-reversion parameter.
To match the market quotes, we use a bootstrapping approach with local and sequential optimization, where the volatility is calibrated period by period.
The calibration procedure is summarized in Algorithm~\ref{algo:HWCalib}.

\begin{algorithm}[h!]
    \footnotesize
    \SetAlgoLined 
    \DontPrintSemicolon 


    \KwIn{Market implied volatility surface, co-terminal maturity $\TCot$, HW mean reversion $\theta_n$}
    \KwOut{Calibrated $\vol_{x}(t)$}
    \BlankLine
    
    Initialize the HW process $\shortRate_n$ with mean reversion $\theta_n$ as per Equation~\eqref{eq:HWSDE1} \;
    Initialize the volatility pillar dates\;
    
    Set calibration tolerance $\varepsilon$ \;

    \For{$c \in \{1,2,\ldots,\NCot\}$}{
        Select expiry $\TExpCot_c$, tenor $\TTenCot_c$ and corresponding ATM strike $\strike_{c}$ \;
        Retrieve market implied volatility $\impliedVolMkt_{c}$ and convert to price $\tradeValMkt_{c}$ \;
        Create an objective function $\objFun$, e.g., $ \objFun(\vol_{x,c}) \ldef \left| \tradeValMkt_{c} - \tradeValMdl_{\shortRate_n; c}(\vol_{x,c}) \right| $ \;
        Use $\vol_{x,c-1}$ as initial guess for $\vol_{x,c}$\;
        Keep $\vol_{x,1}, \ldots, \vol_{x,c-1}$ fixed, we calibrate $\vol_{x,c}$ with the current calibration instrument \;
        \While{$\objFun(\vol_{x,c}) > \varepsilon$}{
            Newton-Raphson step to update $\vol_{x,c}$ \;
            Compute error metric $\objFun(\vol_{x,c})$ to assess the model fit to the market data \;
        }
    }
    \caption{HW volatility calibration algorithm}
    \label{algo:HWCalib}
\end{algorithm}

The mean-reversion parameter can then be used to get a good overall fit to the ATM slice of the market volatility surface~\cite{PuetterRenzitti202011a}.
We fit the mean reversion such that all ATM swaption points for all expiries and all tenors (also the non-co-terminal ones) are fit best in a Mean Squared Error (MSE) sense.
A numerical solver such as Brent's algorithm can be used to obtain the fit.
Inside each iteration of this optimization, the model volatilities are bootstrapped as described in Algorithm~\ref{algo:HWCalib}.
Alternatively, one could try to calibrate the mean-reversion parameter by minimizing the MSE for only the ATM swaptions points that fall approximately under the counter diagonal~\cite{PuetterRenzitti202011b}.

\section{Bermudan swaptions} \label{app:bermudanSwaption}

\subsection{Option pricing} \label{app:bermudanSwaptionPricing}

The value of a Bermudan swaption, i.e., $\tradeVal^{\text{BS}}$, can be seen as an optimal stopping problem, where the holder of the option is looking for the optimal exercising strategy.
At exercise, the option holder has the right to enter into the underlying swap $\tradeVal^{\text{Swap}}$.
The payoff of the option is given by $\payoff(t; \shortRate(t)) = \maxOperator{\tradeVal^{\text{Swap}}(t; \shortRate(t))}$.
The value of the option is given by the risk-neutral expectation of the discounted payoff of entering a swap at the optimal stopping time $\tau$ in a predetermined set of exercise dates $\mathcal{T}_E = \{T_{E,1}, \ldots, T_{E,n}\}$:
\begin{align}
  \tradeVal^{\text{BS}}(t_0)
    &= \sup_{\tau \in \mathcal{T}_E} \condExpSmall{\expPower{-\int_{t_0}^{\tau}\shortRate(s) \ds} \maxOperator{\tradeVal^{\text{Swap}}(\tau; \shortRate(\tau))}}{t_0}. \nonumber
\end{align}
At the last exercise date, the Bermudan swaption collapses to a European swaption.

In practice, the option pricing problem is typically solved using a dynamic problem formulation.
We tackle the option pricing problem using the Least Squares Monte Carlo (LSMC) technique with a forward phase and a backward phase.
In particular we use the Longstaff-Schwartz method~\cite{LongstaffSchwartz200101}, where only the ITM paths at $T_{E,i}$ are used to calibrate the regression, and the regression is not used to value the product itself but only to approximate the exercise rule.
In this way, the regression functions only need to provide an accurate approximation of the continuation values near the exercise boundary rather than on the whole domain.
To avoid perfect foresight bias, which happens when the exercise criteria are computed with the same Monte-Carlo paths as the exercise values, we improve the LSM algorithm by running another, independent simulation to value the early-exercise product, using the exercise rule from the initial simulation.

Let $\tradeVal^{\text{BS}}(t; x)$ denote the time $t$ value of the option given $\shortRate(t) = x$, assuming no previous exercise.
The Bermudan swaption valuation algorithm as introduced in Section~\ref{sec:resultsExposuresBermudanSwaption} can be summarized as follows:
\begin{enumerate}
  \item Forward phase: simulate $M_{\tradeVal}$ short-rate paths, i.e., $\shortRate_j(t)$ for $t \in [t_0, T_{E,n}]$ for all $j$ paths.
  \item Backward phase:
      at final exercise date $T_{E,n}$, the value is given by the option payoff (as the continuation value is zero), i.e., 
      \begin{align}
        \tradeVal^{\text{BS}}(T_{E,n}; x)
          &= \payoff(T_{E,n}; x) = \maxOperator{\tradeVal^{\text{Swap}}(T_{E,n}; x)}. \label{eq:bermSwapPayoff} 
      \end{align}
      This directly gives the exercise rule at $T_{E,n}$.
      Iteratively for all previous exercise dates $T_{E,i} < T_{E,n}$:
      \begin{enumerate}
        \item Assume that the option value at the next exercise date $\tradeVal^{\text{BS}}(T_{E,i+1}; x)$ is available.
        \item The option value is given by the maximum of exercising or continuing:
          \begin{align}
            \tradeVal^{\text{BS}}(T_{E,i}; x)
              &= \maxFun{\payoff(T_{E,i}; x)}{c(T_{E,i}; x)} 
              = \maxFun{\maxOperator{\tradeVal^{\text{Swap}}(T_{E,i}; x)}}{c(T_{E,i}; x)}, \nonumber  
          \end{align}
          where $c(T_{E,i}; x)$ is the continuation value, which is the discounted option value from the next exercise date, i.e.,
          \begin{align}
            c(T_{E,i}; x)
              &= \E\left[ \left.  \expPower{-\int_{T_{E,i}}^{T_{E,i+1}}\shortRate(s) \ds} \tradeVal^{\text{BS}}(T_{E,i+1}; \shortRate(T_{E,i+1})) \right| \shortRate(T_{E,i}) = x\right]. \nonumber
          \end{align}
        This conditional expectation is a random variable itself.
        A brute-force approach would be to use nested Monte Carlo to approximate this continuation value.
        However, we decide to use LSM where we approximate the continuation value $c(T_{E,i}; x)$ as the regression of option value $\tradeVal^{\text{BS}}(T_{E,i+1}; \shortRate(T_{E,i+1}))$ on the current state $x$.
        Then we value the option by a backward iteration (recursion).
        \item Identify all ITM paths by computing $\payoff(T_{E,i}; \shortRate_j(T_{E,i})) = \maxOperator{\tradeVal^{\text{Swap}}(T_{E,i}; \shortRate_j(T_{E,i}))}$ $\forall j$ paths.
        \item We approximate the continuation value through a linear combination of basis functions $\psi_{k}(x)$:
          \begin{align}
            c(T_{E,i}; x)
              &= \sum_{l} \beta_{il} \psi_l(x)
              = \beta_i^\top \psi(x), \nonumber
          \end{align}
          where we use a least-squares regression over all ITM paths $k$ to estimate the coefficients $\beta_{il}$ of this approximation, yielding estimates $\hat{\beta}_{il}$.
          For all ITM paths $k$, compute $\expPower{-\int_{T_{E,i}}^{T_{E,i+1}}\shortRate_k(s) \ds} \hat{\tradeVal}^{\text{BS}}(T_{E,i+1}; \shortRate_k(T_{E,i+1}))$ and regress on $\shortRate_k(T_{E,i})$.
          These estimated coefficients now define the continuation value approximation $\hat{c}(T_{E,i}; x) = \hat{\beta}_{i}^{\top} \psi(x)$.
        \item For all ITM paths $k$, we approximate the exercise rule at $T_{E,i}$ using the payoff $\payoff(T_{E,i}; \shortRate_k(T_{E,i}))$ and the approximate continuation value $\hat{c}(T_{E,i}; \shortRate_k(T_{E,i}))$.
          For the paths where we do not exercise, we plug in the path-wise discounted value from the next time point, i.e., continuation is the base-line scenario.
          So for all paths $j$, $\hat{c}(\cdot)$ is used to approximate the exercise decision only, and not also the cash-flows in the valuation:
          \begin{align}
            \tradeVal^{\text{BS}}(T_{E,i}; x)
              &\approx \left\{
              \begin{array}{ll}
                \payoff(T_{E,i}; x), & \payoff(T_{E,i}; x) > \hat{c}(T_{E,i}; x), \\
                \expPower{-\int_{T_{E,i}}^{T_{E,i+1}}\shortRate(s) \ds} \tradeVal^{\text{BS}}(T_{E,i+1}; \shortRate_j(T_{E,i+1})), & \payoff(T_{E,i}; x) \leq \hat{c}(T_{E,i}; x)  
              \end{array}
              \right. \nonumber  
          \end{align}
          This then also gives the exercise boundary $x^*$, which is the regression variable value (e.g., short rate or swap rate) for which the first exercise is done, i.e., the first value for which the immediate payoff was bigger than the regressed continuation value.
        \item Once the first exercise date $T_{E,1}$ has been reached, the approximate Bermudan swaption value at $t_0$ is given by the discounted Monte Carlo average over all paths
          \begin{align}
            \tradeVal^{\text{BS}}(t_0)
              &\approx \frac{1}{M} \sum_{j = 1}^{M} \expPower{-\int_{t_0}^{T_{E,1}}\shortRate_j(s) \ds} \tradeVal^{\text{BS}}(T_{E,1}; \shortRate_j(T_{E,1})). \nonumber
          \end{align}
  \end{enumerate}
\end{enumerate}

\subsection{Exposures} \label{app:bermudanSwaptionExposures}

The key to computing exposures is to obtain future market values $\tradeVal(t_i)$ at all monitoring dates $t_i$ for all Monte-Carlo scenarios.
A naive approach for Bermudan swaptions would be to re-use the original LSMC method starting at $r_j(t_i)$ for generating future market scenarios to obtain the future market value $\tradeVal(t_i; \shortRate_j(t_i))$ at monitoring dates before the option is exercised on path $j$.
From a computational perspective, this method is too expensive to be used in practice.
Therefore, we use a modified version of the aforementioned Bermudan swaption pricing algorithm to obtain exposures.

For each monitoring date $t_i$, we compute the path-wise future market value $\tradeVal(t_i; \shortRate_j(t_i))$ $\forall j$ using the collocation-based approximation technique from Deelstra~\etal\cite{DeelstraGrzelakWolf202304}, using Lagrangian interpolation to get an approximation over the entire domain.
The idea is to construct an accurate approximation based on only a few valuations of the expensive nested LSMC algorithm at points based on the moments of the underlying distribution, where Corollary~\ref{crllry:randMoments} is used to obtain the moments of the rHW model.
This approximation can then be evaluated swiftly for all monitoring dates and for all paths to obtain an approximate future market value.

The exposure $\EPE(t, t_i)$ of a Bermudan swaption is approximated using a regression:
\begin{enumerate}
  \item First pass: estimating regression functions at all early-exercise dates $\mathcal{T}_E = \{T_{E,1}, \ldots, T_{E,n}\}$.
  \begin{enumerate}
    \item Forward phase: generate $M_{\tradeVal}$ Monte Carlo paths for LSMC regression as usual for early-exercise products.
    \item Backward phase: for all early-exercise products, perform the regression approximation at each early-exercise date $T_{E,i}$.
    This yields the approximate continuation value $\hat{c}(T_{E,i}; x)$ $\forall T_{E,i}$.
    In this step it is fine to only use the ITM paths to approximate the exercise boundary.
  \end{enumerate}
  \item Second pass: estimating exposures.
  \begin{enumerate}
    \item Generate an independent set of $M_{\tradeVal}$ Monte Carlo paths (this deals with foresight bias) for pricing the exposures.
    \item Using the regression from the first pass ($\hat{c}(T_{E,i}; x)$ $\forall T_{E,i}$), determine the path-wise early-exercise decision $\tau_j$ for each path $j$.
    \item For each monitoring date $t_i$, we compute the path-wise future market value $\tradeVal(t_i; \shortRate_j(t_i))$ $\forall j$ using the collocation-based approximation technique from Deelstra~\etal\cite{DeelstraGrzelakWolf202304}, using Lagrangian interpolation to get an approximation over the entire domain.
        The idea is to construct an accurate approximation based on only a few valuations of the expensive nested LSMC algorithm.
        This approximation can then be evaluated swiftly for all monitoring dates and for all paths to obtain an approximate future market value $\hat{\tradeVal}(t_i; \shortRate_j(t_i))$. 
    \begin{enumerate}
      \item Compute collocation points $x_k$ based on the moments of the distribution of the underlying short-rate.
      Typically, 5 collocation points are sufficient.
      \item For all these collocation points, run a fully nested LSMC simulation to obtain values $\tradeVal(t_i; x_k)$, re-using the calibrated exercise rule from the first pass.
      \item Use 1D Lagrangian interpolation over the collocation points to obtain $\hat{\tradeVal}(t_i; \shortRate_j(t_i))$ $\forall j$.
    \end{enumerate}
    \item Approximate the exposure from Equation~\eqref{eq:EPE} by replacing $\maxOperator{\tradeVal(t_i)}$ by the regression approximation $\maxOperator{\hat{\tradeVal}(t_i)}$.
  \end{enumerate}
\end{enumerate}

\begin{rem}[Moments for collocation points]
  The collocation points at monitoring date $t_i$ are based on moments of the underlying short-rate $\shortRate(t_i)$.

  When using the HW model as underlying, $\shortRate_n(t) \sim \N \left( \condExpSmall{\shortRate_n(t)}{0}, \condVarSmall{\shortRate_n(t)}{0}\right)$.
  In this case the raw (non-central) moments $\condExpSmall{\shortRate^m(t)}{0}$ are known and easy to compute:
  \begin{itemize}
    \item For $Z\sim \N(0,1)$ we have the following raw (non-central) moments:
    \begin{align}
      \E\left[ Z^m \right] 
        &= \left\{ 
        \begin{array}{ll}
          (m-1)!!,  & m \text{ even}, \\
          0,        & m \text{ odd}.
        \end{array}
        \right. \nonumber
    \end{align}
    \item For $X\sim \N\left(\mu, \sigma^2\right)$ the following central moments:
    \begin{align}
      \E\left[ (X-\mu)^m \right] 
        &= \E\left[ (\sigma Z)^m \right] 
        = \sigma^m \E\left[ Z^m \right]. \nonumber
    \end{align}
    Using the binomial theorem we can easily obtain an expression for the raw (non-central) moments $\E\left[ X^m \right]$.
  \end{itemize}
  For the rHW model, i.e., $\shortRate(t)$ is constructed as per Definition~\ref{def:randPDF}, then we use the result from Corollary~\ref{crllry:randMoments} to compute the desired moments.
\end{rem}    

Alternative methods to compute exposures of early-exercise products are an updated version of the LSMC algorithm for exposures~\cite{FengJainKarlssonKandhaiOosterlee201609}; the Stochastic Grid Bundling Method (SGBM) method~\cite{KarlssonJainOosterlee201609}, which uses a different way in which the regression approximation of the continuation value is obtained; Glau~\etal\cite{GlauPachonPotz202007} use an iterative polynomial interpolation over Chebyshev nodes, where a different algorithm is proposed such that the nested simulation is avoided altogether; finally, Joshi and Kwon~\cite{JoshiKwon201610} propose to only approximate the sign of the future exposure through a regression, such that the regression only needs to be accurate close to the exposure sign change, similarly to the Longstaff-Schwartz approach for pricing.

\section{Additional results for EUR} \label{app:resultsEUR}

Here, we present additional results for the EUR market data from Figure~\ref{fig:EURVolSurfaceTenor10Y}.

\subsection{Calibration} \label{app:resultsEURCalibration}

From the calibration results presented in Figure~\ref{fig:calibrationEUR} it can be seen that compared to the USD calibration results from Figure~\ref{fig:calibrationUSD}, the rHW quadrature points are substantially closer to the calibrated HW mean-reversion parameter $a_x$
Furthermore, the HW and rHW model volatilities are not far apart.

\begin{figure}[ht!]
  \centering
  \begin{subfigure}[b]{0.435\linewidth}
    \includegraphics[width=\linewidth]{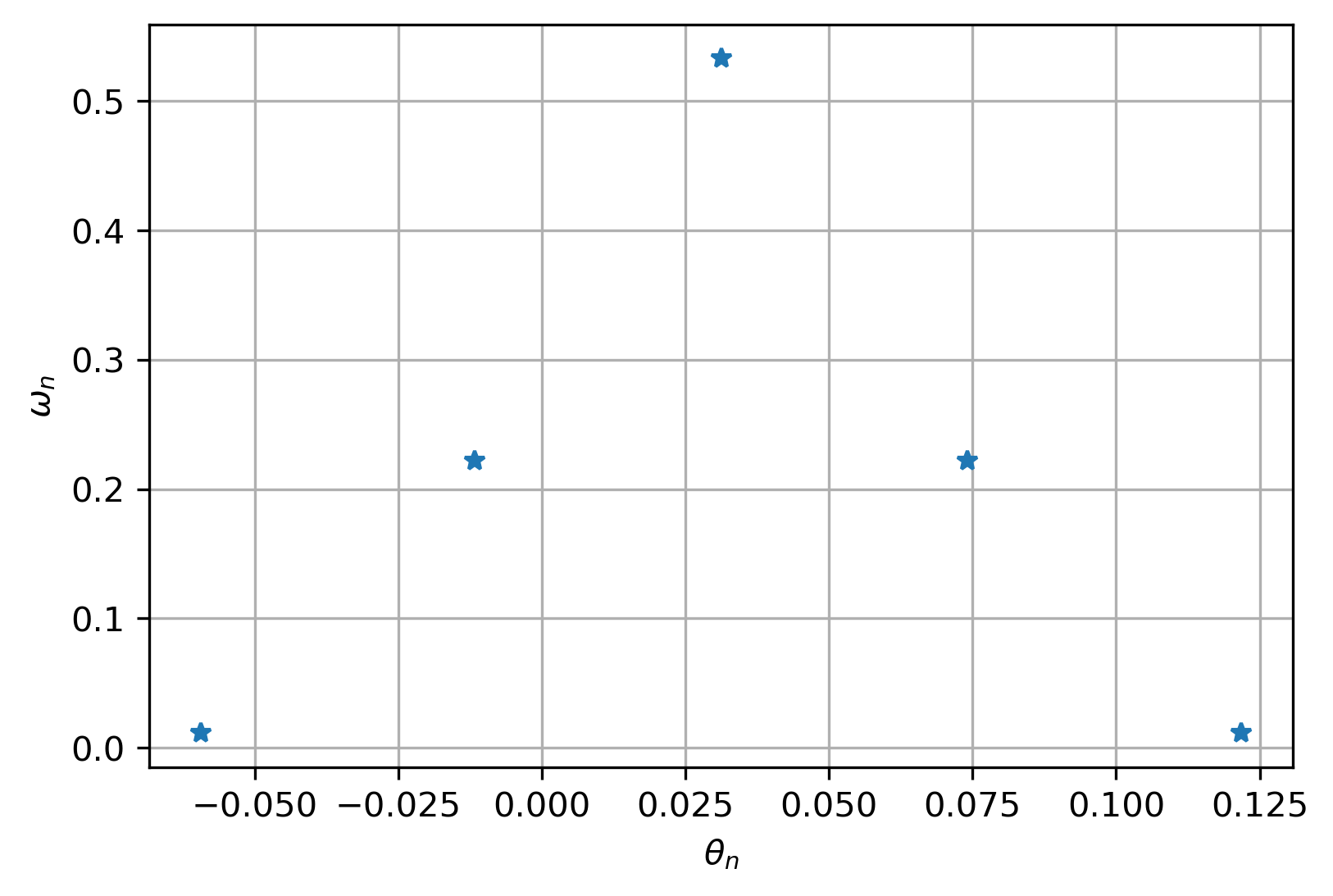}
    \caption{Quadrature points for rHW.}
    \label{fig:quadraturePointsEUR}
  \end{subfigure}
  \begin{subfigure}[b]{0.45\linewidth}
    \includegraphics[width=\linewidth]{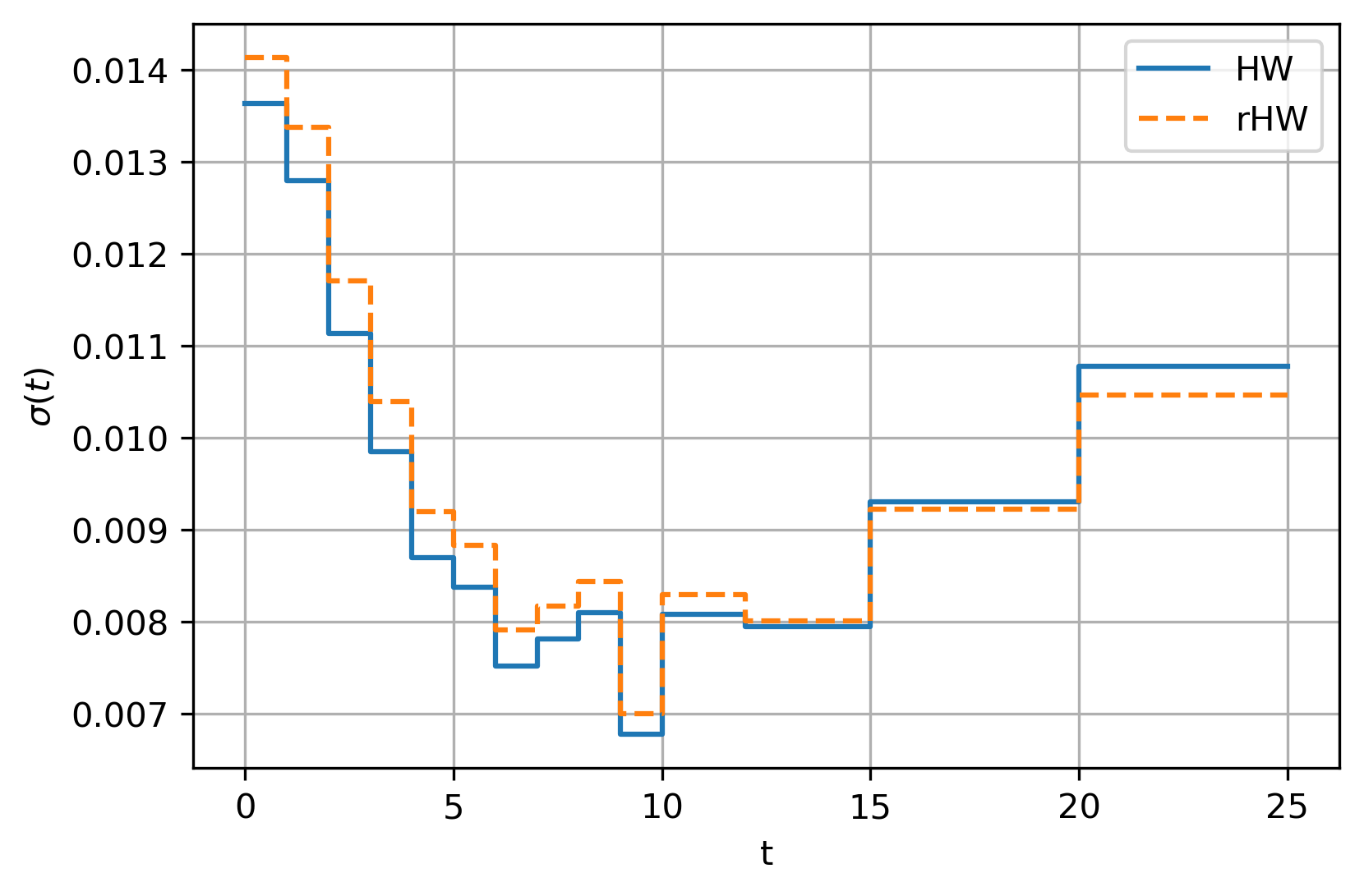}
    \caption{Model volatilities.}
    \label{fig:modelVolatilitiesEUR}
  \end{subfigure}
  \caption{Calibration of the EUR market data from Figure~\ref{fig:VolSurfaceCoterminal} for $\NMix=5$ and quadrature points for $\N (\hat{a}, \hat{b}^2 )$ with $\hat{a} = 0.031220$ and $\hat{b} = 0.031681$.
    The calibrated HW mean-reversion is $a_x = 0.019659$.}
  \label{fig:calibrationEUR}
\end{figure}

In Figure~\ref{fig:EURVolSlicesNormalRandCotsmilesCalib} we see that there is less volatility smile compared to USD case, and in the EUR case there is only volatility smile for the short expiries, and the smile effect swiftly fades when the expiry increases. 
All in all, there is mainly skew in this market data set.
As a result, the HW implied volatilities are significantly closer to the market implied volatilities than in the USD case.

\begin{figure}[ht!]
  \centering
  \begin{subfigure}[b]{0.45\linewidth}
    \includegraphics[width=\linewidth]{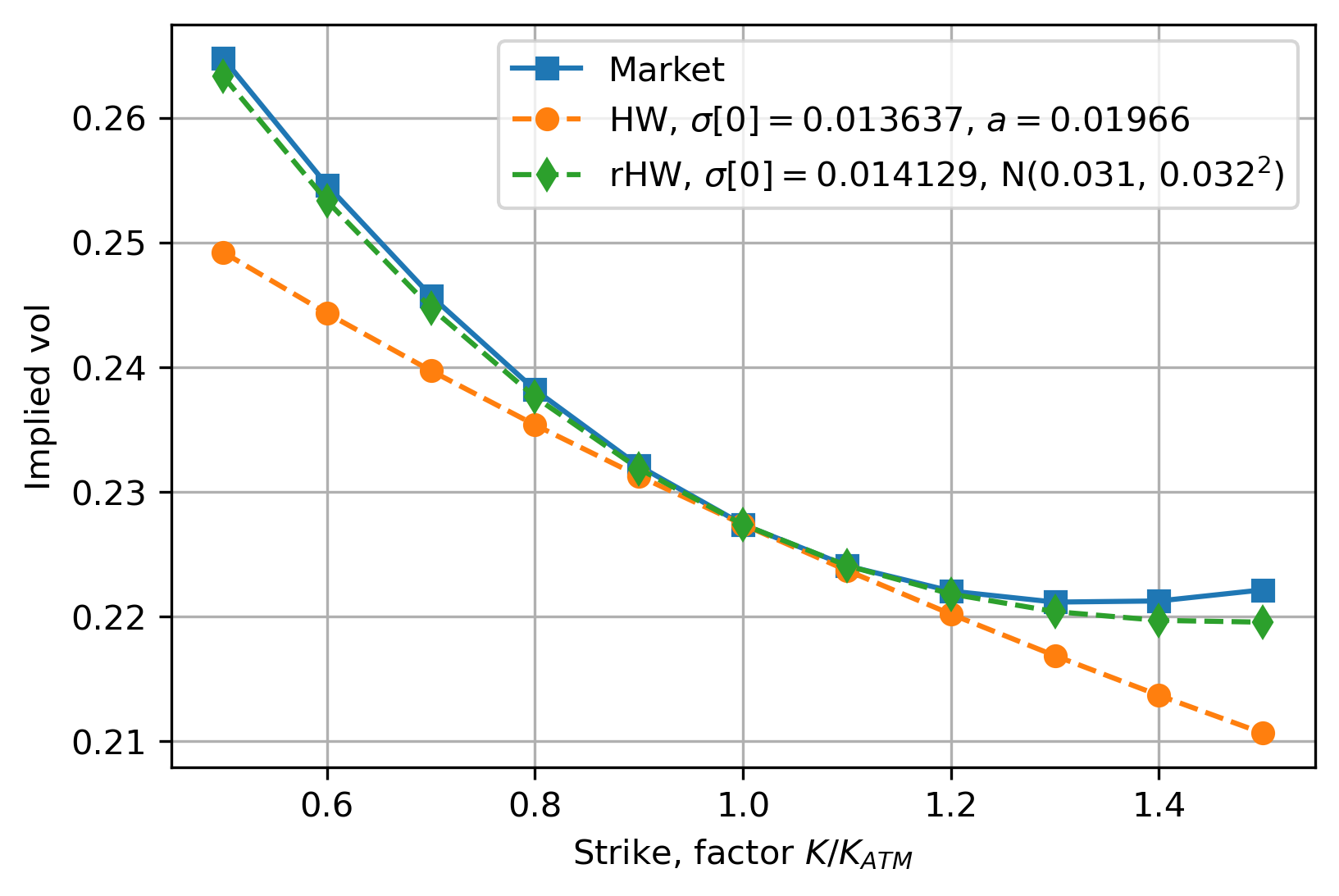}
    \caption{1Y expiry, 29Y tenor.}
    \label{fig:EUR1yExpiry29yTenorVolSliceNormalRandCotsmilesCalib}
  \end{subfigure}
  \begin{subfigure}[b]{0.45\linewidth}
    \includegraphics[width=\linewidth]{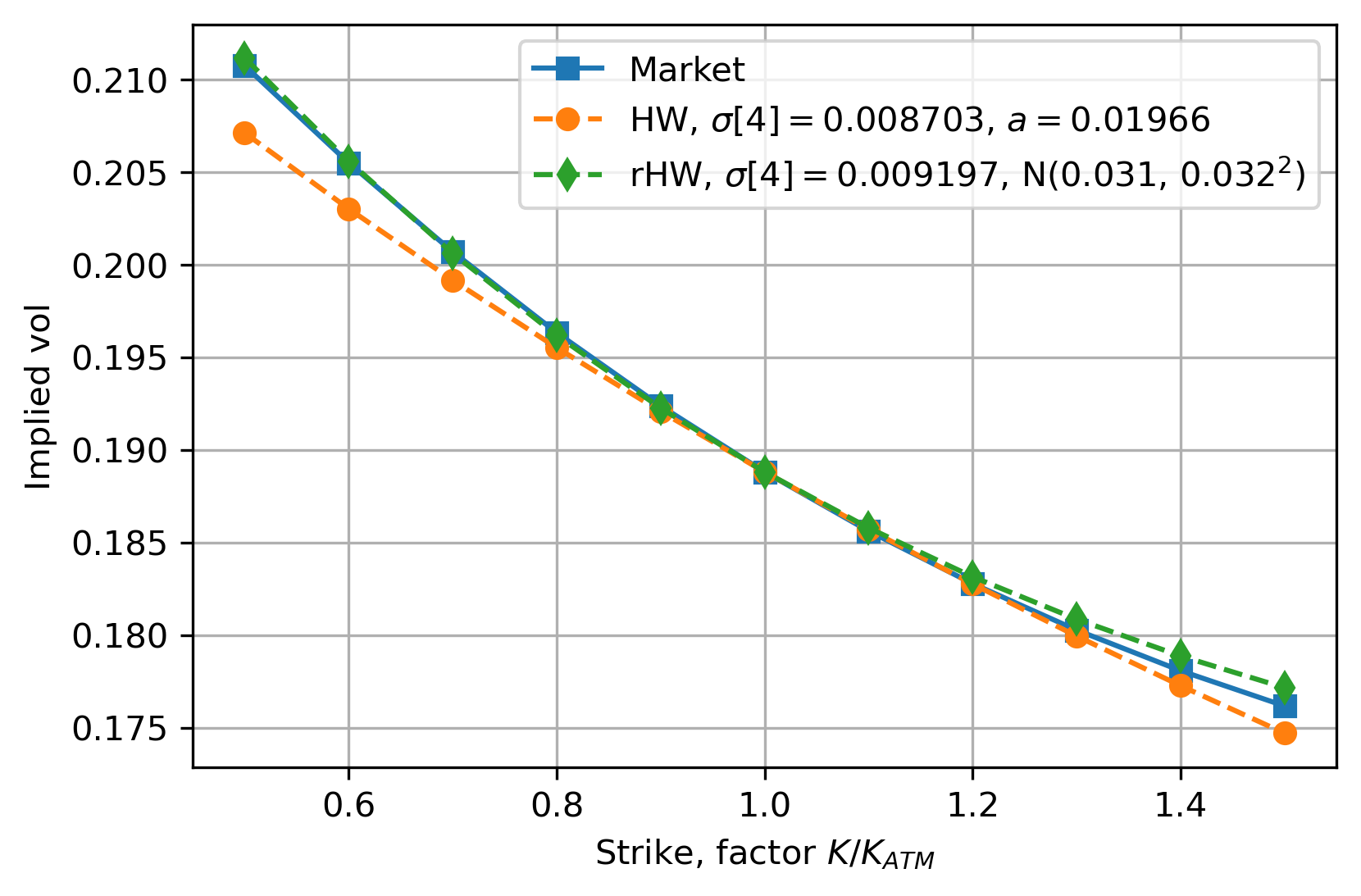}
    \caption{5Y expiry, 25Y tenor.}
    \label{fig:EUR5yExpiry25yTenorVolSliceNormalRandCotsmilesCalib}
  \end{subfigure}
  \begin{subfigure}[b]{0.45\linewidth}
    \includegraphics[width=\linewidth]{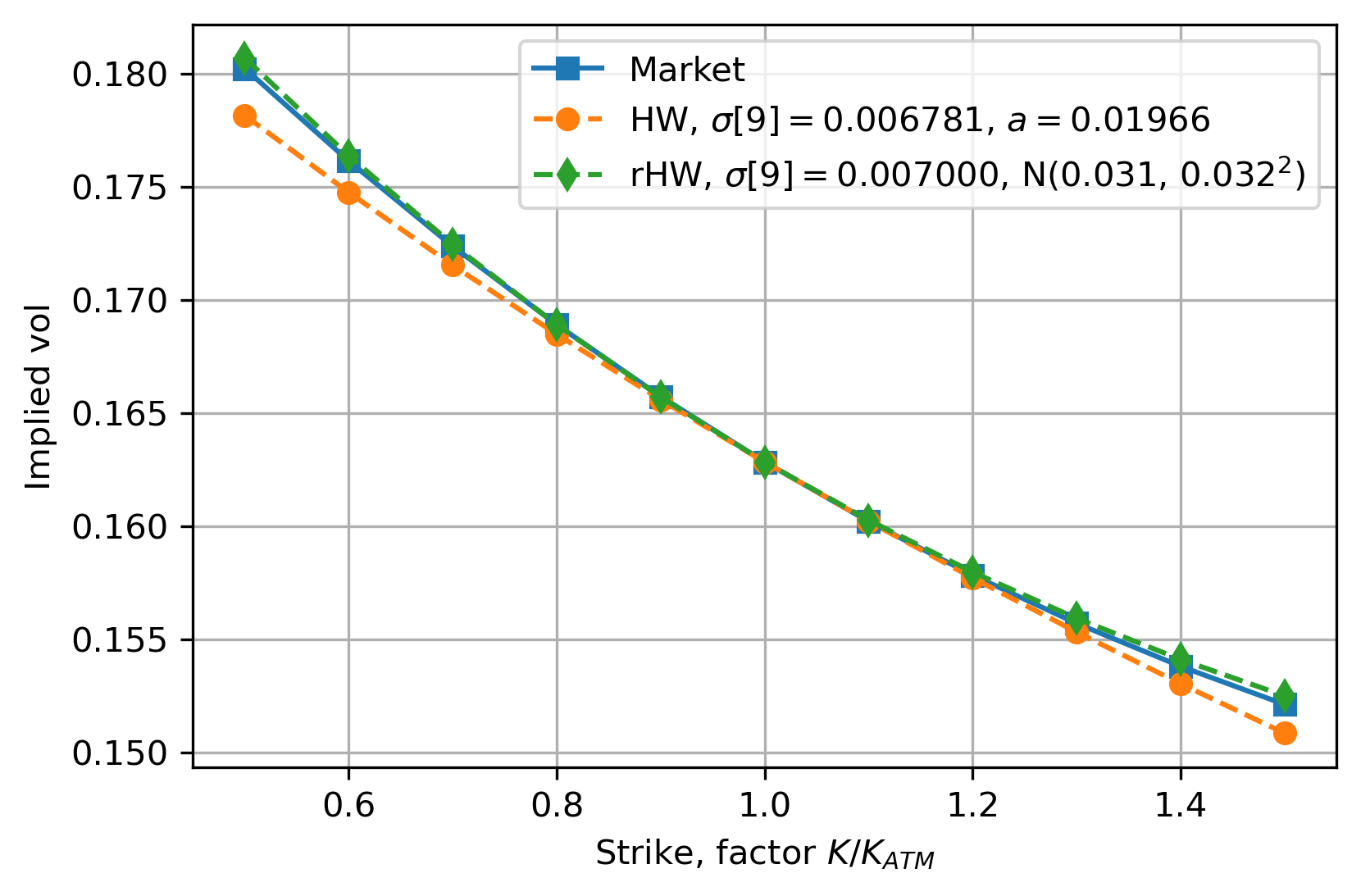}
    \caption{10Y expiry, 20Y tenor.}
    \label{fig:EUR10yExpiry20yTenorVolSliceNormalRandCotsmilesCalib}
  \end{subfigure}
  \begin{subfigure}[b]{0.45\linewidth}
    \includegraphics[width=\linewidth]{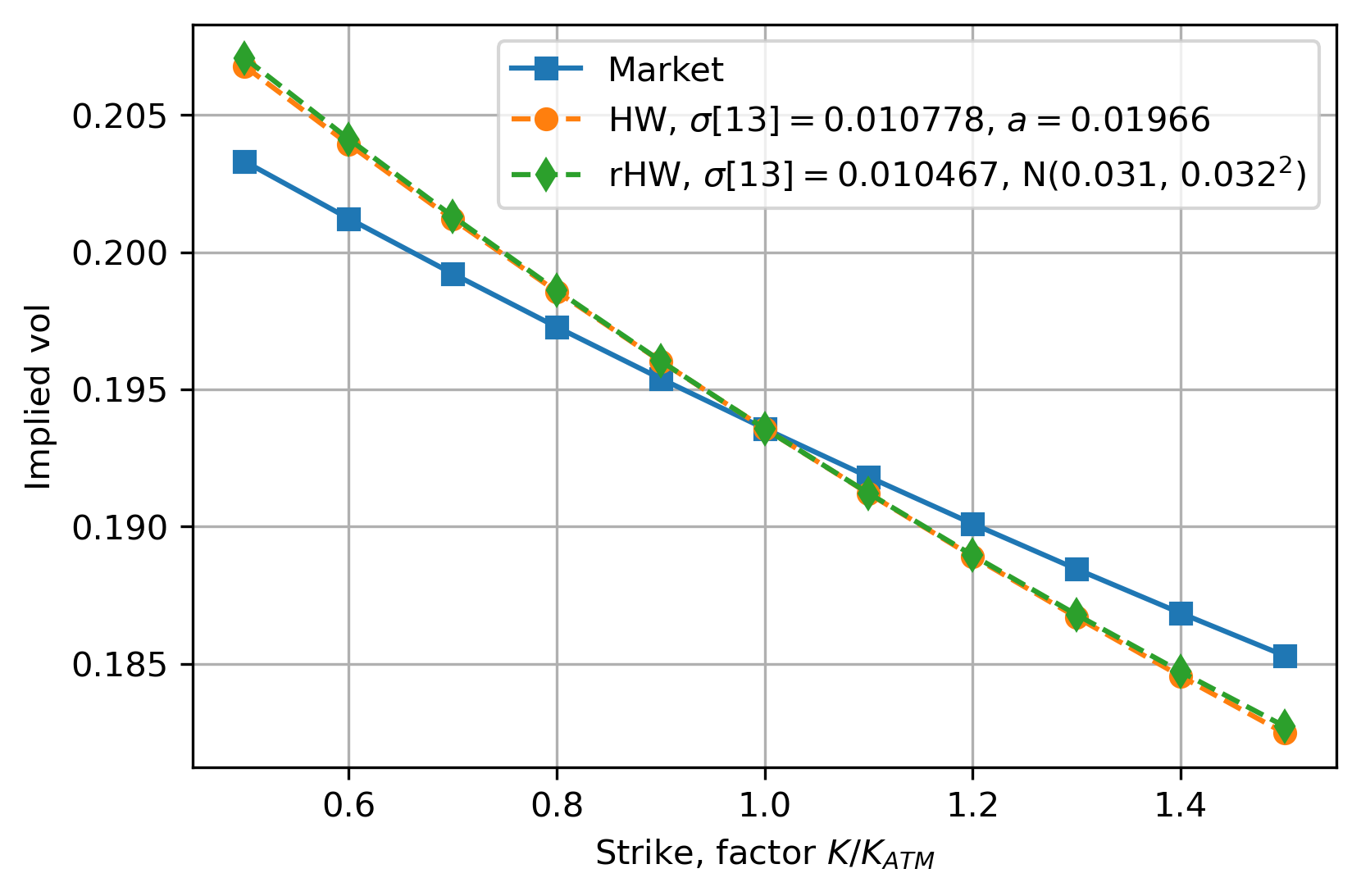}
    \caption{25Y expiry, 5Y tenor.}
    \label{fig:EUR25yExpiry5yTenorVolSliceNormalRandCotsmilesCalib}
  \end{subfigure}
  \caption{Market and model swaption implied volatilities for the calibrated parameters from Figure~\ref{fig:calibrationEUR}.}
  \label{fig:EURVolSlicesNormalRandCotsmilesCalib}
\end{figure}

Figure~\ref{fig:EURNormalRandCotsmilesCalibImpvolError} shows that the HW and rHW models have comparable error profiles, also for the non-coterminal ATM points.
This is a consequence of the limited amount of smile in the market, that results in quadrature points $\left\{\omega_n,\theta_n\right\}_{n=1}^{\NMix}$ where the $\theta_n$ are close to the HW mean-reversion $a_x$, as well as the comparable model volatilities in Figure~\ref{fig:modelVolatilitiesEUR}.
As a result, we also expect that this will translate into a smaller difference between HW and rHW exposures and $\xva$s than in the USD case.
    
\begin{figure}[ht!]
  \centering
  \begin{subfigure}[b]{0.45\linewidth}
    \includegraphics[width=\linewidth]{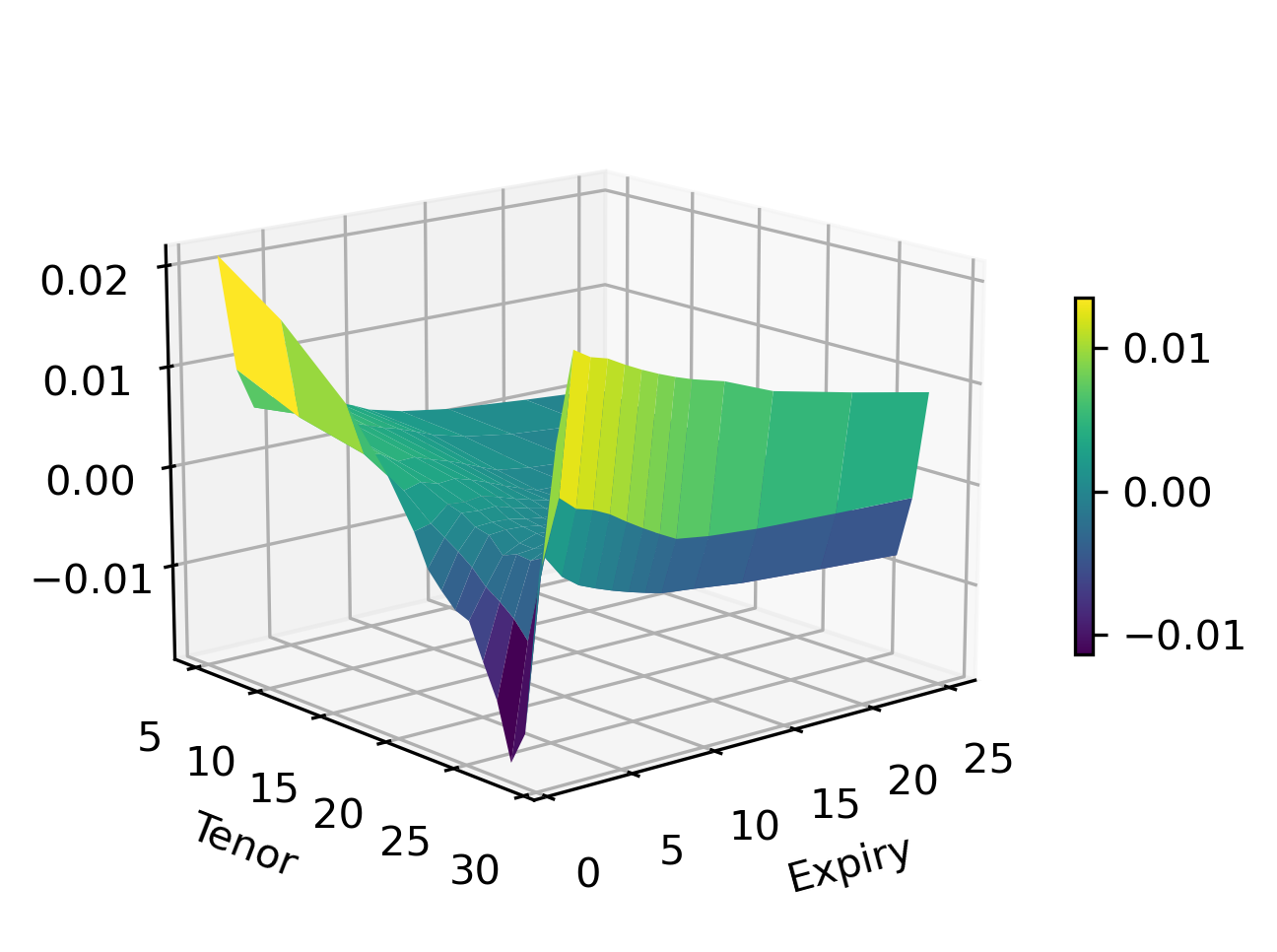}
    \caption{HW model.}
    \label{fig:EURNormalRandCotsmilesCalibImpvolErrorHW}
  \end{subfigure}
  \begin{subfigure}[b]{0.45\linewidth}
    \includegraphics[width=\linewidth]{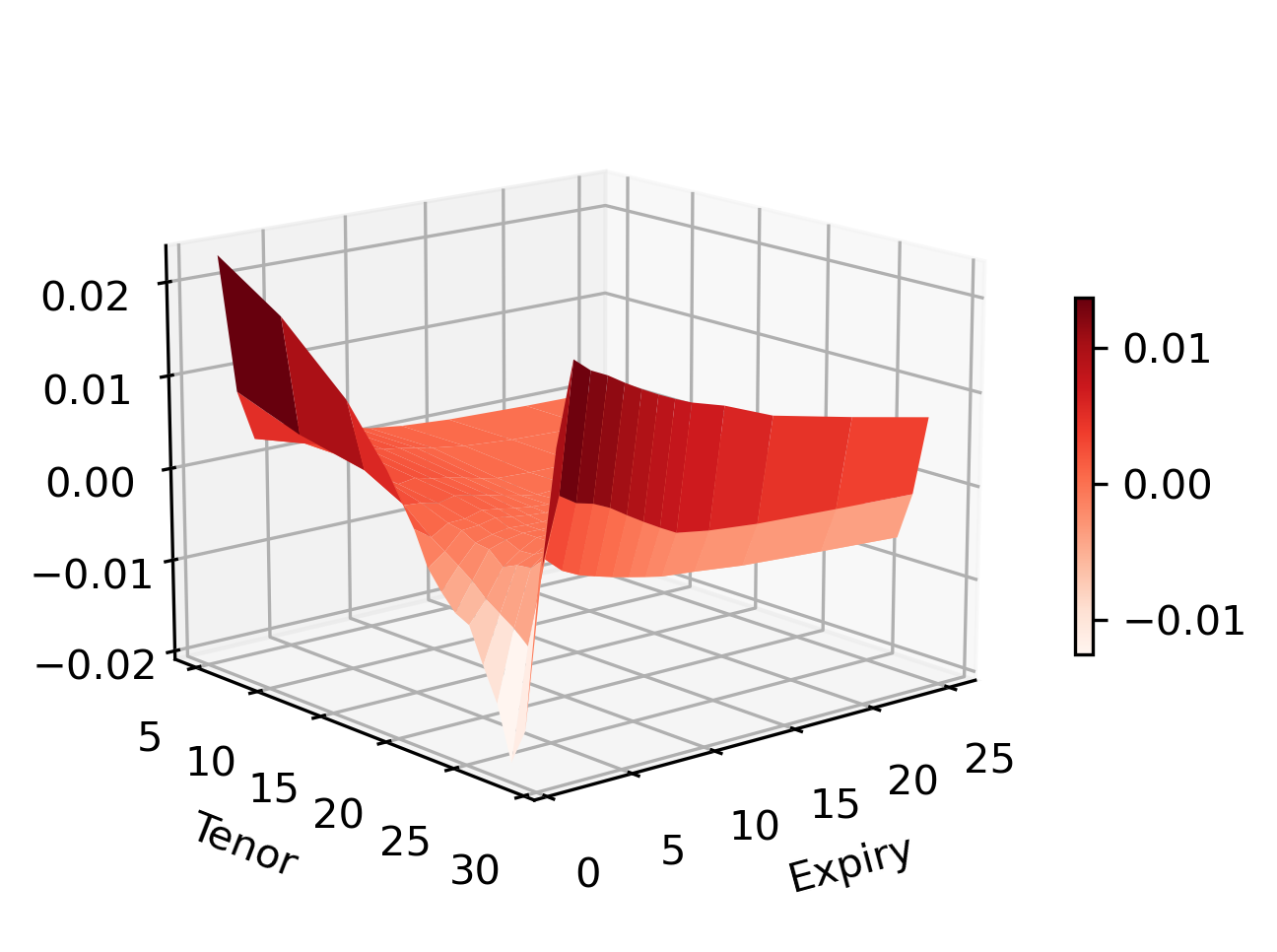}
    \caption{rHW model.}
    \label{fig:EURNormalRandCotsmilesCalibImpvolErrorrHW}
  \end{subfigure}
  \caption{Implied volatility calibration error for all ATM points of the EUR volatility surface.}
  \label{fig:EURNormalRandCotsmilesCalibImpvolError}
\end{figure}

\subsection{Exposures and $\xva$s} \label{app:resultsEURExposures}

\subsubsection{Swaps} \label{app:resultsEURExposuresSwaps}
As expected, the exposures presented in Figure~\ref{fig:RAnDEURNormalCotsmilesReceiverSwapATMExposureExpiry0} show that the effects are the same as in the USD case, albeit with a less significant impact.

\begin{figure}[ht!]
  \centering
  \begin{subfigure}[b]{0.45\linewidth}
    \includegraphics[width=\linewidth]{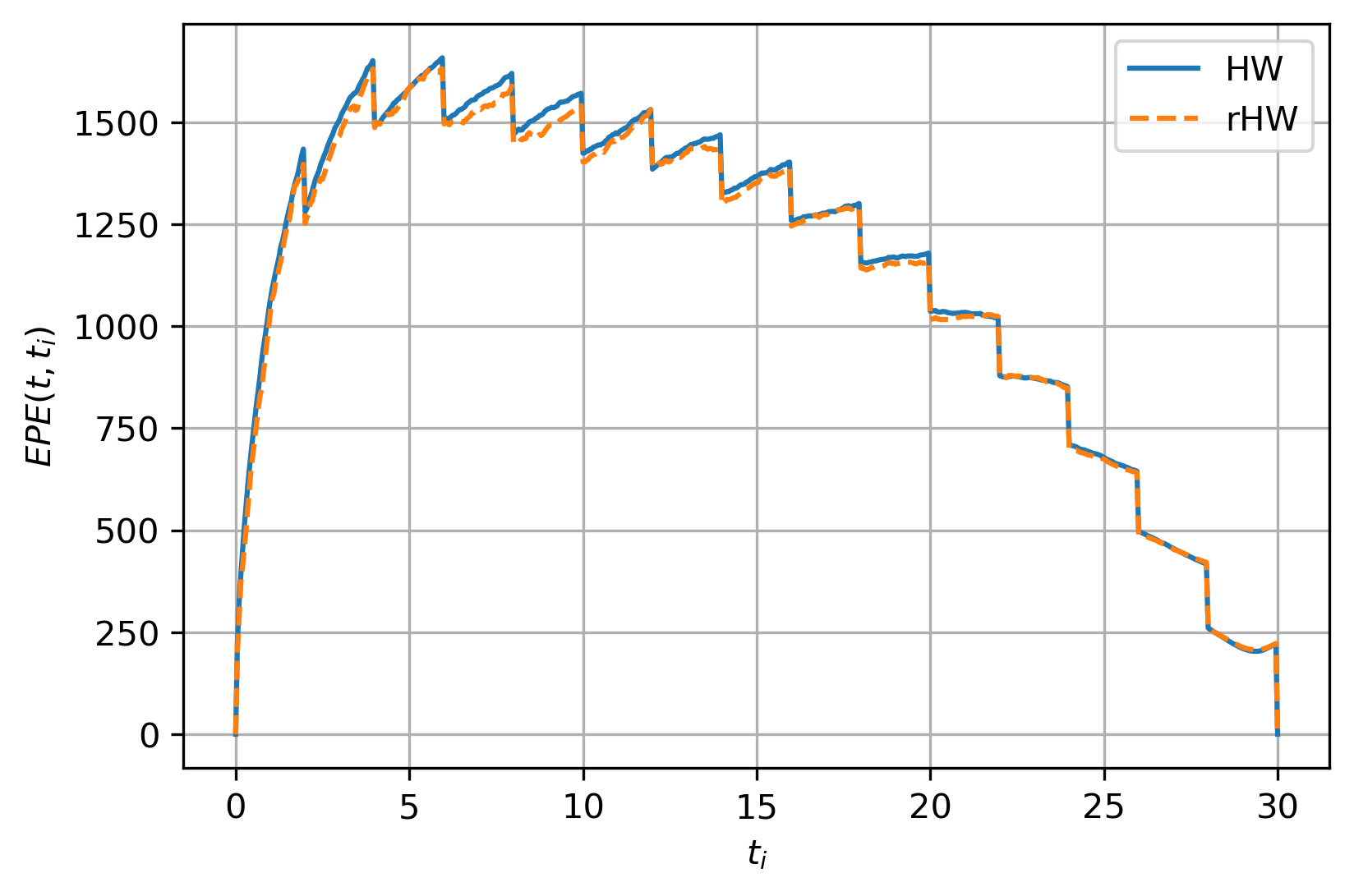}
    \caption{$\EPE(t, t_i)$}
    \label{fig:RAnDEURNormalCotsmilesReceiverSwapATMExposureExpiry0EPE}
  \end{subfigure}
  \begin{subfigure}[b]{0.45\linewidth}
    \includegraphics[width=\linewidth]{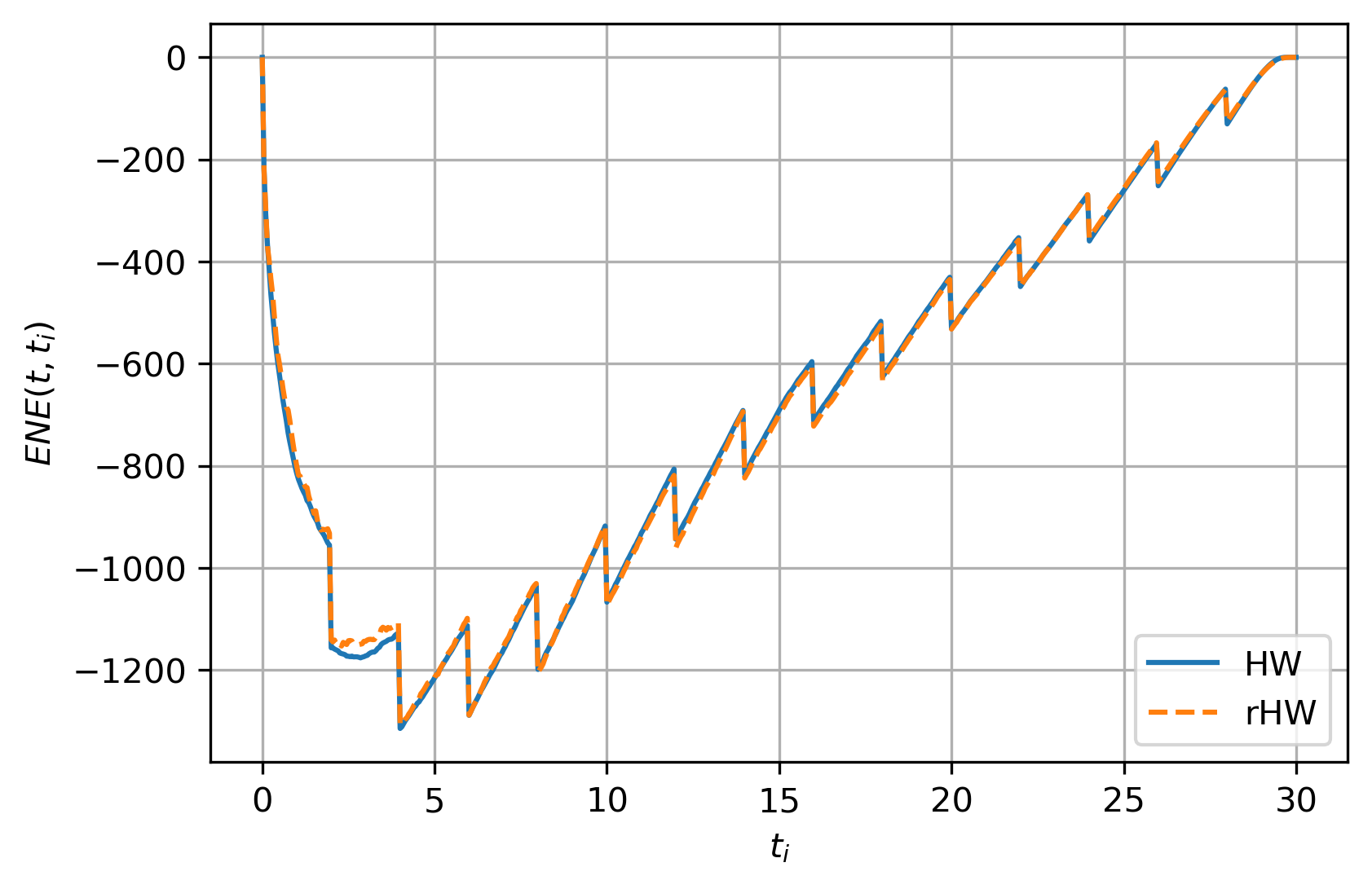}
    \caption{$\ENE(t, t_i)$}
    \label{fig:RAnDEURNormalCotsmilesReceiverSwapATMExposureExpiry0ENE}
  \end{subfigure}
  \begin{subfigure}[b]{0.45\linewidth}
    \includegraphics[width=\linewidth]{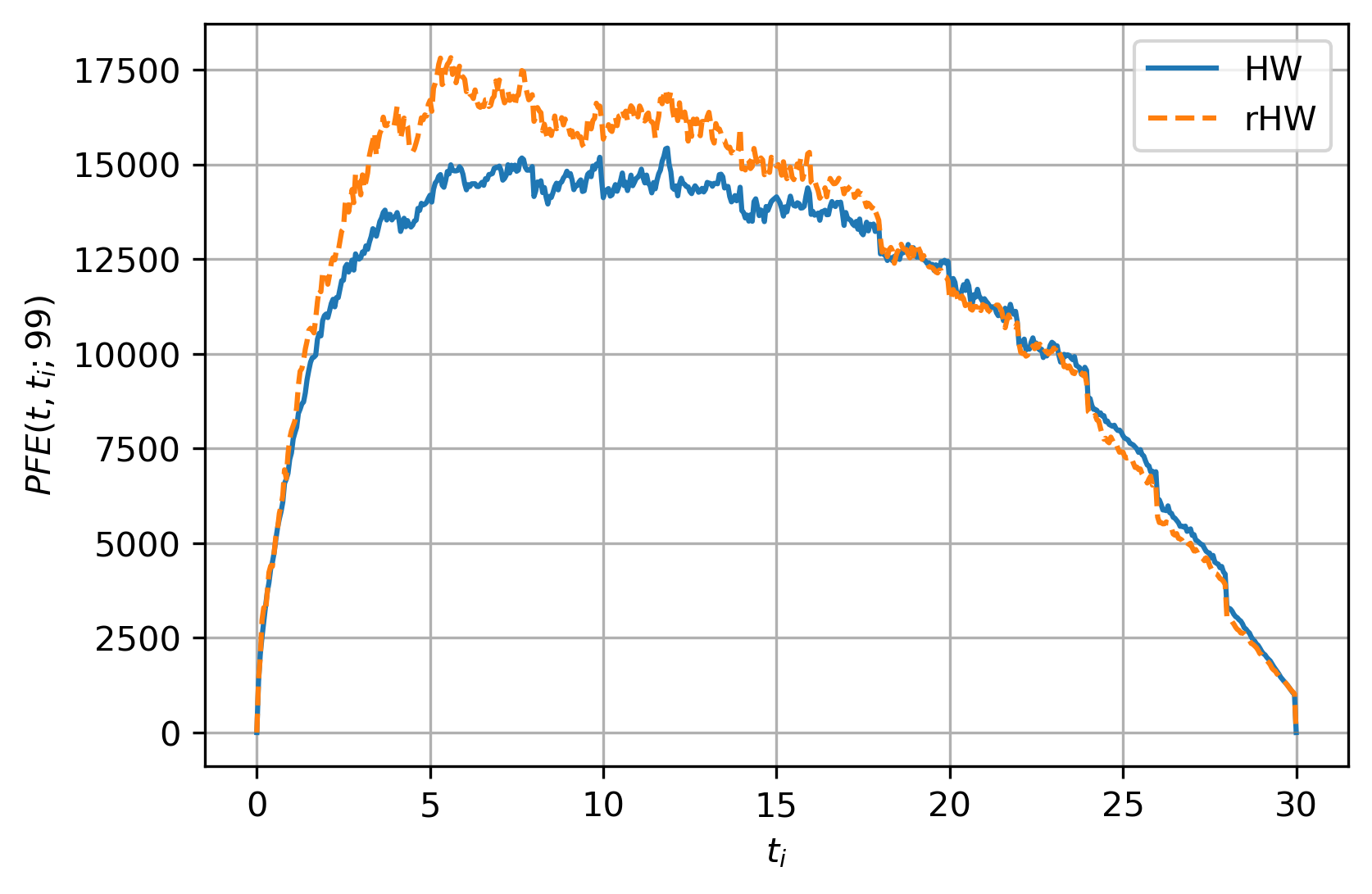}
    \caption{$\PFE(t, t_i; 99)$}
    \label{fig:RAnDEURNormalCotsmilesReceiverSwapATMExposureExpiry0PFE(99)}
  \end{subfigure}
  \begin{subfigure}[b]{0.45\linewidth}
    \includegraphics[width=\linewidth]{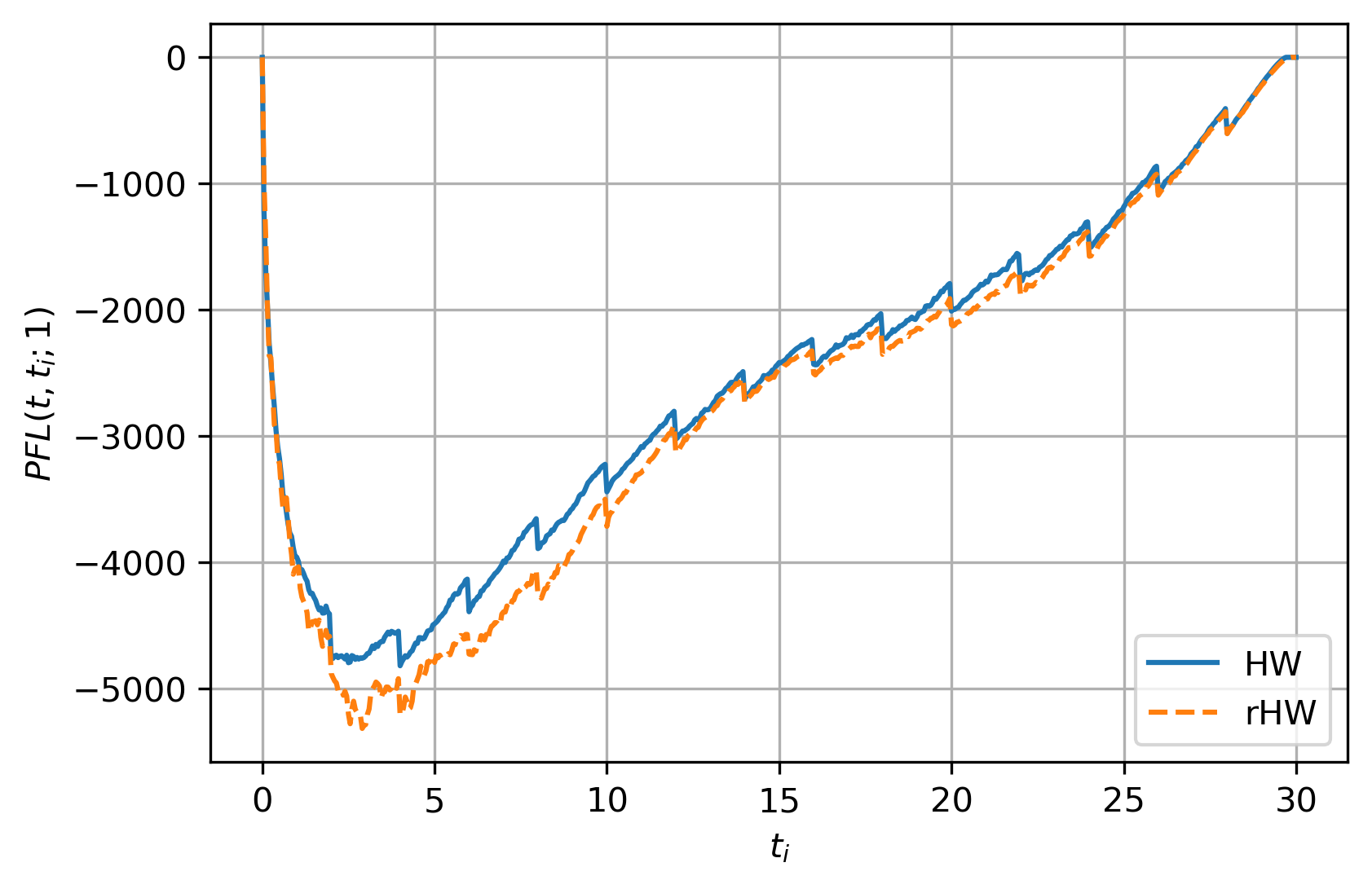}
    \caption{$\PFL(t, t_i; 1)$}
    \label{fig:RAnDEURNormalCotsmilesReceiverSwapATMExposureExpiry0PFL(1)}
  \end{subfigure}
  \caption{Comparing exposures for an ATM receiver swap ($\strike = \strikeATM$).
  }
  \label{fig:RAnDEURNormalCotsmilesReceiverSwapATMExposureExpiry0}
\end{figure}

The same can be said for the $\xva$ metrics in Table~\ref{tab:RAnDEURNormalCotsmilesReceiverSwapExposureExpiry0}: there is less impact of the smile on the metrics, where in the EUR case we see a maximum of 3\% impact on the metrics (ignoring the OTM $\BCVA$ case due to low magnitude).
Even if the smile effect are less pronounced than in the USD example, it is still non-negligible.

\begin{table}[ht!]
\centering
\begin{tabular}{l|ll|rrr}
    Model   & $\strike$             & Moneyness     & $\CVA(t_0)$   & $\DVA(t_0)$   & $\BCVA(t_0)$  \\ \hline
    HW      & $\strikeATM$          & ATM           & 537.883       & -183.441      & 330.007       \\
    rHW     &                       &               & 530.321       & -183.049      & 323.514       \\ \hline
    HW      & $1.5\cdot\strikeATM$  & ITM           & 879.959       &  -94.056      & 718.750       \\
    rHW     &                       &               & 867.425       &  -91.322      & 709.683       \\ \hline
    HW      & $0.5\cdot\strikeATM$  & OTM           & 298.740       & -327.966      & -10.502       \\
    rHW     &                       &               & 302.685       & -333.410      & -11.358
\end{tabular}
\caption{$\xva$ metrics for the EUR receiver swap example, for various strikes.}
\label{tab:RAnDEURNormalCotsmilesReceiverSwapExposureExpiry0}
\end{table}

\subsubsection{Bermudan swaption} \label{app:resultsEURExposuresBermudanSwaption}

For the Bermudan swaption we draw similar conclusions as before, where the USD conclusions extend to the EUR case, but the smile impact is smaller.
For example, for the ATM case presented in Figure~\ref{fig:RAnDEURNormalCotsmilesReceiverBermSwaptionATMExposureExpiry0} there is a 12\% impact on $\CVA$, whereas in the USD case this was a 62\% impact.
Again, even though the smile effect is smaller, it is still non-negligible.

\begin{figure}[ht!]
  \centering
  \begin{subfigure}[b]{0.45\linewidth}
    \includegraphics[width=\linewidth]{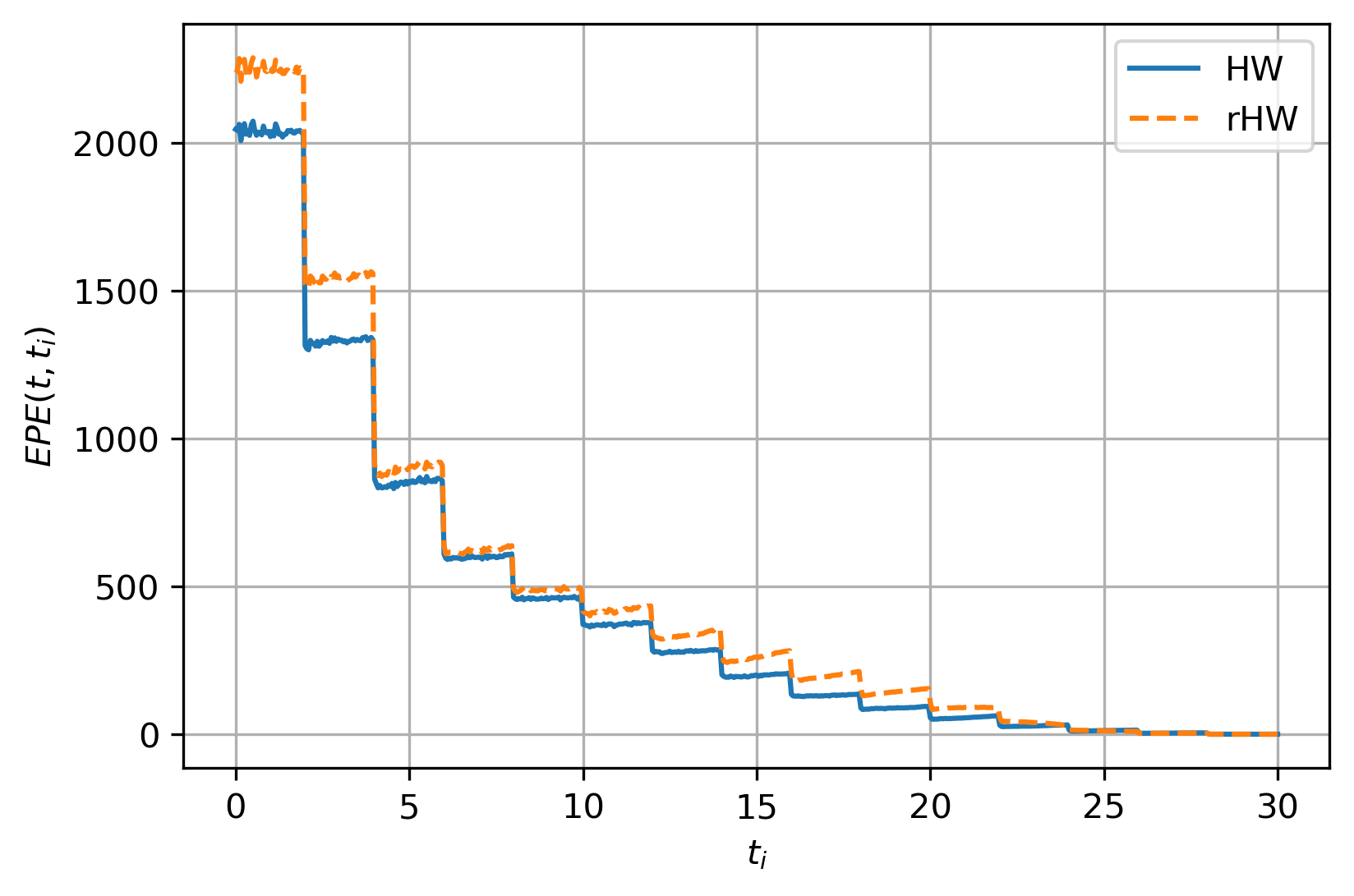}
    \caption{$\EPE(t, t_i)$}
    \label{fig:RAnDEURNormalCotsmilesReceiverBermSwaptionATMExposureExpiry0EPE}
  \end{subfigure}
    \begin{subfigure}[b]{0.45\linewidth}
    \includegraphics[width=\linewidth]{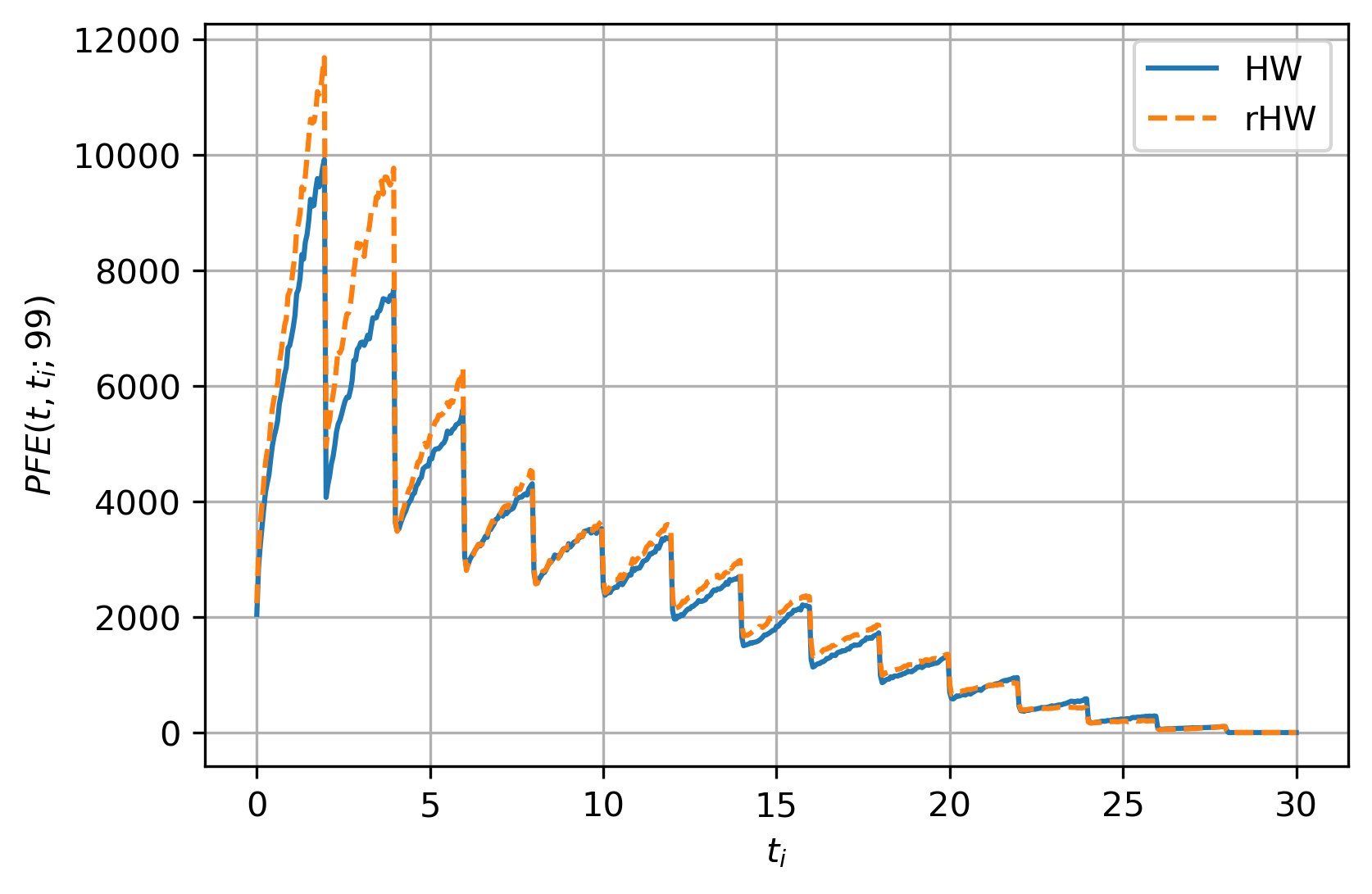}
    \caption{$\PFE(t, t_i; 97.5)$}
    \label{fig:RAnDEURNormalCotsmilesReceiverBermSwaptionATMExposureExpiry0PFE(99)}
  \end{subfigure}
  \caption{Comparing HW and rHW (potential future) exposures for a receiver Bermudan swaption on an ATM swap ($\strike = \strikeATM$).
  HW $\CVA(t_0)=230.486$ and rHW $\CVA(t_0)=259.783$.
  }
  \label{fig:RAnDEURNormalCotsmilesReceiverBermSwaptionATMExposureExpiry0}
\end{figure}

\end{document}